\documentclass[11pt,reqno,english]{amsart}
\usepackage{amsfonts,amsmath,latexsym,verbatim,amscd,mathrsfs,color,array}
\usepackage[normalem]{ulem}
\usepackage[colorlinks=true]{hyperref}
\usepackage[hmargin=2.5cm, vmargin=2.5cm]{geometry}	
\usepackage{amsmath,amssymb,amsthm,amsfonts,graphicx,color}
\usepackage{amssymb}
\usepackage{pdfsync}
\usepackage{epstopdf}

\usepackage{subcaption}

\usepackage{amsmath,amssymb,amsthm,amsfonts,graphicx,color}
\usepackage{amssymb}

\usepackage{amsfonts,amsthm,amsmath,latexsym,bm,graphics}

\def\theequation{\thesection.\@arabic\c@equation}
\makeatother

\renewcommand{\theequation}{\thesection.\arabic{equation}}
\newtheorem{lemma}{Lemma}[section]

\newtheorem{proposition}{Proposition}[section]

\newtheorem{theorem}{Theorem}[section]

\newcommand{\sech}{\text{sech}}

\numberwithin{equation}{section}

\begin{document}
	
\title[]{Multi-Spike Patterns in the Gierer-Meinhardt System with a Non-Zero Activator Boundary Flux}

\thanks{J. Wei is partially supported by NSERC of Canada. D. Gomez is supported by NSERC of Canada.}

\author[D. Gomez]{Daniel Gomez}
\address{Department of Mathematics,
	The University of British Columbia,
	Vancouver, BC
	Canada V6T 1Z2}
\email{dagubc@math.ubc.ca}

\author[J. Wei]{Juncheng Wei}
\address{Department of Mathematics,
	The University of British Columbia,
	Vancouver, BC
	Canada V6T 1Z2}
\email{jcwei@math.ubc.ca}
	
\begin{abstract}
The structure, linear stability, and dynamics of localized solutions to singularly perturbed reaction-diffusion equations has been the focus of numerous rigorous, asymptotic, and numerical studies in the last few decades. However, with a few exceptions, these studies have often assumed homogeneous boundary conditions. Motivated by the recent focus on the analysis of bulk-surface coupled problems we consider the effect of inhomogeneous Neumann boundary conditions for the activator in the singularly perturbed one-dimensional Gierer-Meinhardt reaction-diffusion system. We show that these boundary conditions necessitate the formation of spikes that concentrate in a boundary layer near the domain boundaries. Using the method of matched asymptotic expansions we construct boundary layer spikes and derive a new class of shifted Nonlocal Eigenvalue Problems for which we rigorously prove partial stability results. Moreover by using a combination of asymptotic, rigorous, and numerical methods we investigate the structure and linear stability of selected one- and two-spike patterns. In particular we find that inhomogeneous Neumann boundary conditions increase both the range of parameter values over which asymmetric two-spike patterns exist and are stable.
\end{abstract}

\date{\today}
\keywords{Gierer-Meinhardt system, singular perturbation, matched asymptotic , nonlocal eigenvalue problem (NLEP), spiky solution}

\maketitle

\section{Introduction}

Originally proposed by Gierer and Meinhardt in 1972, the Gierer-Meinhardt (GM) model illustrates the pattern forming potential of two mechanisms: local self-activation and lateral inhibition \cite{gierer_1972}. In the original GM model two diffusing chemical species, an \textit{activator} and \textit{inhibitor}, interact through prescribed reaction-kinetics. Specifically, letting $u$ and $v$ denote the concentrations of the activator and inhibitor respectively, the GM model takes the form of the two-component reaction-diffusion-system
\begin{subequations}
\begin{equation}\label{eq:gierer-meinhardt}
u_t = D_u \Delta u - u + u^pv^{-q},\quad v_t = D_v \Delta v - v + u^r v^{-s},\qquad x\in\Omega,
\end{equation}
where it is assumed that the exponents $(p,q,m,s)$ satisfy
\begin{equation}\label{eq:gm_exponents}
p>1, \quad q>0, \quad r>0, \quad s\geq 0,\quad \text{and}\quad p-1 - (s+1)^{-1}qr < 0,
\end{equation}
\end{subequations}
and homogeneous Neumann boundary conditions are typically imposed. As is evident from the choice of reaction kinetics, the GM model describes a process of self activation, or auto-catalysis, by the activator and lateral inhibition by the inhibitor. In addition, the spatially homogeneous steady state $(u,v)=(1,1)$ undergoes a spontaneous symmetry breaking bifurcation into a patterned state through a Turing instability when $D_u/D_v$ is sufficiently small. While one of the early applications of the GM model was specifically to  head generation and regeneration in Hydra \cite{gierer_1972}, self-activation and lateral inhibition is considered to be an important pattern formation mechanism in general biological systems and this has been supported by the identification of several molecular candidates for the activator and inhibitor in activator-inhibitor systems \cite{meinhardt_2000}. 

In the singularly perturbed limit of an asymptotically small diffusivity ratio $D_u/D_v=\mathcal{O}(\varepsilon)\ll 1$ a classical Turing stability analysis of \eqref{eq:gierer-meinhardt} reveals an asymptotically large band of unstable spatial modes. In this singularly perturbed limit the GM system, as well as numerous other two-component reaction-diffusion systems, are known to exhibits multi-spike patterns in which the activator concentrates at a discrete collection of asymptotically small regions. The rich structural and stability properties of these multi-spike patterns have been the focus of extensive asymptotic, rigorous, and numerical studies over the past two decades \cite{wei_1999,iron_2001,doelman_2001_gm,wei_2014_book,ward_2018}. While boundary conditions have been identified as playing an important role in pattern formation \cite{maini_1997,dillon_1994}, relatively few studies have investigated the role of boundary conditions on the structure and stability of multi-spike solutions to singularly perturbed reaction diffusion systems. Instead most such studies have assumed either homogeneous Neumann or homogeneous Dirichlet boundary conditions. Notable exceptions include the investigation of homogeneous Robin boundary conditions for the activator in the GM model \cite{berestycki_2003,maini_2007} and inhomogeneous Robin boundary conditions for the inhibitor in the two-dimensional Brusselator model \cite{tzou_2018_brusselator}. These two studies and their illustration of the effect of boundary conditions on the structure and stability of multi-spike patterns serve as the primary motivation for the present paper in which we consider inhomogeneous Neumann boundary conditions for the activator in the one-dimensional singularly perturbed GM model. Additionally, this paper aims to address some of the technical issues that arise in bulk-surface coupled reaction diffusion systems for which inhomogeneous boundary conditions naturally arise \cite{levine_2005,ratz_2014,madzvamuse_2015,gomez_2019}.

To simplify the presentation in this paper we consider the singularly perturbed one-dimensional GM model with exponents $(p,q,r,s)=(2,1,2,0)$ though we remark that different exponents result only in quantitative rather than qualitative differences. After possibly rescaling the spatial variable in \eqref{eq:gierer-meinhardt} we assume that the domain is the unit interval $\Omega=(0,1)$ and that $D_u=\varepsilon^2$ while $D_v=\mathcal{O}(1)$ where $\varepsilon\ll 1$ is an asymptotically small parameter. The activator in an equilibrium solution will then concentrate in intervals of $\mathcal{O}(\varepsilon)$ length and by integrating the inhibitor equation in \eqref{eq:gierer-meinhardt} it is easy to see that if $v=\mathcal{O}(1)$ then we must have $u=\mathcal{O}(\varepsilon^{-1/2})$ in each interval on which it is concentrated. This motivates our choice of rescaling $u=\varepsilon^{-1}\tilde{u}$ and $v=\varepsilon^{-1}\tilde{v}$ which when substituted into \eqref{eq:gierer-meinhardt} and dropping the tildes gives
\begin{subequations}\label{eq:pde_gm}
\begin{align}
u_t      & = \varepsilon^2 u_{xx} - u + u^2 v^{-1}, & 0<x<1, \label{eq:pde_gm_u}\\
\tau v_t & = D v_{xx} - v + \varepsilon^{-1}u^2, & 0<x<1, \label{eq:pde_gm_v}
\end{align}
and for which we observe both $u$ and $v$ will be $\mathcal{O}(1)$ in each interval where $u$ is concentrated. In addition we impose inhomogeneous and homogeneous Neumann boundary conditions for the activator and inhibitor respectively which are given by
\begin{equation}\label{eq:pde_gm_bc}
-\varepsilon u_x(0) = A,\quad \varepsilon u_x(1) = B,\quad v_x(0) = 0,\quad v_x(1) = 0,
\end{equation}
\end{subequations}
where we assume that $A,B\geq 0$ and for which we note that the scaling for the activator boundary conditions arises naturally from the scaling argument. 

\begin{figure}[t!]
	\centering
	\begin{subfigure}{0.5\textwidth}
		\centering
		\includegraphics[scale=0.7]{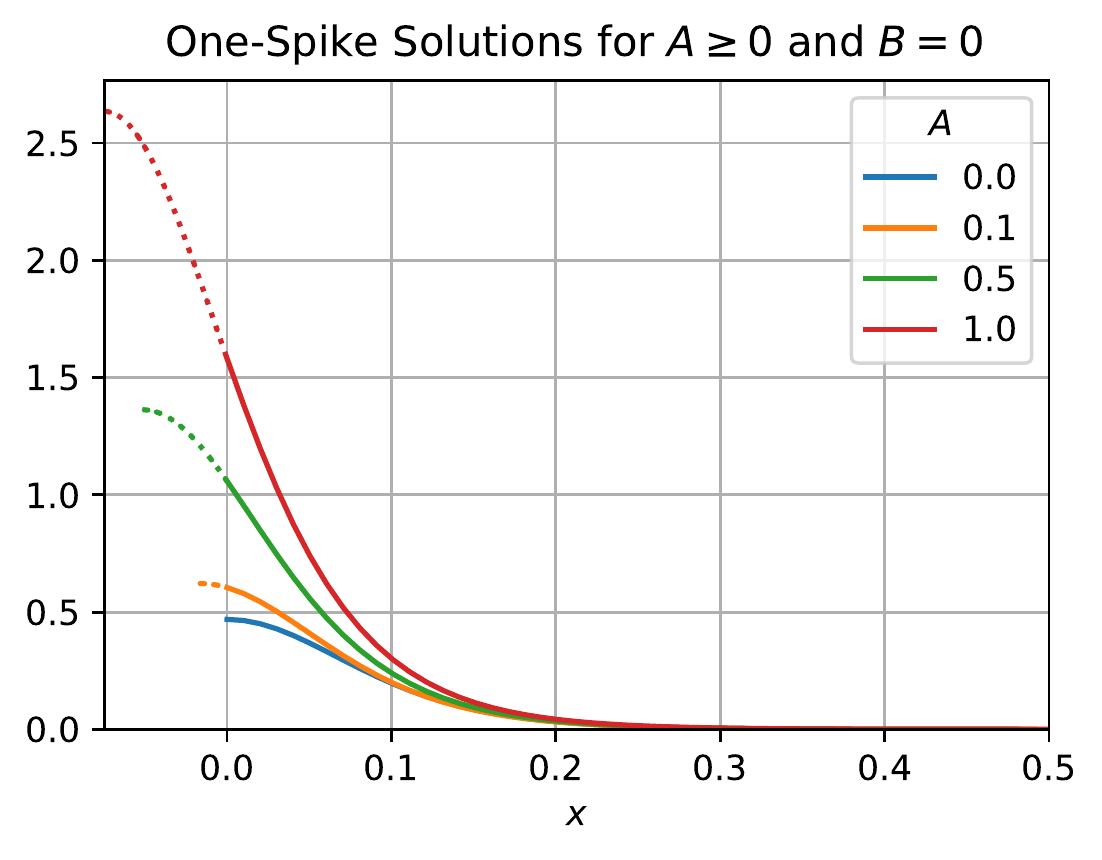}
		\caption{}\label{fig:shifted-spike}
	\end{subfigure}%
	\begin{subfigure}{0.5\textwidth}
		\centering
		\includegraphics[scale=0.6]{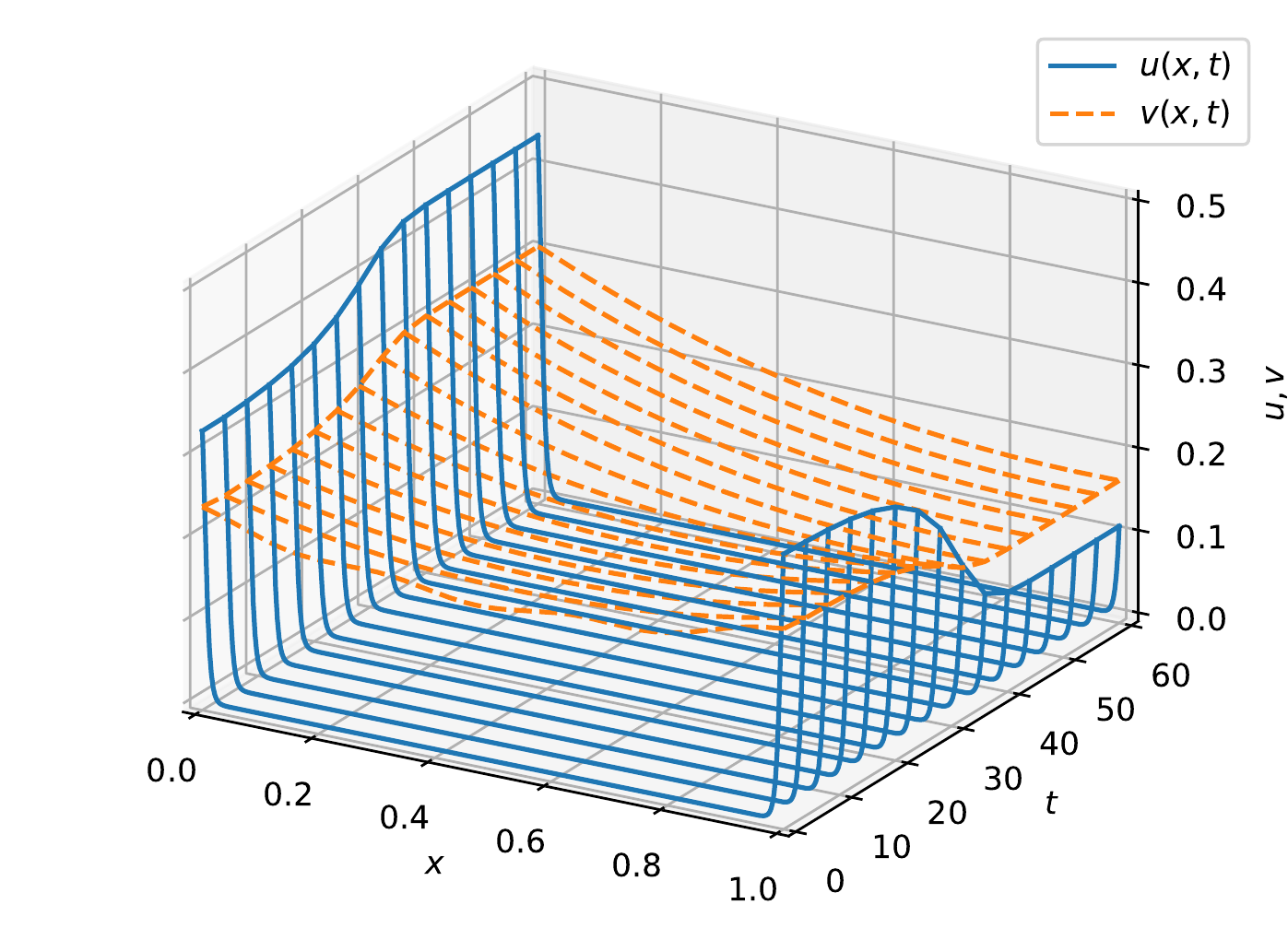}
		\caption{}\label{fig:competition-time}
	\end{subfigure}%
	\caption{(A) Examples of shifted one-spike solution concentrated at $x=0$ for various values of $A\geq 0$ and with $\varepsilon=0.05$ and $D=5$. (B) Evolution of solution to GM problem with $D=0.6$, $\varepsilon=.005$, $\tau=0.1$, and $A=B=0.08$. The initial condition is an unstable two-spike equilibrium where both spikes concentrate at the boundaries. A competition instability predicted by our asymptotic results in Figure \ref{fig:BOT_BB_BA_colormap} is triggered and leads to the solution settling at an asymmetric pattern.}\label{fig:intro-figs}
\end{figure}

In contrast to systems with homogeneous Neumann boundary conditions, we note that \eqref{eq:pde_gm} does \textit{not} have a spatially homogeneous steady state when $A>0$ and/or $B>0$ and therefore traditional Turing stability analysis methods no longer apply. In particular the inhomogeneous Neumann boundary conditions for the activator in \eqref{eq:pde_gm_bc} necessitate that the activator forms a boundary layer near each boundary when $A>0$ and/or $B=0$. As we demonstrate in \S\ref{sec:quasi_equilibrium} the appropriate boundary layer solution takes the form of a \textit{shifted} spike (see Figure \ref{fig:shifted-spike}) that is analogous to the near-boundary spike solution of \cite{berestycki_2003,maini_2007} though it has different stability properties which we investigate in \S\ref{sec:rigorous}. Furthermore, by considering examples of one- and two-spike patterns in \S\ref{sec:examples} we investigate the role of non-zero boundary fluxes on the structure of symmetric and asymmetric patterns, as well as their stability with respect to Hopf (example 1), competition (examples 2-4), and drift (example 4) instabilities. In Figure \ref{fig:competition-time} we plot the time-evolution of a linearly unstable equilibrium consisting of two boundary spikes of equal height. We see that the solution undergoes a competition instability, but rather than being subcritical as is the case when $A=B=0$ \cite{ward_2002_asymmetric}, the non-zero boundary fluxes force the solution to settle to an asymmetric pattern. This illustrates that one of the key features of introducing non-zero boundary fluxes is that it leads to a kind of robustness of asymmetric solutions similar to that observed in the presence of an inhomogeneous precursor gradient \cite{kolokolnikov_2020}. 

The remainder of this paper is organized as follows. In \S\ref{sec:quasi_equilibrium} we use the method of matched asymptotic expansions to construct multi-spike equilibrium solutions. The key idea of the construction is to leverage the localized of the spike solution to reduce their construction to a problem of finding the spike heights and their locations. In \S\ref{sec:stability} we derive a nonlocal eigenvalue problem which determines the linear stability of the multi-spike solutions on an $\mathcal{O}(1)$ time scale. Then, in \S\ref{sec:rigorous} we rigorously prove partial stability results for an equilibrium consisting of a single boundary spike. This is done by analyzing a class of \textit{shifted} nonlocal eigenvalue problems analogous to those studied in \cite{maini_2007}. In \S\ref{sec:examples} we consider four examples for which we construct one- and two-spike equilibrium patterns and study their linear stability and dynamics. Finally in \S\ref{sec:discussion} we conclude with a summary of our results and highlight several open problems and suggestions for future research.
	
\section{Quasi-Equilibrium Multi-Spike Solutions and their Slow Dynamics}\label{sec:quasi_equilibrium}

In this section we use the method of matched asymptotic expansions to derive an algebraic system and an ordinary differential equation that determine the profile and slow dynamics of a multi-spike quasi-equilibrium solution to \eqref{eq:pde_gm}. The derivation uses techniques that are now common in the study of localized patterns in one dimension. Our presentation will therefore be brief, highlighting only the novel aspects introduced by the inhomogeneous Neumann boundary conditions for the activator.  We begin by supposing that there are two spikes concentrated at the boundaries $x_L=0$ and $x_R=1$ as well as $N$ spikes concentrated in the interior at $0<x_1<...<x_N<1$. In addition we assume that the spikes are well separated in the sense that $|x_i-x_j|=\mathcal{O}(1)$ as $\varepsilon\rightarrow 0^+$ for all $i\neq j\in\{L,1,...,N,R\}$. This last assumption is key for effectively applying the method of matched asymptotic expansions.

We first construct an asymptotic approximation for the solution near $x=0$ by letting $x = \varepsilon y$ where $y=\mathcal{O}(1)$ and expanding
\begin{equation}
u \sim u_{L,0}(y) + \mathcal{O}(\varepsilon),\qquad v \sim v_{L,0}(y) + \varepsilon v_{L,1}(y) + \mathcal{O}(\varepsilon^2).
\end{equation}
It is easy to see that $v_L=\xi_L$ where $\xi_L$ is is an undetermined constant, and $u_{L,0}(y) = \xi_L w_c(y+y_L)$ where $w_c(y)$ is the unique \textit{homoclinic} solution to
\begin{subequations}
\begin{equation}\label{eq:wc_equation}
w_c'' - w_c + w_c^2 = 0,\quad 0<y<\infty,\qquad w_c'(0) = 0, \qquad  w_c(y)\rightarrow 0 \quad\text{as } y\rightarrow \infty,
\end{equation}
given explicitly by
\begin{equation}\label{eq:wc_def}
w_c(y) = \frac{3}{2}\sech^2\frac{y}{2}.
\end{equation}
\end{subequations}
Moreover, the undetermined \textit{shift} parameter $y_L$ is chosen to satisfy the inhomogeneous Neumann boundary condition
\begin{equation}\label{eq:yL_eq_1}
w_c'(y_L) = -A/\xi_L.
\end{equation}
The unknown constants $\xi_L$ and $y_L$ are found by matching with the outer solution. To determine the appropriate \textit{Neumann} boundary conditions for the outer problem we must first calculate $v_{L,1}'(y)$ as $y\rightarrow\infty$. This is done by integrating the $\mathcal{O}(\varepsilon)$ equation
\begin{equation}
D v_{L,1}'' = - \xi_L^2 w_c(y+y_L)^2,\quad 0<y<\infty;\qquad v_{L,1}'(0) = 0.
\end{equation}
over $0<y<\infty$ to obtain the limit
\begin{equation}\label{eq:v_lim_left_bdry}
\lim_{y\rightarrow +\infty} v_{L,1}'(y) = -\frac{\xi_L^2}{D}\eta(y_L),
\end{equation}
where
\begin{equation}\label{eq:eta_def}
\eta(y_0) \equiv \int_{0}^\infty w_c(y+y_0)^2 dy = \frac{6e^{-2y_0}(3+e^{-y_0})}{(1+e^{-y_0})^3}
\end{equation}
Note that $\eta(0)=3$ and $\eta\rightarrow 0^+$ monotonically as $z\rightarrow\infty$. In a similar way we obtain the inner solution near $x=1$ by letting $x=1-\varepsilon y$ and finding that
$$
u \sim \xi_{R}w_c(y+y_R) + \mathcal{O}(\varepsilon),\qquad v \sim \xi_{R} + \varepsilon v_{R,1}(y) + \mathcal{O}(\varepsilon^2),
$$
where $y_{R}$ is determined by solving
\begin{equation}\label{eq:yR_eq_1}
w_c'(y_{R}) = -B/\xi_{R},
\end{equation}
and for which we calculate the limit
\begin{equation}\label{eq:v_lim_right_bdry}
\lim_{y\rightarrow +\infty} v_{R,1}'(y) = -\frac{\xi_{R}^2}{D}\eta(y_R).
\end{equation}

We now consider the inner solution at each interior spike location. By balancing dominant terms in a higher order asymptotic expansion, it can be shown that the interior spike locations var on an $\mathcal{O}(\varepsilon^{-2})$ timescale. Therefore we let $x_i = x_i(\varepsilon^2 t)$ for each $i=1,...,N$ and with $x = x_i(\varepsilon^2 t)+y$ we calculate the inner asymptotic expansions
\begin{equation}
u \sim \xi_i w_c(y) + \mathcal{O}(\varepsilon),\qquad v\sim \xi_{i} + \varepsilon v_{i1}(y) + \mathcal{O}(\varepsilon^2),\qquad i=1,...,N.
\end{equation}
Furthermore, we must impose a solvability condition on the $v_{i1}$ problem which gives
\begin{equation}\label{eq:slow_dynamics_inner}
\frac{1}{\varepsilon^2}\frac{d x_i}{dt} = -\frac{1}{\xi_i}\biggl( \lim_{y\rightarrow +\infty} v_{i1}'(y) + \lim_{y\rightarrow-\infty}v_{i1}'(y)\biggr).
\end{equation}

To determine the $2(N+2)$ undetermined constants $\xi_i$ and $y_i$ where $i\in\{L,1,...,N,R\}$ we must now calculate the outer solution, defined for $|x-x_i|=\mathcal{O}(1)$ for each boundary and interior spike location, and match it with each of the inner solutions. Since $w_c$ decays to zero exponentially as $y\rightarrow\pm\infty$ we determine that the activator is asymptotically small in the outer region. On the other hand \eqref{eq:v_lim_left_bdry} and \eqref{eq:v_lim_right_bdry} imply the boundary conditions
\begin{equation*}
v_x(0) \sim -\frac{\xi_L^2}{D}\eta(y_L),\qquad v_x(1) \sim \frac{\xi_R^2}{D}\eta(y_R),
\end{equation*}
while the exponential decay of $w_c$ implies that the following limits hold (in the sense of distributions)
\begin{equation*}
\varepsilon^{-1} u^2 \longrightarrow 6 \sum_{j=1}^N \xi_j^2\delta(x-x_j),\qquad\qquad\qquad (\varepsilon\rightarrow 0^+),
\end{equation*}
Thus, to leading order in $\varepsilon\ll 1$, the outer problem for the inhibitor is given by
\begin{align}
& D v_{xx} - v = -6 \sum_{j=1}^N \xi_j^2 \delta(x-x_j),\qquad 0<x<1,\\
& D v_{x}(0) = -\xi_0^2\eta(y_L),\qquad D v_x(1) = \xi_{N+1}^2\eta(y_R).
\end{align}
This boundary value problem can be solved explicitly by letting $G_\omega$ be the Green's function satisfying
\begin{subequations}\label{eq:greens}
\begin{equation}\label{eq:greens_eq}
G_{\omega,xx} - \omega^2 G_\omega = -\delta(x-\xi),\quad 0<x<1;\qquad G_{\omega,x}(0,\xi) = 0,\quad G_{\omega,x}(1,\xi) = 0,\qquad \omega>0
\end{equation}
and given explicitly by
\begin{equation}\label{eq:greens_func}
G_\omega(x,\xi) = \frac{1}{\omega\sinh\omega}\begin{cases} \cosh\omega x\cosh\omega(1-\xi), & 0<x<\xi,\\ \cosh\omega(1-x)\cosh\omega\xi, & \xi < x <1.\end{cases}
\end{equation}
Formally substituting $\xi=0$ or $\xi=1$ into the above expression gives
\begin{equation}
G_\omega(x,0) = \frac{\cosh\omega(1-x)}{\omega \sinh\omega},\qquad G_\omega(x,1) = \frac{\cosh\omega x}{\omega\sinh\omega},
\end{equation}
\end{subequations}
which is readily seen to satisfy
\begin{equation*}
G_{\omega,xx}-\omega^2 G_\omega=0,\qquad 0<x<1,
\end{equation*}
with boundary conditions
\begin{equation*}
G_{\omega,x}(0,0) = -1,\qquad G_{\omega,x}(1,0) = 0,\qquad G_{\omega,x}(0,1) = 0, \qquad G_{\omega,x}(1,1) = 1.
\end{equation*}
Letting $\omega_0 \equiv D^{-1/2}$ we obtain the following leading order asymptotic expansion for the quasi-equilibrium solution to \eqref{eq:pde_gm}
\begin{subequations}\label{eq:quasi_equilibrium}
\begin{align}
& u_e(x) \sim \xi_L w_c\bigl(\tfrac{x}{\varepsilon} + y_L\bigr) + \xi_R w_c\bigl( \tfrac{1-x}{\varepsilon} + y_R\bigr) + \sum_{j=1}^N \xi_{j} w_c\bigl(\tfrac{x-x_j}{\varepsilon}\bigr), \label{eq:quasi_u}\\
& v_e(x) \sim \omega_0^2 \biggl(\xi_L^2 \eta(y_L) G_{\omega_0}(x,0) + \xi_R^2\eta(y_R) G_{\omega_0}(x,1) + 6 \sum_{j=1}^N \xi_j^2 G_{\omega_0}(x,x_j)\biggr). \label{eq:quasi_v}
\end{align}
\end{subequations}
Furthermore, by imposing the \textit{consistency} condition $v_e(x_i) = \xi_i$ for each $i\in\{L,1,...,N,R\}$ we obtain the system of $N+2$ nonlinear equations

\begin{subequations}\label{eq:quasi_equations}
\begin{equation}\label{eq:sys_1}
\bm{B} \equiv \bm{\xi} - \omega_0^2 \mathscr{G}_{\omega_0} \mathscr{N} \bm{\xi}^2 = 0,
\end{equation}
where $\mathscr{G}_{\omega_0}$ and $\mathscr{N}$ are the $(N+2)\times(N+2)$ matrices given by
\begin{equation}
(\mathscr{G}_{\omega_0})_{ij} = G_{\omega_0}(x_i,x_j)\quad i,j = L,R,1,...,N,\qquad \mathscr{N}\equiv \text{diag}(\eta(y_L),\eta(y_R),6,...,6),
\end{equation}
and
\begin{equation}
\bm{\xi} \equiv (\xi_L,\xi_R,\xi_1,...,\xi_N)^T,\qquad \bm{\xi}^2 = (\xi_L^2,\xi_R^2,\xi_1^2,...,\xi_N^2)^T.
\end{equation}
\end{subequations}
Thus, for  given spike configuration $0<x_1<...<x_N<1$, the system \eqref{eq:sys_1} together with \eqref{eq:yL_eq_1} and \eqref{eq:yR_eq_1} can be solved for the unknown spike heights $\xi_L,\xi_1,...,\xi_{N},\xi_R$ and boundary \textit{shifts} $y_L$ and $y_R$. Summarizing, we have the following proposition.

\begin{proposition}\label{prop:quasi-equi}
In the limit $\varepsilon\rightarrow 0^+$ and for $t\ll \mathcal{O}(\varepsilon^{-2})$ an $N+2$ spike quasi-equilibrium solution  to \eqref{eq:pde_gm} consisting of two boundary spikes and $N$ well separated interior spikes concentrated at specified locations $0<x_1<...<x_N<1$ is given asymptotically by \eqref{eq:quasi_equilibrium} where $G_{\omega_0}(x,\xi)$ is given explicitly by \eqref{eq:greens_func} and $\omega_0=D^{-1/2}$. The boundary shifts, $y_L$ and $y_R$, and spike heights, $\xi_L,\xi_1,...,\xi_{N},\xi_R$, are found by solving the system of $N+4$ equations \eqref{eq:yL_eq_1}, \eqref{eq:yR_eq_1}, and \eqref{eq:sys_1}.
\end{proposition}

The asymptotic solution constructed in the above proposition will not generally be an equilibrium of \eqref{eq:pde_gm} due to the slow, $\mathcal{O}(\varepsilon^{-2})$, drift motion of the interior spikes described by \eqref{eq:slow_dynamics_inner}. However, this solution can be made into an equilibrium by choosing the interior spike locations $x_1,...,x_N$ appropriately.

\begin{proposition}\label{prop:dynamics}
The interior spike locations of a multi-spike pattern consisting of two boundary spikes and $N$ interior spikes vary on an $\mathcal{O}(\varepsilon^{-2})$ time scale according to the differential equation
\begin{equation}\label{eq:ode_dynamics}
\begin{split}
\frac{1}{\varepsilon^2}\frac{d x_i}{dt} = & -\frac{6 \xi_i}{D} \big\langle \partial_x G_{\omega_0}(x,x_i) \big\rangle_{x=x_i} - \frac{12}{\xi_i D} \sum_{j\neq i} \xi_j^2 G_x(\xi_i,\xi_j) \\
& -\frac{2}{\xi_i D}\biggl[\xi_L^2\eta(y_R) G_x(x_i,0) + \xi_R^2\eta(y_R) G_x(x_i,1)  \biggr],\qquad (i=1,...,N),
\end{split}
\end{equation}
where
\begin{equation}
\big\langle f(x) \big\rangle_{x_0} = \lim_{x\rightarrow x_0^+} f(x) + \lim_{x\rightarrow x_0^{-}}f(x),
\end{equation}
which is to be solved together with \eqref{eq:sys_1}, \eqref{eq:yL_eq_1}, and \eqref{eq:yR_eq_1} for the spike heights $\xi_L,\xi_1,...,\xi_N,\xi_R$ and shifts $y_L$ and $y_R$. In particular, if the configuration $x_1,...,x_N$ is stationary with respect to the ODE \eqref{eq:ode_dynamics}, then to leading order the quasi-equilibrium solution of Proposition \ref{prop:quasi-equi} is an equilibrium for all $t\geq 0$.
\end{proposition}

\subsection{Equilibrium Multi-Spike Solutions by the Gluing Method}\label{subsec:gluing-method}

We now use an alternative method for constructing \textit{asymmetric} multi-spike equilibrium solutions to \eqref{eq:pde_gm}. This method extends that of Ward and Wei \cite{ward_2002_asymmetric} to account for inhomogeneous Neumann boundary conditions. The key idea is to construct a single boundary spike solution in an interval of variable length and use this solution to \textit{glue} together a multi-spike solution. In particular, we begin considering the problem
\begin{subequations}\label{eq:pde_gm_interval}
\begin{align}
& \varepsilon^2 u_{xx} - u + v^{-1}u^2 = 0,\quad D v_{xx} - v + \varepsilon^{-1} u^2 = 0,\qquad 0<x<l, \\
& \varepsilon u_x(0) = -A,\quad \varepsilon u_x(l) = 0,\quad Dv_x(0) = 0,\quad D v_x(l) = 0,
\end{align}
\end{subequations}
where $l>0$ is fixed and for which we will use the method of matched asymptotic expansions to construct a single spike solution concentrated at $x=0$. Proceeding as in \S\ref{sec:quasi_equilibrium} we readily find that the equilibrium solution in the outer region (i.e. for $x=\mathcal{O}(1)$) is given by
\begin{equation}\label{eq:gluing_solution}
u(x;l,A) \sim \xi_0 w_c\bigl(\varepsilon^{-1}x+y_0\bigr),\qquad v(x;l,A) \sim \xi_0\frac{\cosh\omega_0 (l-x)}{\cosh\omega_0 l},
\end{equation}
where $\omega_0 \equiv D^{-1/2}$ while the shift parameter $y_0$ and spike height $\xi_0$ satisfy
\begin{equation}\label{eq:matching_equation}
w_c'(y_0)=-\frac{A}{\xi_0},
\end{equation}
and for which, using \eqref{eq:wc_def} and \eqref{eq:eta_def}, we explicitly calculate
\begin{equation}\label{eq:gluing_parameters}
\xi_0 = \frac{\tanh\omega_0 l}{\omega_0 \eta(y_0)},\quad y_0 = \log\biggl(\frac{1 + 3q + \sqrt{9q^2 + 10q + 1}}{2}\biggr),\quad q\equiv \frac{\omega_0 A}{\tanh\omega_0 l},
\end{equation}
for which we remark that $y_0\sim 4 q$ as $q\rightarrow 0$ and therefore $\xi_0\sim(3\omega_0)^{-1}\tanh\omega_0 l$ as $A\rightarrow 0^+$. Finally, we note that $y_0$ is monotone increasing in $A$ and monotone decreasing in $D$ and $l$ when $A>0$ is fixed.

A multi-spike pattern is constructed by first partitioning the unit interval $0<x<1$ into $N+2$ subintervals defined by
\begin{align*}
& x_L = 0, \qquad & x_i = l_L + 2\sum_{j=1}^i l_j + l_i \quad (i=1,...,N), \qquad   & x_R=1, & \\
& I_L = [0,l_L), \qquad & I_i = [x_i-l_i,x_i+l_i)\quad (i=1,...,N), \qquad & I_R = [1-l_R,1], &
\end{align*}
where $l_L,l_1,...,l_N,l_R$ are chosen to satisfy the $N+2$ constraints
\begin{subequations}\label{eq:gluing_equation}
\begin{align}
& l_L + 2l_1 + \cdots 2 l_N + l_R = 1. \label{eq:gluing_equation_1} \\
v(l_L;l_L,A) = v(&l_1;l_1,0) = \cdots = v(l_N;l_N,0) = v(l_{R};l_{R},B), \label{eq:gluing_equation_2}.
\end{align}
The first constraint guarantees that the intervals are mutually disjoint, while the second set of $N+1$ constraints guarantees the continuity of the multi-spike equilibrium solution
\begin{equation}
u_e(x) = \begin{cases} u(x;l_L,A), & x\in I_L \\ u(|x-x_i|;l_i,0), & x\in I_i \\ u(1-x,l_R,B), & x\in I_R \end{cases},\quad v_e(x) = \begin{cases} v(x;l_L,A), & x\in I_L \\ v(|x-x_i|;l_i,0), & x\in I_i \\ v(1-x,l_R,B), & x\in I_R \end{cases}.
\end{equation}
\end{subequations}
We remark that the local symmetry of each interior spike implies that the interior spikes are stationary with respect to the slow dynamics found in \eqref{eq:slow_dynamics_inner}, and therefore the multi-spike solution constructed above is an equilibrium of \eqref{eq:pde_gm}.

\section{Linear Stability of Multi-Spike Pattern}\label{sec:stability}

In this section we derive a nonlocal eigenvalue problem (NLEP) that, to leading order in $\varepsilon\ll 1$, determines the linear stability of the quasi-equilibrium solution given in Proposition \ref{prop:quasi-equi} on an $\mathcal{O}(1)$ timescale. Letting $u_e$ and $v_e$ be the quasi-equilibrium solution from Proposition \ref{prop:quasi-equi}, we consider the perturbations $u = u_e + e^{\lambda t}\Phi$ and $v = v_e + e^{\lambda t} \Psi$ with which \eqref{eq:pde_gm} becomes
\begin{subequations}\label{eq:pde_sys_stability}
\begin{align}
& \varepsilon^2 \Phi_{xx} - \Phi + 2\frac{u_e}{v_e}\Phi - \frac{u_e^2}{v_e^2}\Psi = \lambda \Phi, & 0<x<1, \label{eq:pde_sys_Phi}\\
& D \Psi_{xx} - \Psi + 2\varepsilon^{-1}u_e\Phi = \tau \lambda \Psi, & 0<x<1. \label{eq:pde_sys_Psi}
\end{align}
\end{subequations}
This problem admits both \textit{large} and \textit{small} eigenvalues characterized by $\lambda=\mathcal{O}(1)$ and $\mathcal{O}(\varepsilon^{2})$ respectively. The small eigenvalues are closely related to the linearization of the slow-dynamics \eqref{eq:ode_dynamics} and the resulting instabilities therefore take place over a $\mathcal{O}(\varepsilon^{-2})$ timescale \cite{wei_2007_existence}. In contrast, the large eigenvalues lead to amplitude instabilities over a $\mathcal{O}(1)$ timescale. In this section we focus exclusively on the large eigenvalues and limit our discussion of the small eigenvalues to the specific example given in \S\ref{sec:example_4} in which a two-spike solution consisting of one spike on the boundary and one interior spike is considered.

Using the method of matched asymptotic expansions as in \S\ref{sec:quasi_equilibrium} we readily find that, to leading order in $\varepsilon\ll 1$, the inhibitor perturbation $\Psi$ satisfies
\begin{subequations}\label{eq:outer-inhibitor-perturbation-equation}
\begin{align}
& D \Psi_{xx} - (1+\tau\lambda) \Psi = -2\sum_{j=1}^N \xi_j \int_{-\infty}^\infty w_c(y)\phi_j(y)dy \delta(x-x_j),\qquad 0<x<1,\\
& D\Psi_x(0) = -2\xi_L\int_0^\infty w_c(y+y_L)\phi_L(y) dy,\quad D\Psi_x(1) = 2 \xi_R\int_0^\infty w_c(y+y_R)\phi_R(y) dy,
\end{align}
\end{subequations}
where $\phi_L$ and $\phi_R$ are the leading order inner expansions of the activator perturbation $\Phi$ at the boundaries satisfying
\begin{subequations}\label{eq:inner-perturbation-equations}
\begin{equation}\label{eq:inner-boundary-perturbation-equation}
\mathscr{L}_{y_i}\phi_i - w_c(y+y_i)^2\Psi(x_i) = \lambda\phi_i,\quad 0<y<\infty,\quad \phi_i'(0)=0,\quad \phi_i\rightarrow 0\quad\text{as } y\rightarrow\infty,
\end{equation}
for $i=L,R$ respectively, while $\phi_1,...,\phi_N$ are the leading order inner expansions of $\Phi$ at each of the interior spike locations $x_1,..,x_N$ satisfying
\begin{equation}\label{eq:inner-interior-perturbation-equation}
\mathscr{L}_0 \phi_i - w_c(y)^2\Psi(x_i) = \lambda\phi_i,\quad -\infty < y < \infty,\quad \phi\rightarrow 0\quad\text{as } y\rightarrow\pm\infty,
\end{equation}
\end{subequations}
for each $i=1,...,N$ respectively. The linear differential operator $\mathscr{L}_{y_0}$ parametrized by $y_0\geq 0$ appearing in each equation is explicitly given by
\begin{equation}\label{eq:L_y0-definition}
\mathscr{L}_{y_0} \phi \equiv \phi'' - \phi + 2 w_c(y+y_0)\phi.
\end{equation}
Note that by decomposing each $\phi_i = \phi_i^\text{even} + \phi_i^\text{odd}$ ($i=1,...,N$) where $\phi_i^\text{even}$ and $\phi_i^\text{odd}$ are even and odd about $y=0$ respectively, we find that either $\phi_i^\text{odd} = 0$ or else $\lambda\leq 0$. In particular, the odd components of each $\phi_i$ ($i=1,...,N$) do not contribute to any instabilities and without loss of generality we may therefore assume that each $\phi_i$ is even about $y=0$. Hence it suffices to pose \eqref{eq:inner-interior-perturbation-equation} on the half line with the same homogeneous Neumann boundary conditions used in \eqref{eq:inner-boundary-perturbation-equation}.

Letting $\omega_\lambda \equiv \sqrt{(1+\tau\lambda)/D}$ and recalling the definition of $G_{\omega}$ in \eqref{eq:greens} we readily find that the solution to \eqref{eq:outer-inhibitor-perturbation-equation} is explicitly given by
\begin{equation*}
\Psi(x) = 2\omega_0^2\sum_{j=L,R,1}^{N} \hat{\xi}_j G_{\omega_\lambda}(x,x_j) \int_{0}^\infty w_c(y+y_j)\phi_j(y)dy.
\end{equation*}
where we let
\begin{equation}
y_1=...=y_N = 0,\qquad \hat{\xi}_i \equiv \begin{cases} \xi_i, & i=L,R,\\ 2\xi_i, & i=1,...,N,\end{cases}.
\end{equation}
Evaluating $\Psi(x)$ at each $x=x_i$ and substituting into \eqref{eq:inner-perturbation-equations}  yields the system of NLEPs
\begin{subequations}
\begin{align}\label{eq:nlep}
\mathscr{L}_{y_i} \phi_i - 2\omega_0^2  w_c(y+y_i)^2 & \sum_{j=L,R,1}^N \hat{\xi}_j G_{\omega_\lambda}(x_i,x_j)\int_0^\infty w_c(y+y_j)\phi_j(y)\, d y = \lambda \phi_i, \quad y>0\\
& \phi_i'(0) = 0,\qquad \phi_i \rightarrow 0\qquad\text{as } y\rightarrow + \infty.
\end{align}
\end{subequations}
for each $i=L,R,1,...,N$ where $\mathscr{L}_{y_i}$ is defined by \eqref{eq:L_y0-definition}.

The NLEP system \eqref{eq:nlep} has two key features that distinguish it from analogous NLEPs in singularly perturbed reaction diffusion systems \cite{iron_2001,ward_2002_asymmetric,wei_1999,wei_2007_existence} and are explored in the rigorous stability results of \S\ref{sec:rigorous} as well as in the specific examples of \S\ref{sec:examples}. First, it considers both boundary-bound and interior-bound spikes. As explored in Examples 2 to 4 this has immediate consequences for both the existence and stability of asymmetric patterns even in the zero-flux case where $A=B=0$. The second distinguishing feature of \eqref{eq:nlep} is the introduction of the shift parameters $y_L\geq 0$ and $y_R\geq 0$. We remark that an analogous \textit{negative} shift parameter has been examined in the context of near-boundary spike solutions for \textit{homogeneous} Robin boundary conditions \cite{berestycki_2003,maini_2007}. However, as highlighted in the stability results of \S\ref{sec:rigorous}, the positive shift parameter plays a key role in the stability properties of boundary-bound spikes. An important critical value of the shift parameter is the unique value $y_{0c}>0$ such that $w_c''(y_{0c})=0$ and which is explicitly given by
\begin{equation}\label{eq:y0c_def}
y_{0c} = \log(2+\sqrt{3}).
\end{equation}
In particular, it can be shown that if $y_0\lessgtr y_{0c}$ then $\mathscr{L}_{y_0}$ has an unstable and stable spectrum respectively (see Lemma \ref{lem:principal_eig_Ly0} below). Moreover the operator $\mathscr{L}_{y_{0c}}$ has a one-dimensional kernel spanned by $w_c'(y+y_{0c})$.

\subsection{Reduction of NLEP to an Algebraic System}\label{subsec:algebraic}

It is particularly useful to rewrite \eqref{eq:nlep} as an algebraic system as follows. Assuming that $\lambda$ is not an eigenvalue of $\mathscr{L}_{y_i}$ for all $i=L,R,1,...,N$ we let
\begin{equation}\label{eq:phi_i_solve}
\phi_i = c_i (\mathscr{L}_{y_i} - \lambda)^{-1} w_c(y+y_i)^2,\qquad i\in\{L,R,1,...,N\},
\end{equation}
where the coefficients $c_L,c_R,c_1,...,c_N$ are undetermined. Note that in \eqref{eq:phi_i_solve} the homogeneous Neumann boundary condition $\phi_i'(0)=0$ is assumed. In addition, note that if $\lambda = 0$ then \eqref{eq:phi_i_solve} is only valid if $y_L,y_R\neq y_{0c}$. 

Substituting into \eqref{eq:nlep} then yields the linear homogeneous system for $\bm{c}\equiv(c_L,c_R,c_1,...,c_N)^T$
\begin{equation}\label{eq:algebraic_sys}
\mathscr{G}_{\omega_\lambda} \mathscr{D}_\lambda \pmb{c} = (2\omega_0^2)^{-1}\pmb{c},
\end{equation}
where $\mathscr{G}_{\omega_\lambda}$ is the $(N+2)\times(N+2)$ matrix with entries
\begin{equation}\label{eq:G_lambda_def}
(\mathscr{G}_{\omega_\lambda})_{ij} = G_{\omega_\lambda}(x_i,x_j),\qquad (i,j=L,R,1,...,N),
\end{equation}
while $\mathscr{D}_\lambda$ is the diagonal $(N+2)\times(N+2)$ matrix given by 
\begin{equation}\label{eq:D_lambda_def_1}
(\mathscr{D}_\lambda)_{ij} = \frac{1}{\omega_0}\begin{cases} \eta(y_i) \xi_i \mathscr{F}_{y_L}(\lambda), & i=j=L,R,\\ 6 \xi_i\mathscr{F}_0(\lambda), & i=j=1,...,N,\\ 0, & i\neq j, \end{cases}
\end{equation}
where
\begin{equation}\label{eq:F_y0_def}
\mathscr{F}_{y_0}(\lambda) \equiv \frac{\int_0^\infty w_c(y+y_0)(\mathscr{L}_{y_0}-\lambda)^{-1}w_c(y+y_0)^2 dy}{\int_0^\infty w_c(y+y_0)^2 dy},
\end{equation}
and for which \eqref{eq:L_y0_inverses} and \eqref{eq:L_y0_integrals} below imply that for all $y_0\neq y_{0c}$
\begin{equation}\label{eq:Fy00}
\mathscr{F}_{y_0}(0) = 1 + \frac{w_c'(y_0)w_c(y_0)^2}{2 w_c''(y_0) \eta(y_0)}.
\end{equation}
Comparing \eqref{eq:algebraic_sys} and \eqref{eq:nlep}, it follows that $\lambda$ is an eigenvalue of \eqref{eq:nlep} if and only if $(2\omega_0^2)^{-1}$ is an eigenvalue of $\mathscr{G}_{\omega_\lambda}\mathscr{D}_\lambda$. In particular, when $\lambda$ is not an eigenvalue of $\mathscr{L}_{y_L}$, $\mathscr{L}_{y_R}$, and $\mathscr{L}_0$ then it is an eigenvalue of the NLEP \eqref{eq:nlep} if and only if it satisfies the algebraic equation
\begin{equation}\label{eq:algebraic_eq}
\det\bigl(\mathbb{I}_{N+2} - 2\omega_0^2 \mathscr{G}_{\omega_\lambda}\mathscr{D}_\lambda\bigr) = 0,
\end{equation}
where $\mathbb{I}_{N+2}$ is the $(N+2)\times(N+2)$ identity matrix. 

We conclude by noting that if either $y_L=y_{0c}$ and/or $y_R=y_{0c}$ then the algebraic reduction fails when searching for a zero eigenvalue $\lambda=0$ since $\mathscr{L}_{y_{0c}}$ is not invertible. However, in this case we can deduce an analogous system. In particular letting $\lambda=0$ and assuming that $y_L=y_{0c}$ and $y_R\neq y_{0c}$, we multiply the $i=L$ NLEP in \eqref{eq:nlep} by $w_c'(y+y_{0c})$ and integrate over $0<y<\infty$ to get
\begin{equation}
\xi_L \int_{0}^\infty w_c(y+y_L)\phi_L dy = - \frac{1}{G_{\omega_0}(0,0)} \sum_{j=R,1}^N \tilde{\xi}_j G_{\omega_0}(0,x_j)\int_0^\infty w_c(y+y_j)\phi_j dy.
\end{equation}
Proceeding as above we then deduce that the NLEP \eqref{eq:nlep} with $\lambda=0$ is then equivalent to the algebraic equation
\begin{equation}
\det (\mathbb{I}_{N+1} - 2\omega_0^2 \tilde{\mathscr{G}_{\omega_0}}\tilde{\mathscr{D}}) = 0,
\end{equation}
where $\tilde{\mathscr{G}_{\omega_0}}$ and $\tilde{\mathscr{D}}$ are the $(N+1)\times (N+1)$ matrices with entries
\begin{equation}
(\tilde{\mathscr{G}_{\omega_0}})_{ij} = \mathscr{G}_{ij} - \frac{1}{G_{\omega_0}(0,0)}G_{\omega_0}(x_i,0) G_{\omega_0}(0,x_j),\qquad (\tilde{\mathscr{D}})_{ij} = \mathscr{D}_{ij},\qquad i,j=R,1,...,N.
\end{equation}
The same approach can likewise be used if $y_L=y_R=y_{0c}$.

\subsection{Zero-Eigenvalues of the NLEP and the Consistency Condition}

The conditions under which $\lambda=0$ is an eigenvalue of the NLEP \eqref{eq:nlep} can be directly linked to the system \eqref{eq:sys_1} as highlighted in \cite{wei_2007_existence}. Specifically, assume that $x_1,...,x_N$ are fixed (not necessarily at an equilibrium configuration of the slow dynamics ODE \eqref{eq:ode_dynamics}) and let $\xi_L,\xi_R,\xi_1,...,\xi_N$ together with $y_L$ and $y_R$ solve \eqref{eq:sys_1}, \eqref{eq:yL_eq_1}, and \eqref{eq:yR_eq_1} with the additional assumption that $y_L,y_R\neq y_{0c}$. From the definition of $\eta$ in \eqref{eq:eta_def} and from \eqref{eq:yL_eq_1} and \eqref{eq:yR_eq_1} we calculate
\begin{equation}\label{eq:eta_derivatives}
\frac{\partial \eta(y_i)}{\partial \xi_i} = \frac{w_c(y_i)^2 w_c'(y_i)}{\xi_i w_c''(y_i)},
\end{equation}
for $i=L,R$. Taking the Jacobian of the quasi-equilibrium system \eqref{eq:sys_1} and recalling the definition of $\mathscr{D}_\lambda$ given in \eqref{eq:D_lambda_def_1} we deduce that
\begin{equation}
\nabla_{\bm{\xi}}\bm{B} = \mathbb{I} - 2\omega_0^2 \mathscr{G}_{\omega_0}\mathscr{D}_0.
\end{equation}
Together with the discussion of \S\ref{subsec:algebraic} we deduce that if $y_L,y_R\neq 0$ and each $x_1,...,x_N$ is independent of $\xi_L,\xi_R,\xi_1,...,\xi_N$, then $\lambda=0$ is an eigenvalue of the NLEP \eqref{eq:nlep} if and only if the Jacobian $\nabla_{\bm{\xi}} \bm{B}$ is singular.

\section{Rigorous Stability and Instability Results for the Shifted NLEP}\label{sec:rigorous}

In this section we rigorously prove instability and stability results for the \textit{shifted} NLEP
\begin{equation}\label{eq:rigorous_nlep}
\mathscr{L}_{y_0}\phi - \mu \frac{\int_0^\infty w\phi}{\int_0^\infty w^2} w^2 = \lambda\phi,\quad 0<y<\infty;\qquad \phi'(0) = 0;\qquad \phi\rightarrow 0\quad \text{as }y\rightarrow\infty,
\end{equation}
where $\mu$ is a real constant and for a fixed value of $y_0\geq 0$ we define
\begin{equation}\label{eq:rigorous_L_definition}
\mathscr{L}_{y_0} \phi \equiv \phi'' - \phi + 2 w \phi,\qquad w(y)\equiv w_c(y+y_0),
\end{equation}
and where $w_c$ is the unique solution to \eqref{eq:wc_equation}. When $y_0=0$ the NLEP \eqref{eq:rigorous_nlep} is stable if $\mu > 1$ and unstable if $\mu<1$ \cite{wei_1999}. We begin by collecting a few facts about the operator $\mathscr{L}_{y_0}$ and its spectrum. First, we calculate
\begin{equation}\label{eq:L_y0_inverses}
\mathscr{L}_{y_0}^{-1} w^2 = w - \frac{w'(0)}{w''(0)}w',\qquad \mathscr{L}_{y_0}^{-1} w =  w + \frac{1}{2}yw^{'} - \frac{3 w'(0)}{2 w''(0)}w^{'}
\end{equation}
where the additional terms are chosen so that homogeneous Neumann boundary conditions at $y=0$ are satisfied and which we use to compute
\begin{subequations}\label{eq:L_y0_integrals}
	\begin{align}
	& \int_0^\infty w \mathscr{L}_{y_0}^{-1} w^2   = \int_0^\infty w^2 + \frac{w'(0) w(0)^2}{2 w''(0)},\quad  & \int_0^\infty w\mathscr{L}_{y_0}^{-1} w = \frac{3}{4}\int_0^\infty w^2 + \frac{3 w'(0) w(0)^2}{4 w''(0)},& \\
	& \int_0^\infty w^2 \mathscr{L}_{y_0}^{-1} w^2 = \int_0^\infty w^3 + \frac{w'(0)w(0)^3}{3 w''(0)},\quad & \int_0^\infty w^3 = \frac{6}{5}\int_0^\infty w^2 + \frac{3 w(0) w'(0)}{5}.&
	\end{align}
\end{subequations}
In the next two lemmas, we describe some key properties of the eigenvalue problem
\begin{equation}\label{eq:Ly0_eigenvalue_problem}
\mathscr{L}_{y_0}\Phi = \Lambda \Phi,\qquad 0<y<\infty;\qquad \Phi'(0)=0;\qquad \Phi\rightarrow 0,\quad\text{as }y\rightarrow +\infty.
\end{equation}

\begin{lemma}\label{lem:principal_eig_Ly0}
	Let $y_0\geq 0$ an let $\Lambda_0$ be the principal eigenvalue of \eqref{eq:Ly0_eigenvalue_problem}. Then $\Lambda_0=0$ if $y_0=y_{0c}$ and $\Lambda_0 \lessgtr 0$ if $y_0\gtrless y_{0c}$. Furthermore, the eigenfunction corresponding to the principal eigenvalue is of one sign. 
\end{lemma}
\begin{proof}
	Since $\mathscr{L}_{y_0}$ is self-adjoint, the variational characterization
	\begin{equation}
	-\Lambda_0 = \inf_{\Phi\in H^2([0,\infty))}\frac{\int_0^\infty |\Phi'|^2 + |\Phi|^2 - 2w|\Phi|^2}{\int_0^\infty |\Phi|^2},
	\end{equation}
	implies that the principal eigenfunction $\Phi_0$ is of one sign. Since $\Phi_0(0)\neq 0$ we may, without loss of generality, assume that $\Phi_0(0)=1$ and $\Phi_0>0$. Now we multiply \eqref{eq:Ly0_eigenvalue_problem} by $w'$ and integrate by parts to get
	\begin{equation}
	\Lambda_0 = \frac{w''(0)}{\int_0^\infty w'\Phi_0 dy},
	\end{equation}
	where we remark that the denominator is negative since $w'\leq 0$ for all $y\geq 0$. The claim follows by noting that $w''(0)=0$ when $y_0=0$ and $w''(0)\gtrless0$ for $y_0\gtrless y_{0c}$.
\end{proof}

\begin{lemma}\label{lem:second_eig_Ly0} Let $\Lambda_1$ be the second eigenvalue of $\mathscr{L}_{y_0}$. Then $\Lambda_1<0$ for all $y_0\geq 0$.
\end{lemma}

\begin{proof}
	First note that the second eigenfunction $\Phi_1$ must cross zero at least once since $\int_0^\infty \Phi_0\Phi_1dy = 0$ and $\Phi_0$ is of one sign. Next we assume toward a contradiction that $\Lambda_1 \geq 0$. We begin by showing that $\Phi_1$ has exactly one zero in $0<y<\infty$. Assume that $\Phi_1$ has more than one zero and choose $0<a<b<\infty$ such that $\Phi_1(a)=\Phi_1(b)=0$ and $\Phi_1>0$ in $a<y<b$. Then $\Phi_1'(a)>0$ and $\Phi_1'(b)<0$ so we obtain the contradiction
	\begin{equation}\label{eq:temp_Lambda_1_proof}
	0 \geq \Lambda_1 \int_a^b w' \Phi_1 dy = \int_a^b w'\mathscr{L}_{y_0}\Phi_1 dy = w'(b)\Phi_1'(b) - w'(a)\Phi_1'(a)>0,
	\end{equation}
	where we have used $\mathscr{L}_{y_0}w'=0$ and $w'<0$ for all $y>0$. Thus $\Phi_1$ has a unique zero $0<a<\infty$ and we may assume that $\Phi_1\lessgtr 0$ for $y\lessgtr a$. Setting $b=\infty$ in \eqref{eq:temp_Lambda_1_proof} we get a contradiction.
\end{proof}


In Figure \ref{fig:L0_eigenvalues} we plot the principal and second eigenvalues of the operator $\mathscr{L}_{y_0}$ which we calculated numerically (see Appendix \ref{app:conjecture} for details on the numerical method). 

Lemma \ref{lem:principal_eig_Ly0} implies that the NLEP will have different stability properties depending on whether $y_0$ is greater than or smaller than $y_{0c}$. We will henceforth refer to $0\leq y_0<y_{0c}$ and $y_0>y_{0c}$ as the \textit{small-shift} and \textit{large-shift} cases respectively. When $y_0 = 0$ it is known that for $\mu>0$ sufficiently large, the NLEP \eqref{eq:rigorous_nlep} is stable. In this sense the nonlocal term appearing in \eqref{eq:rigorous_nlep} can stabilize the spectrum of the linearized operator $\mathscr{L}_{y_0}$. Since all the eigenvalues of $\mathscr{L}_{y_0}$ are negative in the large-shift case we expect the spectrum of the NLEP \eqref{eq:rigorous_nlep} to remain stable for all $\mu\geq 0$. Restricting our attention to \textit{real} eigenvalues, we have the following stability result for the large-shift case.

\begin{theorem}\label{thm:large_shift_real_stable}
	All \textit{real} eigenvalues of the NLEP \eqref{eq:rigorous_nlep} are negative when $y_0 > y_{0c}$.
\end{theorem}

To prove this, we first prove the following lemma.

\begin{lemma}\label{lem:sign_lemma}
	Let $y_0 > y_{0c}$ and suppose that $\phi$ satisfies
	\begin{equation}
	\mathscr{L}_{y_0} \phi -\lambda\phi \geq 0,\quad 0<y<\infty;\qquad \phi'(0) \geq 0;\qquad \phi\rightarrow 0,\quad \text{as }y\rightarrow+\infty,
	\end{equation}
	where $\lambda \geq 0$. Then $\phi < 0$ for all $y\geq 0$.
\end{lemma}
\begin{proof}
	Assume toward a contradiction that $\phi > 0$ in $0\leq a < y < b\leq \infty$. Without loss of generality we may assume that $\phi(a)=0$ if $a>0$ and $\phi(0)>0$ if $a=0$. Then, for any such $0\leq a<b\leq \infty$ we have
	\begin{equation}
	\phi(a)\geq 0,\quad \phi'(a)\geq 0,\quad \phi(b)=0,\quad \phi'(b)\leq 0.
	\end{equation}
	Let $g(y) \equiv w''(y) - \beta w'(y)$ where $\beta \equiv \tfrac{\max_{y\geq 0}|w'''(y)|}{w''(a)}$ is well-defined and positive. Then $g>0$ for all $y\geq 0$, $g'(a)\leq 0$, and moreover $(\mathscr{L}_{y_0}-\lambda) g = \mathscr{L}_{y_0} w'' - \lambda g = -(w')^2 -\lambda g < 0$. Integrating by parts we obtain the contradiction
	\begin{equation}
	0 < \int_a^b g (\mathscr{L}_{y_0}-\lambda)\phi dy - \int_a^b \phi (\mathscr{L}_{y_0} - \lambda) g dy = g(b)\phi'(b) - g(a)\phi'(a) - g'(b)\phi(b) + g'(a)\phi(a) \leq 0.
	\end{equation}
\end{proof}

\begin{proof}[Proof [Theorem \ref{thm:large_shift_real_stable}]] Suppose that $\lambda \geq 0$ is an eigenvalue of \eqref{eq:rigorous_nlep} so that by Lemma \ref{lem:principal_eig_Ly0} the operator $\mathscr{L}_{y_0}-\lambda$ is invertible and from \eqref{eq:rigorous_nlep} we calculate
	\begin{equation*}
	\phi = \mu\frac{\int_0^\infty w\phi}{\int_0^\infty w^2} (\mathscr{L}_{y_0}-\lambda)^{-1} w^2.
	\end{equation*}
	But $w^2>0$ so by Lemma \ref{lem:sign_lemma} we obtain the contradiction
	\begin{equation}
	1 = \mu\frac{\int_0^\infty w (\mathscr{L}_{y_0} - \lambda)^{-1} w^2}{\int_0^\infty w^2} < 0.
	\end{equation}
\end{proof}

From Theorem \ref{thm:large_shift_real_stable} we immediately deduce that the NLEP does not admit a zero eigenvalue for any $\mu\geq 0$ when $y_0>y_{0c}$. On the other hand, when $0\leq y_0\leq y_{0c}$ we suspect that the NLEP admits a zero eigenvalue for an appropriate choice of $\mu\geq 0$. When $y_0 = y_{0c}$ this is the case for $\mu=0$. When $0\leq y_0<y_{0c}$ we set $\lambda=0$ in \eqref{eq:rigorous_nlep} and obtain
\begin{equation}\label{eq:muc_temp1}
\mathscr{L}_{y_0} \phi  = \mu\frac{\int_0^\infty w\phi}{\int_0^\infty w^2} w^2.
\end{equation}
Using \eqref{eq:L_y0_inverses} we calculate $\phi = \mathscr{L}_{y_0}^{-1}w^2$ and substitute back into \eqref{eq:muc_temp1} to deduce that $\lambda=0$ is an eigenvalue if and only if
\begin{equation}\label{eq:muc_def}
\mu = \mu_c(y_0) \equiv \frac{\int_0^\infty w^2}{\int_0^\infty w\mathscr{L}_{y_0}^{-1} w^2} = \frac{\int_0^\infty w^2}{\int_0^\infty w^2 + \frac{w'(0)w(0)^2}{2w''(0)}} .
\end{equation}
Note that $\mu_c(y_0)\lessgtr 0$ if $y_0 \gtrless y_{0c}$. In terms of the critical value $\mu_c$ we have the following instability result for the small-shift case.

\begin{theorem}\label{thm:unstbale_small_shift}
	Let $0\leq y_0 < y_{0c}$ and $0\leq \mu < \mu_c$ where critical value $\mu_c$ is defined in \eqref{eq:muc_def}. Then the NLEP \eqref{eq:rigorous_nlep} admits a positive real eigenvalue.
\end{theorem}

\begin{proof}
	Let $\Lambda_0$ be the principal eigenvalue of $\mathscr{L}_{y_0}$. First note that by Lemma \ref{lem:principal_eig_Ly0} and \ref{lem:second_eig_Ly0} the principal and second eigenvalues of $\mathscr{L}_{y_0}$ satisfy $\Lambda_1 < 0 < \Lambda_0$. Moreover the corresponding eigenfunction $\Phi_0$ is of one sign and we may assume that $\Phi_0>0$ and $\int_0^\infty \Phi_0^2 = 1$. Observe that if $\lambda_0\neq \Lambda_0$ is a positive eigenvalue of the NLEP \eqref{eq:rigorous_nlep} then
	\begin{equation*}
	\phi = \mu\frac{\int_0^\infty w\phi}{\int_0^\infty w^2} (\mathscr{L}_{y_0} - \lambda_0)^{-1} w^2,
	\end{equation*}
	and since $\int_0^\infty w\phi\neq 0$ the above equation is equivalent to $h(\lambda_0) = 0$ where
	\begin{equation}
	h(\lambda) \equiv \int_0^\infty w (\mathscr{L}_{y_0}-\lambda)^{-1} w^2 - \frac{\int_0^\infty w^2}{\mu}.
	\end{equation}
	We now show that such a $\lambda_0$ can always be found in $0<\lambda_0<\Lambda_0$ for $0\leq \mu <\mu_c$.
	
	First we calculate $h(0) = \int_0^\infty w^2 (\mu_{c}^{-1} - \mu^{-1}) < 0$. Next we we let $\psi$ be the unique solution to
	\begin{equation*}
	(\mathscr{L}_{y_0} - \lambda)\psi = w^2,\quad 0<y<\infty;\qquad \psi'(0) = 0.
	\end{equation*}
	Decomposing $\psi = c_0\Phi_0 + \psi^\perp$ where $\int_0^\infty\Phi_0\psi^\perp = 0$ we find that $\psi^\perp$ satisfies
	\begin{equation}\label{eq:psi^perp_eq}
	(\mathscr{L}_{y_0}-\lambda)\psi^\perp = w^2 - c_0(\Lambda_0-\lambda)\Phi_0,\quad 0<y<\infty;\qquad (\psi^\perp)'(0) = 0.
	\end{equation}
	Multiplying by $\Phi_0$ and integrating by parts we obtain $c_0 = (\Lambda_0-\lambda)^{-1}\int_0^\infty w^2\Phi_0$ and therefore
	\begin{equation}\label{eq:temp_h_expression}
	h(\lambda) = \frac{\int_0^\infty w^2 \Phi_0\int_0^\infty w\Phi_0}{\Lambda_0 - \lambda} + \int_0^\infty w\psi^\perp - \frac{\int_0^\infty w^2}{\mu}.
	\end{equation}
	On the other hand, if we multiply \eqref{eq:psi^perp_eq} by $\psi^\perp$ and integrate then we obtain
	\begin{equation}\label{eq:temp_psi^perp_integral}
	-\int_0^\infty |\psi^\perp|^2 \biggl(\lambda + \frac{-\int_0^\infty\psi^\perp\mathscr{L}_{y_0}\psi^\perp}{\int_0^\infty |\psi^\perp|^2}\biggr) = \int_0^\infty w^2 \psi^\perp.
	\end{equation}
	By Lemma \ref{lem:second_eig_Ly0} and the variational characterization of the second eigenvalue of $\mathscr{L}_{y_0}$ we obtain
	\begin{equation*}
	0 < -\Lambda_1 = \inf_{\substack{\Phi\in H^2([0,\infty)) \\ \int_0^\infty\Phi\Phi_0=0}} \frac{\int_0^\infty |\Phi'|^2 + |\Phi|^2 - 2 w|\Phi|^2}{|\Phi|^2} \leq \frac{-\int_0^\infty \psi^\perp \mathscr{L}_{y_0}\psi^\perp}{\int_0^\infty |\psi^\perp|^2}.
	\end{equation*}
	Substituting into \eqref{eq:temp_psi^perp_integral} we calculate $||\psi^\perp||_{L^2([0,\infty))}^2 \leq \lambda^{-1}||w^2||_{L^2([0,\infty))} ||\psi^\perp||_{L^2([0,\infty))}$ so that $||\psi^\perp||_{L^2([0,\infty))}$ and hence also $\int_0^\infty w \psi^\perp$ are bounded as $\lambda\rightarrow\Lambda_0 > 0$. Therefore, from \eqref{eq:temp_h_expression} we deduce $h(\lambda)\rightarrow +\infty$ as $\lambda\rightarrow\Lambda_0^-$. By a continuity argument we deduce the existence of a $\lambda_0\in(0,\Lambda_1)$ such that $h(\lambda_0)=0$.
\end{proof}

We conclude this section by establishing sufficient conditions for the stability of the NLEP \eqref{eq:rigorous_nlep} in both the small- and large-shift cases. Suppose that $\phi = \phi_R + i\phi_I$ and $\lambda = \lambda_R + i\lambda_I$ satisfies the NLEP. Separating real and imaginary parts in \eqref{eq:rigorous_nlep} then yields the system
\begin{subequations}
	\begin{align}
	& \mathscr{L}_{y_0}\phi_R - \mu\frac{\int_0^\infty w\phi_R}{\int_0^\infty w^2}w^2 = \lambda_R\phi_R - \lambda_I\phi_I,\qquad \mathscr{L}_{y_0}\phi_I - \mu\frac{\int_0^\infty w\phi_I}{\int_0^\infty w^2}w^2 = \lambda_R\phi_I + \lambda_I\phi_R.
	\end{align}
\end{subequations}
Multiplying the first and second equations by $\phi_R$ and $\phi_I$ respectively, integrating, and then adding them together gives
\begin{equation}
\lambda_R \int_0^\infty|\phi|^2 = -L_1(\phi_R,\phi_R) - L_1(\phi_I,\phi_I),
\end{equation}
where we define
\begin{equation}\label{eq:def_L1}
L_1(\Phi,\Phi) \equiv \int_0^\infty |\Phi'|^2 + \Phi^2 - 2w\Phi^2 + \mu\frac{\int_0^\infty w\Phi\int_0^\infty w^2\Phi}{\int_0^\infty w^2}.
\end{equation}
It is clear that if $L_1(\Phi,\Phi) > 0$ for all $\Phi\in H^2([0,\infty))$ then the NLEP \eqref{eq:rigorous_nlep} will be linearly stable. In the next theorem we determine sufficient conditions on $\mu\geq 0$ and $y_0\geq 0$ for which the NLEP is linearly stable.

\begin{figure}
	\centering
	\centering
	\begin{subfigure}{0.33\textwidth}
		\centering
		\includegraphics[scale=0.65]{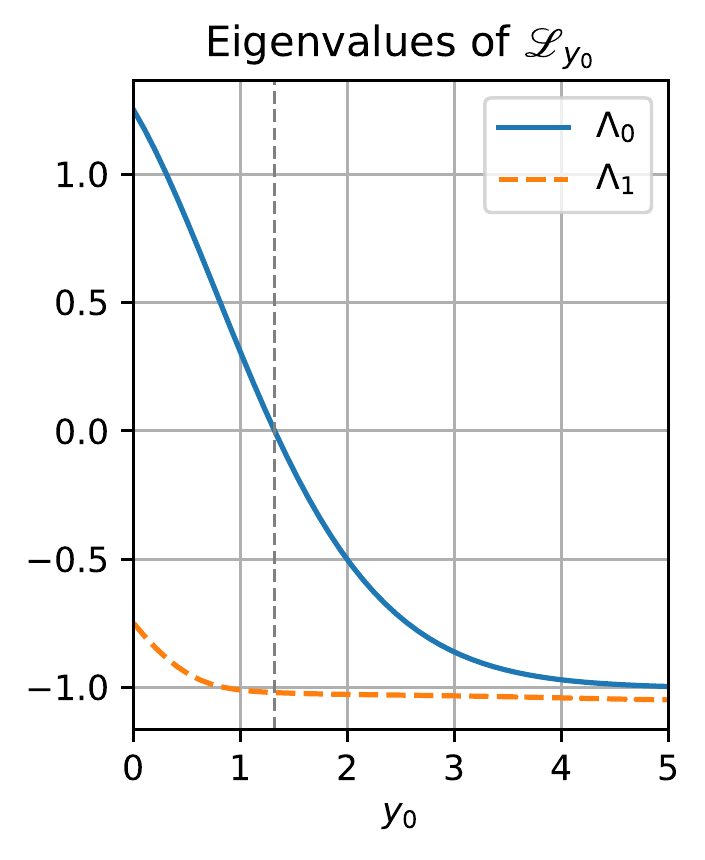}
		\caption{}\label{fig:L0_eigenvalues}
	\end{subfigure}%
	\begin{subfigure}{0.33\textwidth}
		\centering
		\includegraphics[scale=0.65]{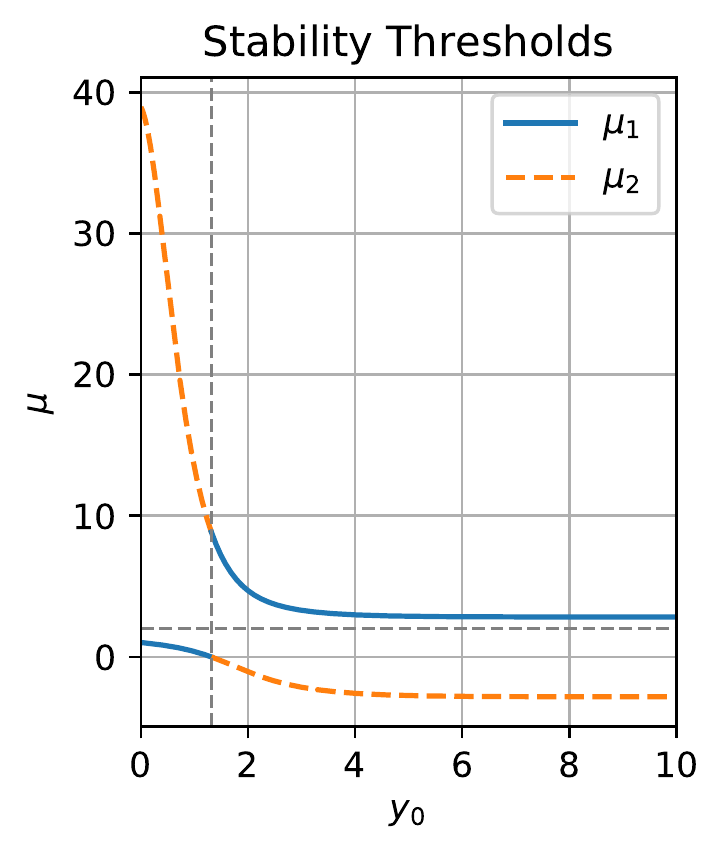}
		\caption{}\label{fig:mu1_mu2}
	\end{subfigure}%
	\begin{subfigure}{0.33\textwidth}
		\centering
		\includegraphics[scale=0.65]{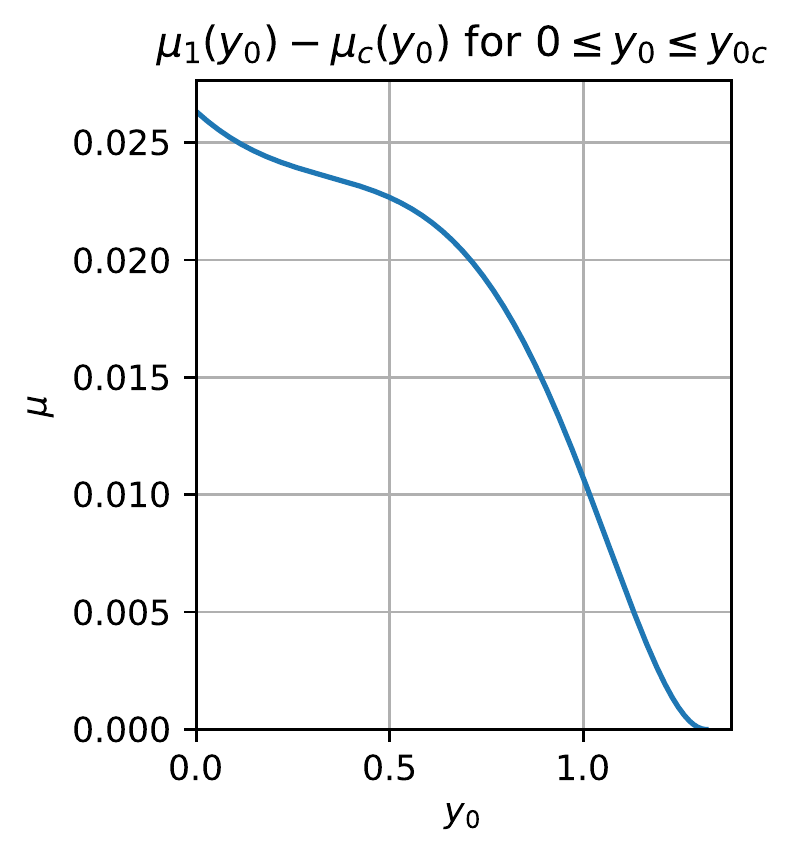}
		\caption{}\label{fig:mu1_minus_mu2}
	\end{subfigure}%
	\caption{(A) Plot of the numerically computed principal and second eigenvalues of the operator $\mathscr{L}_{y_0}$. The dashed vertical line corresponds to $y_0=y_{0c}$. (B) Plot of the stability thresholds $\mu_1$ and $\mu_2$ as functions of $y_0$. The dashed vertical and horizontal lines correspond to $y_0=y_{0c}$ and $\mu=2$ respectively. The NLEP \eqref{eq:rigorous_nlep} has been rigorously demonstrated to be stable in the region bounded by the curves $\mu_1$ and $\mu_2$. Note that $\mu_1$ and $\mu_2$ are interchanged as $y_0$ passes through $y_{0c}$. (C) Plot of $\mu_1(y_0)-\mu_c(y_0)$ for $0\leq y_0 < y_{0c}$. The NLEP is unstable for $\mu<\mu_c$ and stable for $\mu_1<\mu<\mu_2$ when $0\leq y_0<y_{0c}$. It is conjectured that the NLEP is stable for $\mu>\mu_c$.}\label{fig:rigorous}
\end{figure}

\begin{theorem}\label{thm:stability_theorem}
	If $0\leq y_0 < y_{0c}$ and $\mu_1(y_0)<\mu<\mu_2(y_0)$, or $y_0>y_{0c}$ and $0\leq \mu < \mu_1(y_0)$ where
	\begin{align}\label{eq:mu1_mu2_def}
	& \mu_1(y_0) \equiv \frac{2\int_0^\infty w^2}{\int_0^\infty w \mathscr{L}_{y_0}^{-1} w^2 + \sqrt{\int_0^\infty w \mathscr{L}_{y_0}^{-1} w\int_0^\infty w^2 \mathscr{L}_{y_0}^{-1} w^2}},\\
	& \mu_2(y_0) \equiv \frac{2\int_0^\infty w^2}{\int_0^\infty w \mathscr{L}_{y_0}^{-1} w^2 - \sqrt{\int_0^\infty w \mathscr{L}_{y_0}^{-1} w\int_0^\infty w^2 \mathscr{L}_{y_0}^{-1} w^2}},
	\end{align}
	then Re$\lambda<0$ for all eigenvalues of the NLEP \eqref{eq:rigorous_nlep}.
\end{theorem}

\begin{proof}
	We first prove the result for $0\leq y_0 < y_{0c}$. When $\mu=2$ and $y_0=0$, Lemma 5.1 (2) in \cite{wei_1999} implies that $L_1(\Phi,\Phi)>0$ for all $\Phi\in H^2((\infty,\infty))$ and hence, by restricting to even functions, also for all $\Phi\in H^2([0,\infty))$. In particular, by the variational characterization of the principal eigenvalue, this implies that the principal eigenvalue of the self-adjoint operator
	\begin{equation}
	L_1\Phi \equiv \mathscr{L}_{y_0}\Phi - \frac{\mu}{2}\frac{\int_0^\infty w\Phi}{\int_0^\infty w^2} w^2  - \frac{\mu}{2}\frac{\int_0^\infty w^2 \Phi}{\int_0^\infty w^2} w,
	\end{equation}
	must be negative. We then perturb $y_0\geq 0$ and $\mu$ until $L_1$ has a zero eigenvalue and for which we may solve
	\begin{equation}
	\Phi = c_0\mathscr{L}_{y_0}^{-1}w^2 + c_1\mathscr{L}_{y_0}^{-1}w.
	\end{equation}
	Substituting back into $L_1\Phi = 0$ we obtain the system
	\begin{equation}
	\biggl(\frac{\mu}{2}\frac{\int_0^\infty w \mathscr{L}_{y_0}^{-1}w^2}{\int_0^\infty w^2} - 1\biggr) c_0 + \frac{\mu}{2}\frac{\int_0^\infty w \mathscr{L}_{y_0}^{-1}w}{\int_0^\infty w^2} c_1 = 0,\qquad \frac{\mu}{2}\frac{\int_0^\infty w^2 \mathscr{L}_{y_0}^{-1}w^2}{\int_0^\infty w^2} c_0 + \biggl(\frac{\mu}{2}\frac{\int_0^\infty w^2 \mathscr{L}_{y_0}^{-1}w}{\int_0^\infty w^2} - 1\biggr) c_1 = 0.
	\end{equation}
	Since $\int_0^\infty w^2\mathscr{L}_{y_0}^{-1} w = \int_0^\infty w \mathscr{L}_{y_0}^{-1} w^2$ a nontrivial solutions exists if and only if
	\begin{equation}\label{eq:temp_solvability}
	\biggl(\frac{\mu}{2}\frac{\int_0^\infty w\mathscr{L}_{y_0}^{-1} w^2}{\int_0^\infty w^2} - 1\biggr)^2 - \frac{\mu^2}{4} \frac{\int_0^\infty w\mathscr{L}_{y_0}^{-1}w \int_0^\infty w^2 \mathscr{L}_{y_0}^{-1} w^2}{\bigl(\int_0^\infty w^2\bigr)^2 } = 0,
	\end{equation}
	where explicit formulae for each integral can be found in \eqref{eq:L_y0_integrals}. When $\mu=2$ and $y_0 = 0$ the left hand side of \eqref{eq:temp_solvability} equals $ -9/10 < 0$ and therefore we have stability for $0\leq y_0<y_{0c}$ if
	\begin{equation}
	\biggl(\frac{\mu}{2}\frac{\int_0^\infty w \mathscr{L}_{y_0}^{-1} w^2}{\int_0^\infty w^2} - 1\biggr)^2 - \frac{\mu^2}{4} \frac{\int_0^\infty w \mathscr{L}_{y_0}^{-1} w\int_0^\infty w^2 \mathscr{L}_{y_0}^{-1} w^2}{\bigl(\int_0^\infty w^2\bigr)^2 }  < 0,
	\end{equation}
	which is easily seen to be equivalent to $\mu_1<\mu<\mu_2$.
	
	The thresholds $\mu_1$ and $\mu_2$ are singular as $y_0\rightarrow y_{0c}$ and therefore the continuity argument from above does not extend to $y_0>y_{0c}$. However $L_1(\Phi,\Phi)>0$ by Lemma \ref{lem:principal_eig_Ly0} if $y_0>y_{0c}$ and $\mu=0$. Therefore we proceed with the same continuity argument as above, but starting from $\mu=0$. This yields the same criteria, but since $\mathscr{L}_{y_0}^{-1} w^2 < 0$ by Lemma \ref{lem:sign_lemma} the sufficient condition is now $\mu_2(y_0)< 0\leq \mu < \mu_1(y_0)$.
	
\end{proof}

Both of the stability thresholds $\mu_1$ and $\mu_2$ defined in \eqref{eq:mu1_mu2_def} as well as the instability threshold $\mu_c$ defined in \eqref{eq:muc_def} are easily computed using \eqref{eq:L_y0_integrals}. In Figures \ref{fig:mu1_mu2} and \ref{fig:mu1_minus_mu2} we plot the stability thresholds and the difference $\mu_1-\mu_c$ respectively. In particular, from the plot in \ref{fig:mu1_minus_mu2} we see that $\mu_1>\mu_c$. We conjecture, that as in the $y_0=0$ case, the NLEP is stable for all $\mu>\mu_c$. In Appendix \ref{app:conjecture} we provide numerical support for this conjecture by plotting $\text{Re}\lambda_0$ versus $\mu$ and $y_0$ in Figure \ref{fig:conjecture_1}. In addition, we plot $\Lambda_0-\text{Re}(\lambda_0)$ in Figure \ref{fig:conjecture_2} which suggest that $\text{Re}(\lambda_0) \leq \Lambda_0$.

\section{Examples}\label{sec:examples}

In this section we illustrate the effect of introducing a nonzero boundary flux for the activator by considering three distinct examples. Specifically, we first study the stability of a single boundary spike concentrated at $x=0$ when $A\geq 0$ and $B=0$. Using a winding number argument we illustrate that the stability of the single spike is improved by increasing the boundary flux $A$. Moreover, we illustrate that if $A$ exceed a threshold, then the spike is stable independently of the parameters $\tau\geq 0$ and $D>0$. We then consider the structure and stability of a two-boundary-spike pattern when the boundary fluxes are equal, $A=B\geq 0$. One of the key findings is that if $A>0$ then the range of $D>0$ values for which asymmetric patterns exist is extended. Additionally, by assuming that $\tau\ll 1$ we study the stability of both symmetric and asymmetric two-boundary-spike patterns to competition (zero eigenvalue crossing) instabilities. We demonstrate that one branch of asymmetric patterns is always stable. Similarly, in our final example we consider a two-boundary-spike pattern with a one-sided boundary flux $A\geq 0$ and $B=0$. We demonstrate the existence of several asymmetric patterns, with a certain branch of these patterns always being stable. For each example we include full numerical simulations of the GM system \eqref{eq:pde_gm} using the finite element software FlexPDE \cite{flexpde}.

\subsection{Example 1: One Boundary Spike at $x=0$ with $A>0$ and $B=0$}

\begin{figure}[t!]
	\centering
	\begin{subfigure}{0.33\textwidth}
		\centering
		\includegraphics[scale=0.675]{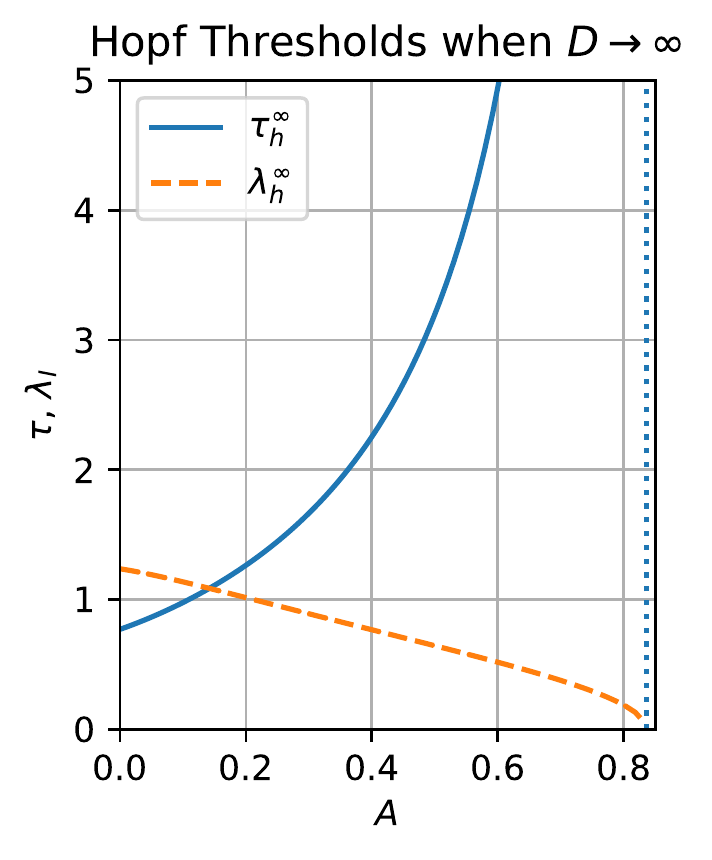}
		\caption{}\label{fig:hopf_inf}
	\end{subfigure}%
	\begin{subfigure}{0.33\textwidth}
		\centering
		\includegraphics[scale=0.675]{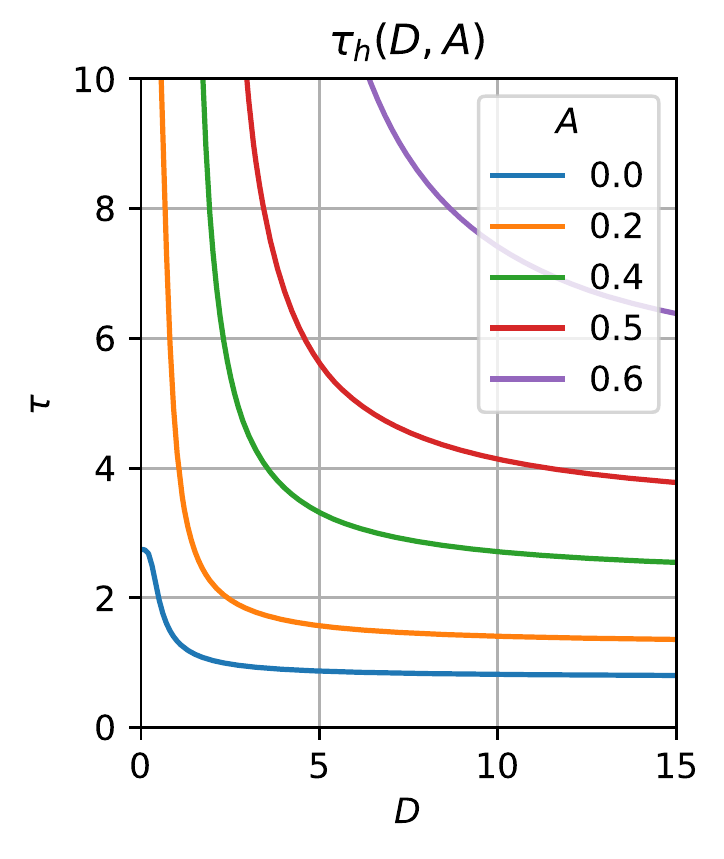}
		\caption{}\label{fig:hopf_tau}
	\end{subfigure}%
	\begin{subfigure}{0.33\textwidth}
		\centering
		\includegraphics[scale=0.675]{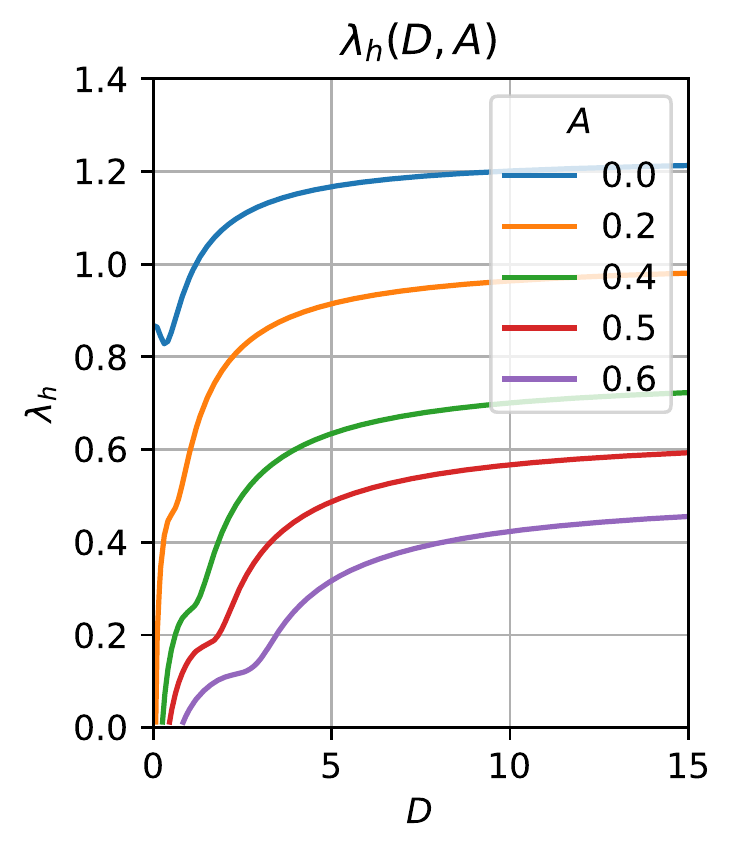}
		\caption{}\label{fig:hopf_lambda}
	\end{subfigure}%
	\caption{Hopf bifurcation threshold and accompanying eigenvalue for a single boundary-spike solution with one-sided boundary flux $A\geq 0$ in (A) the shadow limit, $D\rightarrow\infty$, and (B and C) for finite $D>0$ at select values of $0\leq A < q_{0c}$. In (A) the dashed vertical line corresponds to the threshold $A=q_{0c}$ beyond which no Hopf instabilities occur.}\label{fig:hopf_thresholds}
\end{figure}

In this example we assume that $B=0$ and investigate the role of a non-negative flux, $A\geq 0$, on the stability of a single boundary spike concentrated at $x=0$. We denote the left boundary shift parameter by $y_0 = y_L$ so that using \eqref{eq:greens_func} and \eqref{eq:gluing_parameters} we reduce  \eqref{eq:nlep} to \eqref{eq:rigorous_nlep} where
\begin{equation}
\mu(\lambda) = 2\frac{\omega_0\tanh\omega_0}{\omega_\lambda\tanh\omega_\lambda}.
\end{equation}
Recalling that $\omega_\lambda = \sqrt{(1+\tau\lambda)/D}$ we first observe that $0 < \mu(\lambda) \leq 2$ for all real-valued $\lambda\geq 0$ so that by Theorem \ref{thm:stability_theorem} and Figure \ref{fig:mu1_mu2} the NLEP has no non-negative real eigenvalues. Next we determine whether the NLEP has any unstable complex-valued eigenvalues by using a winding number argument. Assuming that $\lambda$ is not in the spectrum of $\mathscr{L}_{y_0}$ we let $\phi = (\mathscr{L}_{y_0}-\lambda)^{-1} w_c(y+y_0)^2$ so that as in \S\ref{subsec:algebraic} the NLEP reduces to the algebraic equation
\begin{equation}
\mathscr{A}_{y_0}(\lambda) \equiv \frac{1}{\mu(\lambda)} - \mathscr{F}_{y_0}(\lambda) = 0,
\end{equation}
where $\mathscr{F}_{y_0}(\lambda)$ is defined in \eqref{eq:F_y0_def}. In Figure \ref{fig:F_properties} we plot $\mathscr{F}_{y_0}(0)$ versus $y_0$, as well as the real and imaginary parts of $\mathscr{F}_{y_0}(i\lambda_I)$ versus $\lambda_I$ for select values of $y_0$. We integrate in $\lambda$ over a closed counter-clockwise contour consisting of the imaginary axis and a large semicircle in the right half-plane. Since
\begin{equation}\label{eq:F_mu_asymptotics}
\mathscr{F}_{y_0}(\lambda)=O(|\lambda|^{-1})\quad\text{and}\quad \mu(\lambda) = O(\lambda^{-1/2})\qquad \text{as}\quad |\lambda|\rightarrow\infty,\quad \text{Re}\lambda>0,
\end{equation}
the change in argument of $\mathscr{A}_{y_0}(\lambda)$ over the large semicircle is $\frac{\pi}{2}$. Moreover, in $\text{Re}\lambda>0$, $\mu(\lambda)\neq 0$ whereas by Lemmas \ref{lem:principal_eig_Ly0} and \ref{lem:second_eig_Ly0} we deduce that $\mathscr{F}_{y_0}(\lambda)$ has one (resp. zero) pole(s) if $y_0<y_{0c}$ (resp. $y_0\geq y_{0c}$). Letting $Z$ denote the number of zeros of $\mathscr{A}_{y_0}(\lambda)$ in $\text{Re}\lambda>0$ it follows from the argument principle that
\begin{equation}
Z = \frac{1}{\pi}\Delta \text{arg}\mathscr{A}_{y_0}(i\lambda_I)\bigr|_{+\infty}^0  + \begin{cases} 5/4, & y_0 < y_{0c}, \\ 1/4, & y_0\geq y_{0c},\end{cases}
\end{equation}
where the first term on the right hand side denotes the change in argument of $\mathscr{A}_{y_0}(\lambda)$ as $\lambda$ follows the imaginary axis from $\lambda=+i\infty$ to $\lambda=0$. Note in addition that we have used $\mathscr{A}_{y_0}(\bar{\lambda}) = \overline{\mathscr{A}_{y_0}(\lambda)}$ to obtain $\Delta\arg\mathscr{A}_{y_0}(i\lambda_I)\bigr|_{+\infty}^{-\infty} = 2\Delta\arg\mathscr{A}_{y_0}(i\lambda_I)\bigr|_{+\infty}^0$. From \eqref{eq:F_mu_asymptotics} we immediately deduce that $\arg\mathscr{A}_{y_0}(+i\infty) = \pi/4$. On the other hand, using \eqref{eq:L_y0_integrals}, we evaluate $\mathscr{A}_{y_0}(0) = \frac{1}{2} - \mathscr{F}_{y_0}(0)\lessgtr0$ for $y_0\lessgtr y_{0c}$ (see also Figure \ref{fig:F0}). We will consider the cases $y_0\geq y_{0c}$ and $y_0<y_{0c}$ separately below.

\begin{figure}[t!]
	\centering
	\begin{subfigure}{0.33\textwidth}
		\centering
		\includegraphics[scale=0.675]{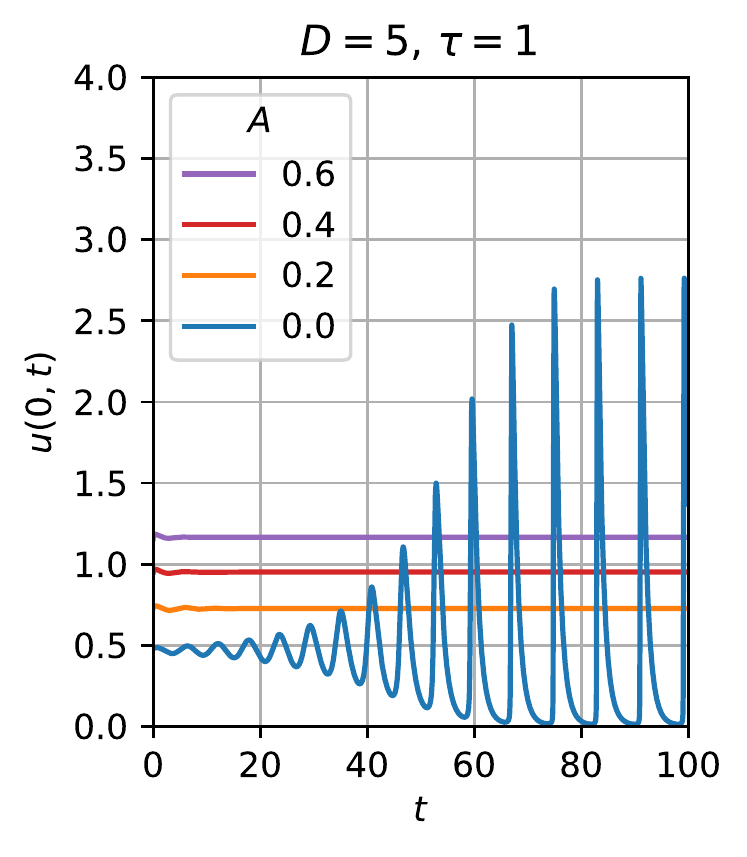}
		\caption{}\label{fig:hopf_0}
	\end{subfigure}%
	\begin{subfigure}{0.33\textwidth}
		\centering
		\includegraphics[scale=0.675]{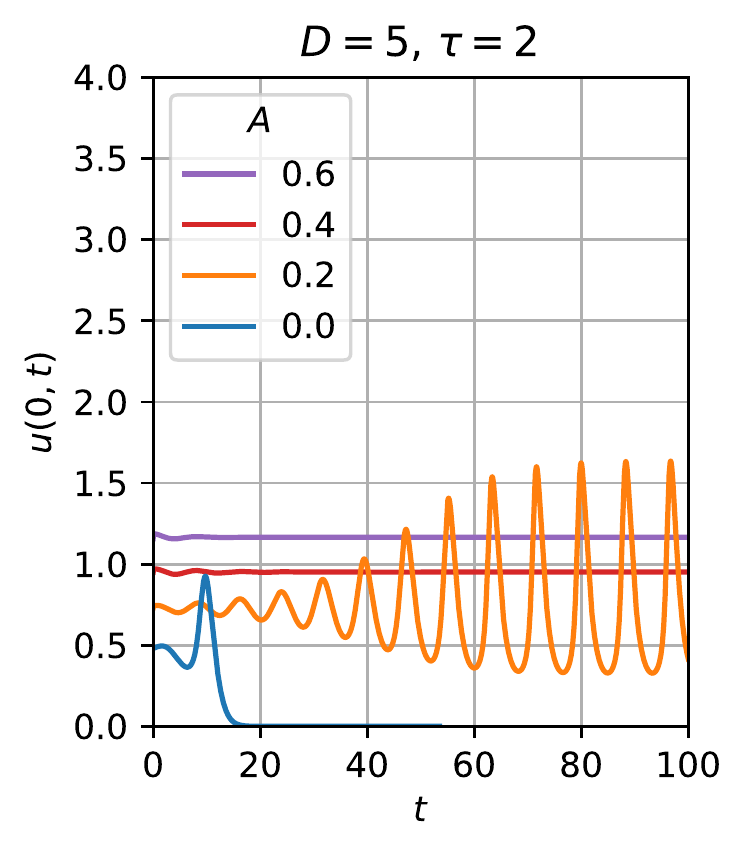}
		\caption{}\label{fig:hopf_1}
	\end{subfigure}%
	\begin{subfigure}{0.33\textwidth}
		\centering
		\includegraphics[scale=0.675]{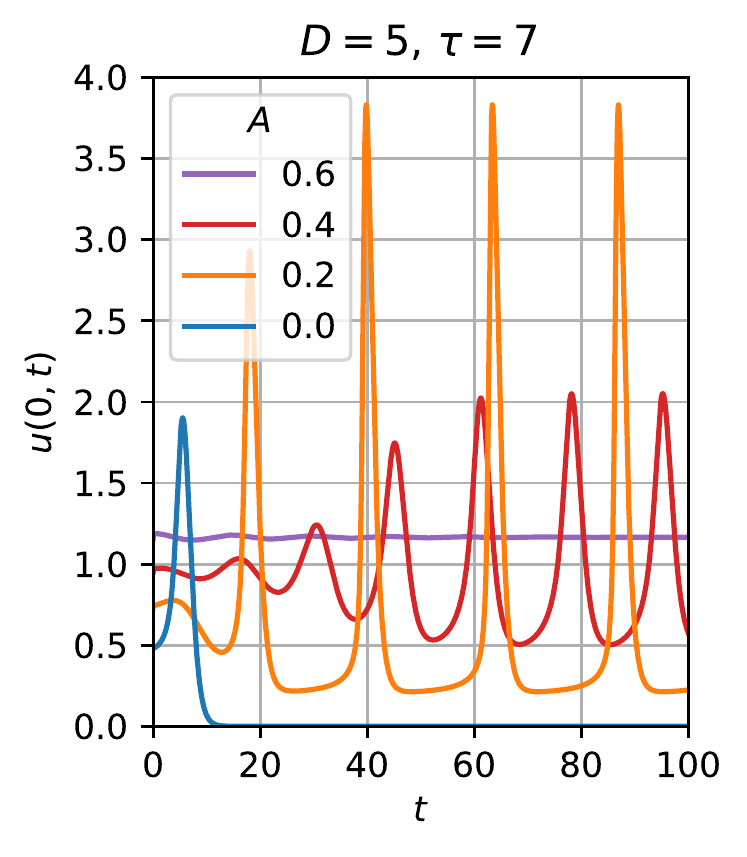}
		\caption{}\label{fig:hopf_2}
	\end{subfigure}%
	\caption{Plots of $u(0,t)$ for a one boundary-spike solution with one-sided boundary flux $x=0$ (i.e.\@ $A\geq 0$ and $B=0$) with $\varepsilon=0.005$. Note that increasing the boundary flux $A$ stabilizes the single boundary-spike solution for fixed values of $D$ and $\tau$.}\label{fig:hopf_numerical}
\end{figure}

If $y_0\geq y_{0c}$ then $\text{Re}\mathscr{F}_{y_0}(i\lambda_I) < 0$ for all $\lambda_I>0$ (see Figure \ref{fig:real_F_imag}) so $\mathscr{A}_{y_0}(i\lambda_I)$ never crosses the imaginary axis for all $\lambda_I>0$. As a result $\Delta\arg\mathscr{A}_{y_0}(i\lambda_I)|_{\infty}^{0}=-\pi/4$ and therefore $Z=0$ if $y_0>y_{0c}$. Since $y_0$ is monotone decreasing in $D$ when $A>0$, we deduce that there is a threshold $D_c(A)$ such that the single spike pattern is stable for all $\tau\geq 0$ if $D\geq D_c(A)$. Substituting $y_0=y_{0c}$ and using \eqref{eq:gluing_parameters} and \eqref{eq:y0c_def} yields the threshold parameter
\begin{equation}\label{eq:q0c_def}
q_{0c} = \frac{4+ 3\sqrt{3}}{11} \approx 0.83601,
\end{equation}
with which $D_c(A)$ is found by solving the transcendental equation
\begin{equation}
\tanh D_c^{-1/2} = q_{0c}^{-1}AD_c^{-1/2}.
\end{equation}
It is easy to see that if $A\geq q_{0c}$ then $D_c=\infty$ is the only positive solution. On the other hand, if $0<A<q_{0c}$ then this equation has a unique positive solution that is monotone increasing in $A$ and satisfies $D_c\rightarrow 0^+$ as $A\rightarrow 0^+$ and $D_c\rightarrow\infty$ as $A\rightarrow q_{0c}^-$. In summary we deduce that the single spike pattern is stable for all $D>0$ and $\tau\geq 0$ if $A\geq q_{0c}$, or for all $\tau\geq 0$ if $0<D\leq D_c(A)$ and $0<A<q_{0c}$. To determine the stability when $0\leq A < q_{0c}$ and $D>D_c(A)$ we must consider the case $0\leq y<y_{0c}$.

Next we assume that $0\leq y_{0} < y_{0c}$. We begin by considering the \textit{shadow limit}, defined by $D\rightarrow\infty$, for which $\mu(i\lambda_I)\sim 2 (1+i\tau\lambda_I)^{-1}$ and hence also
$$
\text{Re}\mathscr{A}_{y_0}(i\lambda_I)\sim \frac{1}{2} - \text{Re}\mathscr{F}_{y_0}(i\lambda_I),\qquad \text{Im}\mathscr{A}_{y_0}(i\lambda_I)\sim \frac{\tau\lambda_I}{2} - \text{Im}\mathscr{F}_{y_0}(i\lambda_I),\qquad D\rightarrow\infty.
$$
Since $\mathscr{F}_{y_0}(0)>1/2$ and $\text{Re}\mathscr{F}_{y_0}(i\lambda_I)\rightarrow 0$ as $\lambda_I\rightarrow\infty$ (see Appendix \ref{app:F_properties} and accompanying Figure \ref{fig:F_properties}) we deduce that there exists a solution to $\text{Re}\mathscr{A}_{y_0}(i\lambda_I)=0$. Moreover, in Figure \ref{fig:real_F_imag} we observe that when $\text{Re}F_{y_0}(i\lambda_I)$ is positive it is also monotone decreasing in $\lambda_I$. Therefore there exists a unique eigenvalue $\lambda_h^\infty$ and time constant $\tau_h^\infty = 2\text{Im}\mathscr{F}_{y_0}(i\lambda_I)/\lambda_h^\infty$ such that $\mathscr{A}_{y_0}(i\lambda_h^\infty)=0$. Furthermore, since $\text{Im}\mathscr{A}_{y_0}(i\lambda_h^\infty)\lessgtr 0$ if $\tau \lessgtr \tau_h^\infty$ we get
$$
\Delta\arg\mathscr{A}_{y_0}(i\lambda_I)|_{\infty}^{0}=\begin{cases} -5\pi/4, & \tau<\tau_h^\infty,\\ 3\pi/4,  &\tau>\tau_h^\infty.\end{cases}
$$
The single boundary spike solution therefore undergoes a Hopf bifurcation as $\tau$ exceeds the Hopf bifurcation threshold $\tau_h^\infty$. Using the shadow limit threshold as an initial guess, we numerically continue the Hopf bifurcation threshold for finite values of $D>0$ to obtain the Hopf bifurcation threshold $\tau_h(D,A)$ and accompanying critical eigenvalue $\lambda=i\lambda_h(D,A)$ shown in Figure \ref{fig:hopf_tau} and \ref{fig:hopf_lambda} respectively.

The above analysis, together with the plots of $\tau_h^\infty(A)$ and $\tau_h(D,A)$ in Figures \ref{fig:hopf_inf} and \ref{fig:hopf_tau} respectively, indicate that the single boundary spike solution is stabilized as $A>0$ is increased. Additionally, if $A$ exceeds the threshold $q_{0c}$ given in \eqref{eq:q0c_def}, then the single boundary spike is stable independently of the parameters $\tau\geq 0$ and $D>0$. We illustrate the onset of Hopf instabilities when $D=5$ for $\tau=1,2,7$ and $A=0,0.2,0.4,0.6$ by numerically computing the solution of \eqref{eq:pde_gm} using FlexPDE 6 \cite{flexpde} and plotting $u(0,t)$ in Figure \ref{fig:hopf_numerical}. In particular we observe that the single spike pattern is stabilized by increasing the boundary flux $A$. Additionally, our numerical simulations show good qualitative agreement with the Hopf bifurcation thresholds plotted in Figure \ref{fig:hopf_tau}.

\subsection{Example 2: Two Boundary Spikes with $A=B\geq 0$}

In this example we investigate the role of equal boundary fluxes on the structure and stability of a two-boundary-spike pattern. Using the method of \S\ref{subsec:gluing-method}, a two-boundary-spike pattern is found by letting $l_L=l$ and $l_R=1-l$ and solving \eqref{eq:gluing_equation_2}, which is explicitly given by
\begin{equation}\label{eq:two_bdry_spike_eq}
\frac{\tanh\omega_0 l}{\eta(y_L)\cosh\omega_0 l} - \frac{\tanh\omega_0(1-l)}{\eta(y_R)\cosh\omega_0(1-l)} = 0,
\end{equation}
for $0<l<1$ where $\eta$ is given by $\eqref{eq:eta_def}$ and $y_L$ and $y_R$ are given by \eqref{eq:gluing_parameters}. Note that by \eqref{eq:gluing_solution} the algebraic equation \eqref{eq:two_bdry_spike_eq} is equivalent to
\begin{equation}\label{eq:two_bdry_spike_eq_explicit}
\frac{\xi_L}{\xi_R} = \frac{\cosh\omega_0l}{\cosh\omega_0(1-l)},
\end{equation}
from which we deduce that $l\lessgtr 1/2$ implies $\xi_L \lessgtr \xi_R$. In particular $l=1/2$ solves \eqref{eq:two_bdry_spike_eq} for all $A\geq 0$ and since in this case $\xi_L=\xi_R$ we refer to it as the symmetric solution. For the remainder of this example we will construct asymmetric two-spike patterns for which $0<l<1/2$ (by symmetry the case $l>1/2$ is identical) and then study the linear stability of both the symmetric and asymmetric patterns.

Before constructing asymmetric two-boundary-spike patterns for $A\geq 0$ we first recall the following existence result from \cite{ward_2002_asymmetric} in the case $A=0$. Specifically, we let $z=\omega_0 l_L$ and $\tilde{z}=\omega_0 l_R$ so that when $A=0$ the system \eqref{eq:gluing_equation} (and hence also \eqref{eq:two_bdry_spike_eq}) is equivalent to
\begin{equation}\label{eq:bdry_bdry_A0}
z + \tilde{z} = \omega_0,\qquad b(z) = b(\tilde{z}),\qquad b(z) \equiv \frac{\tanh z}{\cosh z}.
\end{equation}
It follows from Result 2.3 (with $k_1=k_2=1$, $\mu=1$, and $r=1$) of \cite{ward_2002_asymmetric} that \eqref{eq:two_bdry_spike_eq} has a unique solution $0<l<1$ if and only if 
\begin{equation}\label{eq:A=0_asymmetric_threshold}
0<D<D_{c1}\equiv [2\log(1+\sqrt{2})]^{-2}\approx 0.322.
\end{equation}
When $A>0$ we solve \eqref{eq:two_bdry_spike_eq} numerically and find that for given values of $D$ and $A>0$ it accepts zero, one, or two solutions in the range $0<l<1/2$. Rather than solving \eqref{eq:two_bdry_spike_eq} numerically for $l$ as a function of $A$ and $D$, we found it more convenient to solve for $A=A(D,l)$. The results of our numerical calculations are shown in Figure \ref{fig:BOT_BB_BA_colormap} where we plot $A=A(D,l)$ as well as the curve $l=l_\text{max}(D)$ along which $A(D,l)$ is maximized for a fixed value of $D$, and the curve $l=l_{c1}(D)$ along which $A(D,l_{c1}(D)) = A(D,1/2)$ for $D_{c1}<D<D_{c2}\approx0.660$ and $l_{c1}(D)=0$ for $0<D<D_{c1}$. Consequently, \eqref{eq:two_bdry_spike_eq} has zero solutions in $0<l<1/2$ if $A>A_\text{max}(D)\equiv A(D,l_\text{max}(D))$, whereas it has two solutions, one with $l<l_{c1}(D)$ and the other with $l>l_{c1}(D)$, if $A_\text{max}(D) > A > A_{c1}(D) \equiv A(D,l_{c1}(D))$ for $0<D<D_{c2}$. For all other values of $D>0$ and $A>0$ equation \eqref{eq:two_bdry_spike_eq} has exactly one solution in $0<l<1/2$. We summarize these existence thresholds in Figure \ref{fig:BOT_BB_BA_thresholds} where we note in particular that $A>0$ greatly extends the range of $D$ values over which asymmetric two-boundary-spike patterns exist.

\begin{figure}[t!]
	\centering
	\begin{subfigure}{0.5\textwidth}
		\centering
		\includegraphics[scale=0.675]{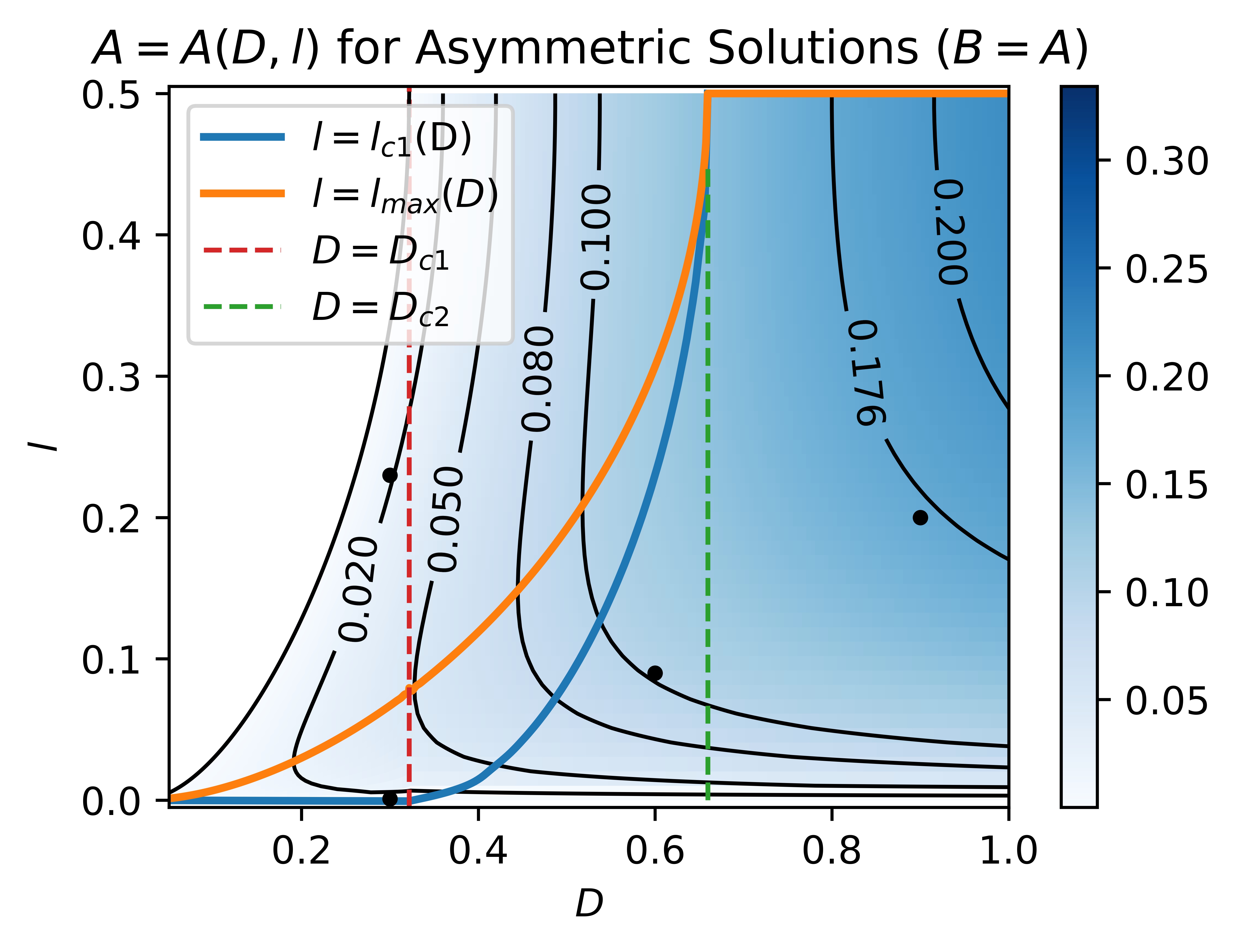}
		\caption{}\label{fig:BOT_BB_BA_colormap}
	\end{subfigure}%
	\begin{subfigure}{0.5\textwidth}
		\centering
		\includegraphics[scale=0.675]{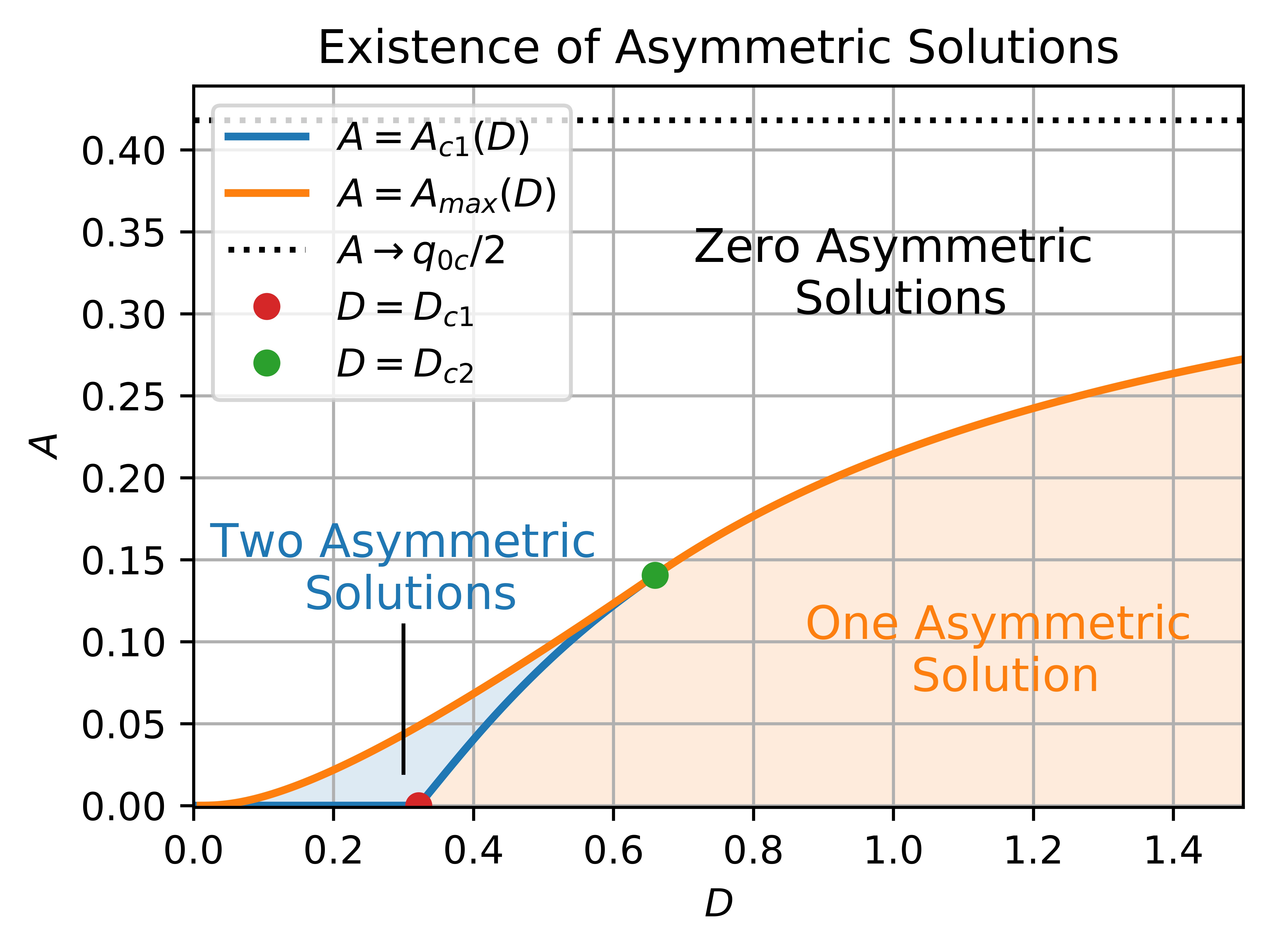}
		\caption{}\label{fig:BOT_BB_BA_thresholds}
	\end{subfigure}%
	\caption{(A) Plots of $A=A(D,l)$ for the construction of asymmetric two-boundary-spikes in Example 2. The curve $l=l_\text{max}(D)$ indicates the values of $l$ at which $A(D,l)$ is maximized, while $l_{c1}(D)<1/2$ indicates the values of $l$ at which $A(D,l)$ has the same value as $A(D,1/2)$. The curve $l=l_\text{max}(D)$ also corresponds to the competition instability threshold for asymmetric two-boundary-spike patterns. If $l_\text{max}(D)<l<1/2$ then the corresponding asymmetric pattern is unstable, whereas it is stable if $l<l_\text{max}(D)$. (B) Thresholds for the existence of zero, one, or two asymmetric two-boundary-spike solutions in the presence of equal boundary fluxes. The values of $A=A_{c1}(D)$ for $D_{c1}<D<D_{c2}$ and $A=A_\text{max}(D)$ for $D>D_{c2}$ correspond to the competition instability threshold for the symmetric two-boundary-spike pattern. It is linearly unstable for values of $A$ below this threshold, and stable otherwise. Note that no asymmetric solutions exist for $A=0$ and $D>D_{c1}$.}\label{fig:BOT_BB_BA}
\end{figure}

We now consider the linear stability of both the symmetric and asymmetric two-boundary-spike patterns constructed above. Note that since $x_L=0$ and $x_R=1$ are fixed there are no slow drift dynamics and the linear stability of the two-spike patterns is completely determined by the $\mathcal{O}(1)$ eigenvalues calculated from the NLEP \eqref{eq:nlep}. Moreover, we assume that $\tau=0$ so that no Hopf instabilities arise (see Example 1) and for which we can exclusively focus on zero eigenvalue crossing, or competition, instabilities. We proceed by first using the rigorous results of \S\ref{sec:rigorous} to determine the linear stability of the symmetric two-spike patterns constructed above, and we will then use the algebraic reduction outlined in \S\ref{subsec:algebraic} to determine the stability of the remaining asymmetric two-spike patterns.

Using \eqref{eq:greens_func} and \eqref{eq:gluing_solution}, the NLEP \eqref{eq:nlep} for the symmetric two-spike pattern constructed above is explicitly given by
\begin{equation*}
\mathscr{L}_{y_0}\pmb{\phi} - 2\omega_0 \tanh\bigl(\tfrac{\omega_0}{2}\bigr)\frac{\int_0^\infty w_c(y+y_0)\mathscr{G}_{\omega_0}\pmb{\phi} dy}{\int_0^\infty w_c(y+y_0)^2 dy} w_c(y+y_0)^2 = \lambda\pmb{\phi},\quad 0<y<\infty;\qquad \pmb{\phi}'(0) = 0,
\end{equation*}
where
\begin{equation}\label{eq:symm_vector_matrix_def}
\pmb{\phi} = \begin{pmatrix}\phi_1 \\ \phi_2\end{pmatrix},\quad \mathscr{G}_{\omega_0} = \frac{1}{\omega_0 \sinh\omega_0}\begin{pmatrix} \cosh\omega_0 & 1 \\ 1 & \cosh\omega_0 \end{pmatrix}.
\end{equation}
The Green's matrix is symmetric and of constant row sum and therefore has eigenvectors $(1,\pm 1)^T$. Substituting $\pmb{\phi} = (\phi,\pm\phi)^T$ into the NLEP therefore yields two uncoupled scalar NLEPs of the form \eqref{eq:rigorous_nlep} with $\mu=\mu_{\pm}$ where
\begin{equation}
\mu_+ \equiv 2,\qquad \mu_- \equiv 2\tanh^2\bigl(\tfrac{\omega_0}{2}\bigr).
\end{equation}
From Theorem \ref{thm:stability_theorem} and accompanying Figure \ref{fig:mu1_mu2} we immediately deduce that the $\pmb{\phi}_+$ mode is linearly stable. On the other hand, by Theorem \ref{thm:unstbale_small_shift} the $\pmb{\phi}_{-}$ mode is unstable if $\mu_- < \mu_c$, where $\mu_c$ is defined by \eqref{eq:muc_def}. We therefore calculate the competition instability threshold by numerically solving $2\tanh^2 \tfrac{\omega_0}{2} = \mu_c(y_0)$ where $y_0=y_L=y_R$ is the shift parameter given by \eqref{eq:gluing_parameters} with $l=1/2$. Our numerical calculations indicate that the resulting instability threshold coincides with the values
\begin{equation}
A(D,1/2) = \begin{cases} A_{c1}(D), & D_{c1}<D<D_{c2}, \\ A_\text{max}(D), & D > D_{c2}, \end{cases}
\end{equation}
calculated above. In particular, the symmetric two-spike pattern is linearly unstable for all $A<A(D,1/2)$ when $D>D_{c1}$. Furthermore, since $\mu_c(y_0)<0$ for $y_0>y_{0c}$ we determine from \eqref{eq:gluing_equation_2} that there are no competition instabilities if
\begin{equation}
A>\omega_0^{-1}q_{0c}\tanh(\omega_0/2),
\end{equation}
where $q_{0c}$ is the threshold identified in Example 1 and is explicitly given by \eqref{eq:q0c_def}. Note that, analogous to the results in Example 1, in the shadow limit ($D\rightarrow\infty$) there are no competition instabilities for the symmetric pattern if $A>q_{0c}/2$ (see Figure \ref{fig:BOT_BB_BA_thresholds}). As in \S\ref{sec:rigorous} we conjecture and have numerically supported that the symmetric two-spike pattern is linearly stable for $\mu>\mu_c$ and hence for all $A > A(D,1/2)$. Finally, as is clear from Figure \ref{fig:BOT_BB_BA_thresholds}, increasing $A>0$ expands the range of $D$ values over which the symmetric two-boundary-spike pattern is linearly stable.

For the asymmetric two-boundary-spike solutions constructed above the NLEP \eqref{eq:nlep} is not diagonalizable since $w_c(y+y_L)\not\equiv w_c(y+y_R)$. We therefore can't directly apply the rigorous results of \S\ref{sec:rigorous}. To determine the competition instability we instead use the algebraic reduction outlined in \S\ref{subsec:algebraic} and seek parameter values such that
\begin{equation}\label{eq:ex_2_asy_alg}
\det(\mathbb{I}_2 - 2\omega_0^2\mathscr{G}_{\omega_0}\mathscr{D}_0) = 0
\end{equation}
where $\mathbb{I}_2$ is the $2\times 2$ identity matrix, $\mathscr{G}_{\omega_0}$ is the $2\times 2$ Green's matrix given in \eqref{eq:symm_vector_matrix_def}, and 
\begin{equation}\label{eq:example_2_comp_det}
\mathscr{D}_0 = \frac{1}{\omega_0} \begin{pmatrix} \tanh\omega_0 l \mathscr{F}_{y_L}(0) & 0 \\ 0 & \tanh\omega_0(1-l)\mathscr{F}_{y_R}(0)\end{pmatrix}.
\end{equation}
Substituting the function $A=A(D,l)$ calculated above into \eqref{eq:ex_2_asy_alg} we can solve for $l$ as a function of $D$ using standard numerical methods (specifically we used a combination of Scipy's brentq and fsolve function in Python 3.6.8). Our computations indicate that the resulting competition instability threshold coincides with the curves $l_\text{max}(D)$ for $D>0$ and $l=1/2$ for $D_{c1}<D<D_{c2}$. In fact, we can show that this is the case explicitly by first differentiating the quasi-equilibrium equation $\bm{B}=\bm{0}$ with respect to $l$ to get
\begin{equation}\label{eq:bdry_bdry_diff_quasi}
\nabla_{\pmb{\xi}} \pmb{B} \biggl(\frac{\partial\pmb{\xi}}{\partial l} + \frac{\partial\pmb{\xi}}{\partial A}\frac{\partial A}{\partial l}\biggr) + \frac{\partial\pmb{B}}{\partial A} \frac{\partial A}{\partial l} = 0.
\end{equation}
Along the curve $l_\text{max}(D)$ for $D>0$ the function $A(D,l)$ is maximized whereas, by symmetry, it is minimized along $l=1/2$ for $D_{c1}<D<D_{c2}$. In both cases $\partial A / \partial l = 0$ along these curves so that \eqref{eq:bdry_bdry_diff_quasi} becomes $\nabla_{\bm{\xi}}\bm{B} \partial\bm{\xi}/\partial l =  \bm{0}$. Differentiating \eqref{eq:two_bdry_spike_eq_explicit} with respect to $l$ implies that $\partial \bm{\xi} / \partial l\neq \bm{0}$ and therefore we deduce that $\nabla_{\bm{\xi}}\bm{B}$ is singular along $l=l_\text{max}(D)$ and $l=1/2$. By the discussion of \S\ref{subsec:algebraic} it follows that along these curves the algebraic equation \eqref{eq:ex_2_asy_alg} is satisfied and they therefore correspond to competition instability thresholds. Note in particular that the competition instability threshold along $l=1/2$ corresponds to the competition instability threshold for the symmetric two-spike pattern. As an immediate consequence it follows that $A_\text{max}\rightarrow q_{0c}/2$ as $D\rightarrow\infty$.

To determine in which of the regions demarcated by the competition instability thresholds the asymmetric two-boundary-spike patterns are linearly stable and unstable, we calculate the stability of the asymmetric patterns along the $A=0$ curve. As outlined in Appendix \ref{app:A0_stability}, the asymmetric two-boundary spike patterns when $A=0$ are always linearly unstable, and therefore we deduce that the asymmetric two-boundary-spike patterns in the region bounded by $l=1/2$ and $l=l_\text{max}(D)$ for $0<D<D_{c2}$ are linearly unstable. On the other hand, numerically calculating the dominant eigenvalue of the NLEP \eqref{eq:nlep} (see Appendix \ref{app:conjecture} for description of numerical method) for select parameter values with $l<l_\text{max}(D)$ we find that such asymmetric two-boundary spike patterns are linearly stable. In particular, we note that in the region where there are two asymmetric patterns (i.e. for $A_{c1}(D)<A<A_\text{max}(D)$), the pattern with $l>l_\text{max}(D)$ is linearly unstable while that with $l<l_\text{max}(D)$ is linearly stable. Moreover, the single asymmetric two-boundary-spike pattern that exists for $A < A_\text{max}(D)$ and $D>D_{c1}$ is linearly stable.

\begin{figure}[t!]
	\centering
	\begin{subfigure}{0.33\textwidth}
		\centering
		\includegraphics[scale=0.675]{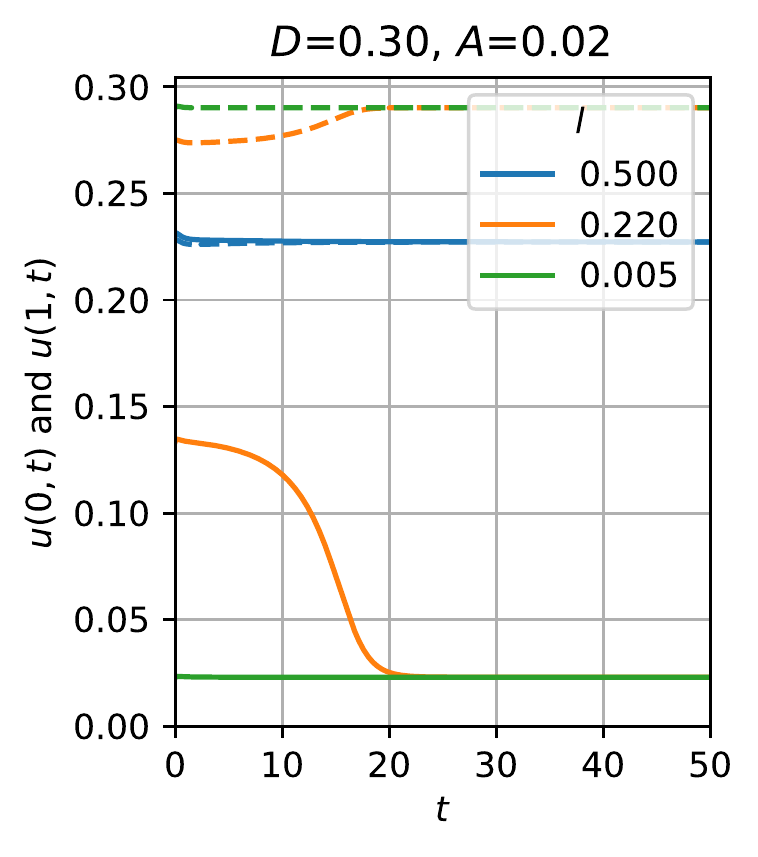}
		\caption{}\label{fig:bot_bb_ba_sim_1}
	\end{subfigure}%
	\begin{subfigure}{0.33\textwidth}
		\centering
		\includegraphics[scale=0.675]{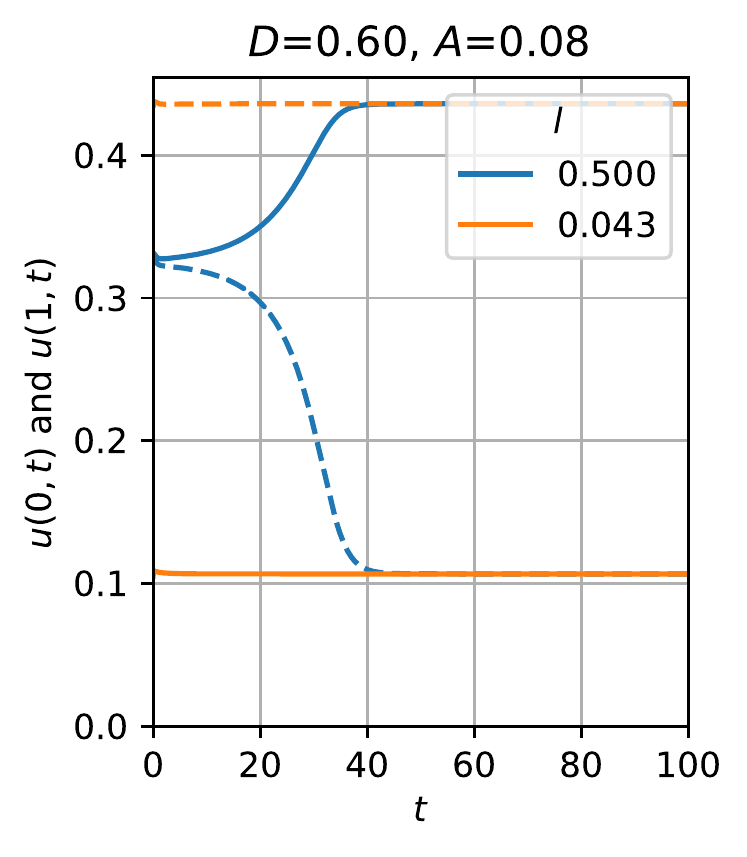}
		\caption{}\label{fig:bot_bb_ba_sim_2}
	\end{subfigure}%
	\begin{subfigure}{0.33\textwidth}
		\centering
		\includegraphics[scale=0.675]{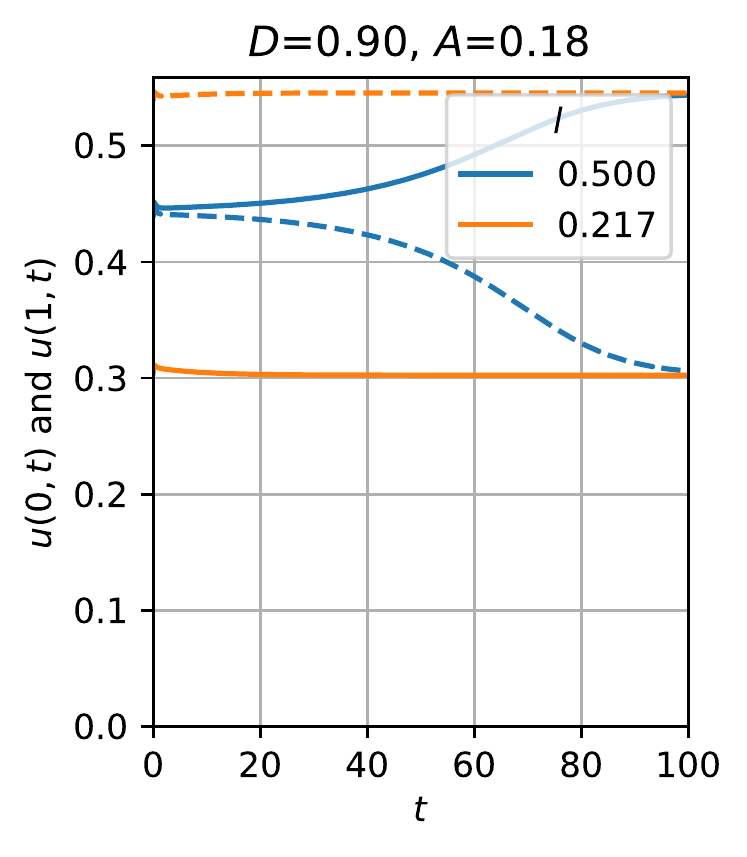}
		\caption{}\label{fig:bot_bb_ba_sim_3}
	\end{subfigure}%
	\caption{Results of numerical simulation of \eqref{eq:pde_gm} using FlexPDE 6 \cite{flexpde} with $\varepsilon=0.005$, $\tau=0.1$, and select values of $D$ and $A$. In each plot the solid (resp. dashed) lines correspond tot he spike height at $x=0$ (resp. $x=1$). Both the asymptotically constructed symmetric ($l=1/2$) and asymmetric $0<l<1/2$ solutions were used as initial conditions. See Figure \ref{fig:BOT_BB_BA_colormap} for position of parameter values relative to existence and stability thresholds.}\label{fig:bot_bb_ba_numerical}
\end{figure}

Finally we support our asymptotic predictions by numerically solving \eqref{eq:pde_gm} using FlexPDE 6 \cite{flexpde} with parameters $\varepsilon=0.005$ and $\tau = 0.1$ for select values of $D$ and $A$. Letting $u_e$ and $v_e$ be the any of the symmetric or asymmetric two-boundary-spike patterns constructed above, we use
\begin{equation}\label{eq:flexpde-initial-condition}
u(x,0) = (1 + 0.025\cos(20x))u_e(x),\qquad v(x,0)=(1+0.025\cos 20x)v_e(x),
\end{equation}
as initial conditions and simulate \eqref{eq:pde_gm} sufficiently long that the solution settles. The results of our numerical simulations indicate good agreement with the asymptocally calculated linear stability thresholds for both symmetric and asymetric two-boundary-spike patterns. We include in Figure \ref{fig:bot_bb_ba_numerical} the results of our numerical calculations for select values of $D$, $l$, and $A$ indicated by black markers in Figure \ref{fig:BOT_BB_BA_colormap}. In particular, in Figure \ref{fig:bot_bb_ba_sim_1} we plot the spike heights at $x_L=0$ (solid) and $x_R=1$ (dashed) for $D=0.30$ and $A=0.02$ with initial conditions given by the two-boundary-spike pattern constructed with $l=0.5$ which is symmetric and predicted to be stable, as well as $l=0.220$ and $l=0.005$ which are both asymmetric but predicted to be linearly unstable and stable respectively. It is clear from the resulting plots that our asymptotic predictions hold in this numerical simulations. Additionally, we observe that the unstable asymmetric two-spike pattern tends toward the linearly stable pattern. We observed this trend for all our numerical simulations in which $l_\text{max}<l<1/2$ though predicting this long-time behaviour analytically is beyond the scope of this paper. Similarly we numerically simulate the dynamics of a symmetric and asymmetric two-boundary-spike pattern when $D=0.60$ and $A=0.08$ (Figure \ref{fig:bot_bb_ba_sim_2}) and when $D=0.90$ and $A=0.18$ (Figure \ref{fig:bot_bb_ba_sim_3}). In both cases the symmetric and asymmetric patterns are predicted to be linearly unstable and stable respectively, which agrees with the outcomes observed in our numerical simulations.

\subsection{Example 3: Two Boundary Spikes with a One Sided Flux ($A\geq 0$ and $B=0$)}

\begin{figure}[t!]
	\centering
	\begin{subfigure}{0.34\textwidth}
		\centering
		\includegraphics[scale=0.65]{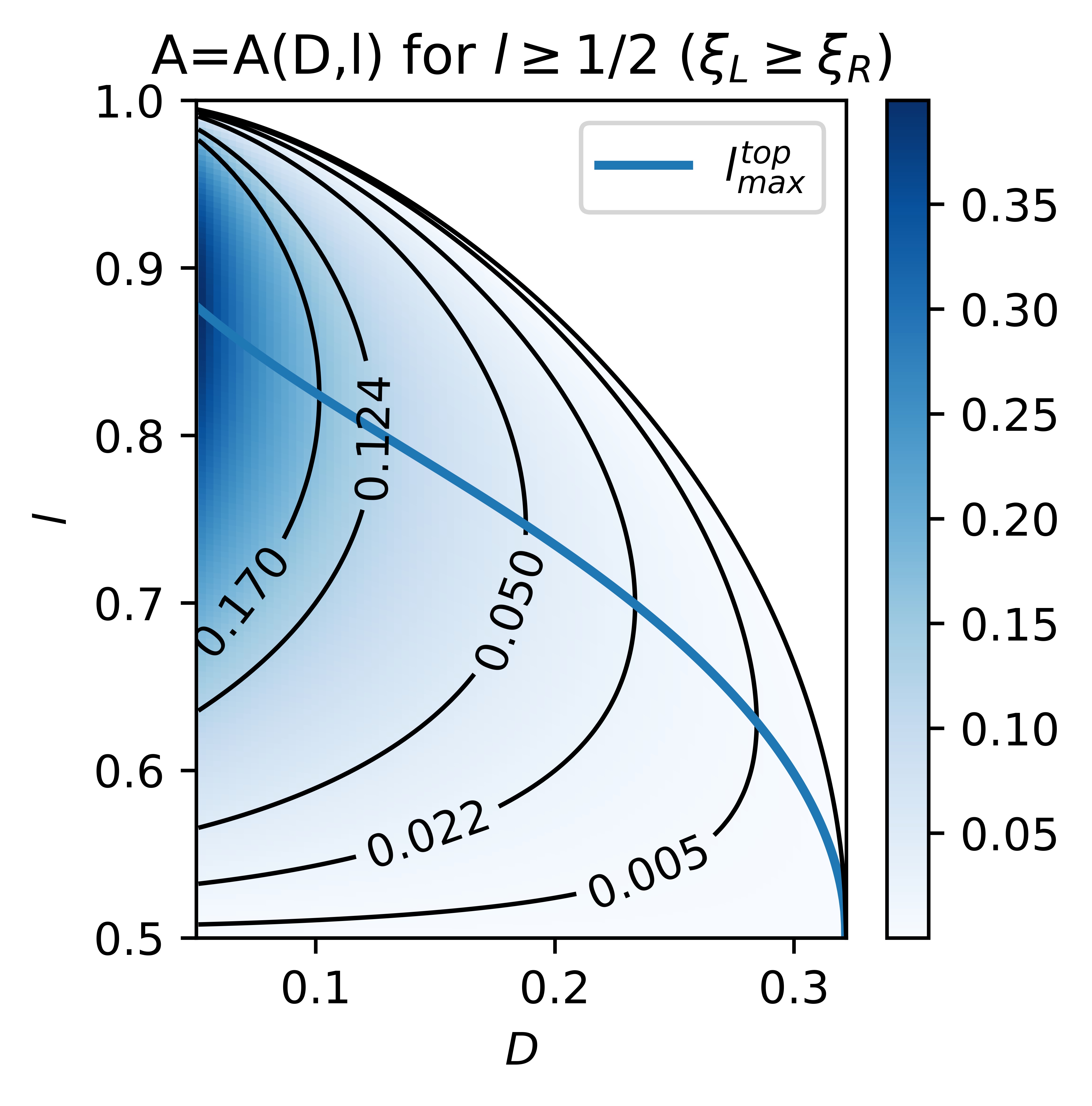}
		\caption{}\label{fig:TOP_BB_B0_color}
	\end{subfigure}%
	\begin{subfigure}{0.34\textwidth}
		\centering
		\includegraphics[scale=0.65]{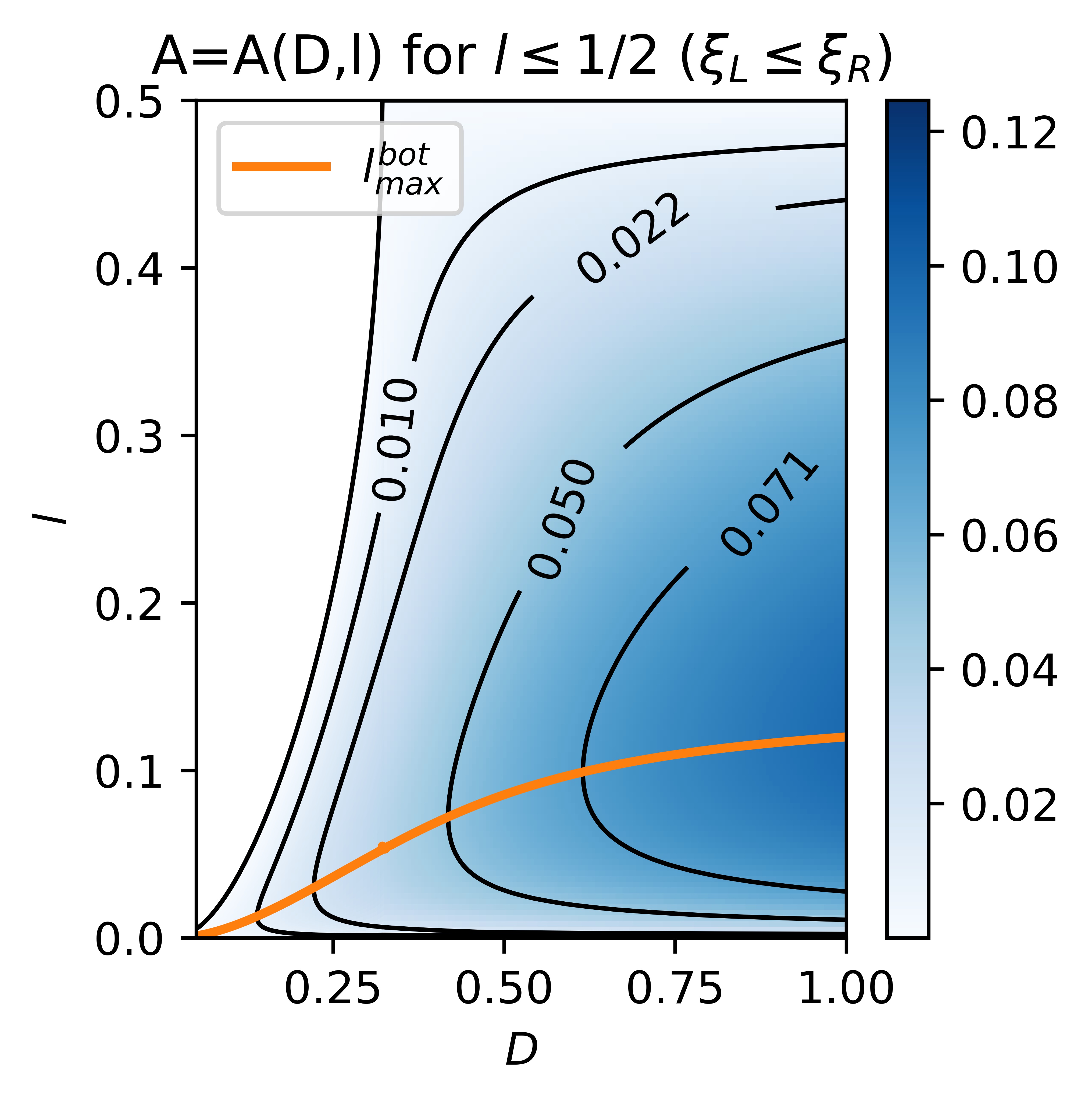}
		\caption{}\label{fig:BOT_BB_B0_color}
	\end{subfigure}%
	\begin{subfigure}{0.32\textwidth}
		\centering
		\includegraphics[scale=0.65]{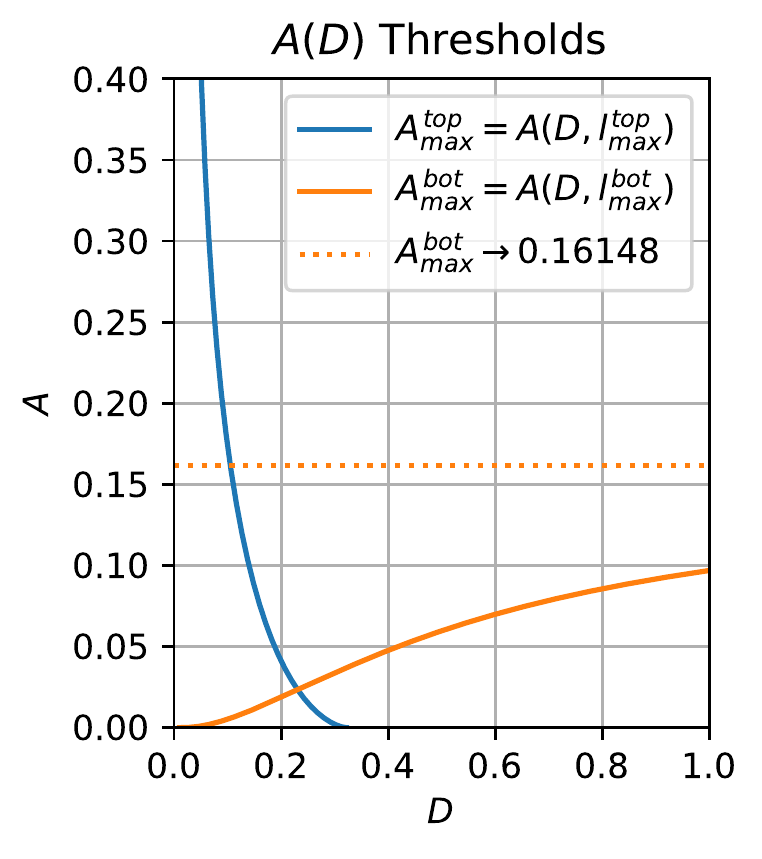}
		\caption{}\label{fig:BB_B0_thresholds}
	\end{subfigure}%
	\caption{Plot of $A=A(D,l)$ for Example 3 obtained by solving \eqref{eq:two_bdry_spike_example_3} when (A) $0<l<1/2$ and (B) $1/2<l<1$. The solid curves $l=l_\text{max}^\text{top}(D)$ and $l_\text{max}^\text{bot}(D)$ indicate the values of $l$ at which $A(D,l)$ is maximized as well as the competition instability threshold in the $l>1/2$ and $l<1/2$ regions respectively. The corresponding existence thresholds of $A$ versus $D$ are plotted against $D$ in (C).}\label{fig:BB_B0}
\end{figure}

In this example we investigate the effect of a one sided boundary flux ($A\geq0$ and $B=0$) on the structure and linear stability of a two-boundary-spike pattern. Letting $l_L=l$ and $l_R=1-l$, the gluing equation \eqref{eq:gluing_equation_2} of \S\ref{subsec:gluing-method} becomes
\begin{equation}\label{eq:two_bdry_spike_example_3}
\frac{\tanh\omega_0l}{\eta(y_L)\cosh\omega_0 l} - \frac{\tanh\omega_0(1-l)}{3\cosh\omega_0(1-l)} = 0,
\end{equation}
which is to be solved for $0<l<1$ where $\eta(y_L)$ and $y_L=y_0(\tfrac{\omega_0A}{\tanh\omega_0})$ are given by \eqref{eq:eta_def} and \eqref{eq:gluing_parameters} respectively. Note that since $\eta(y_L)<3$ for all $A>0$ it follows that $l=1/2$ is a solution of \eqref{eq:two_bdry_spike_example_3} if and only if $A=0$. In particular $\xi_L\neq \xi_R$ for all $A>0$ and by the asymmetry of the boundary fluxes the cases $l\lessgtr 1/2$, for which $\xi_L\lessgtr\xi_R$, must be considered separately. On the other hand, when $A=0$ we apply the same results from \cite{ward_2002_asymmetric} summarized in Example 2.

Proceeding as in Example 2 we numerically solve \eqref{eq:two_bdry_spike_example_3} for $A=A(D,l)>0$ when $l<1/2$ and $l>1/2$. In addition we compute $l=l_\text{max}^\text{bot}(D)$ and $l=l_\text{max}^\text{top}(D)$ defined as the curves along which $A(D,l)$ is maximized in the regions $l<1/2$ and $l>1/2$ respectively. In Figures \ref{fig:TOP_BB_B0_color} and \ref{fig:BOT_BB_B0_color} we plot $A=A(D,l)$ together with $l_\text{max}^\text{top}$ and $l_\text{max}^\text{bot}$ in the regions $1/2<l<1$ and $0<l<1/2$ respectively. In each region the maximum value of $A$ given by $A_\text{max}^\text{top}(D) \equiv A(D,l_\text{max}^\text{top}(D))$ and $A_\text{max}^\text{bot}(D) \equiv A(D,l_\text{max}^\text{bot}(D)$ and plotted in Figure \ref{fig:BB_B0_thresholds} gives an existence threshold for the boundary flux beyond which no two-boundary-spike with $\xi_L>\xi_R$ and $\xi_L<\xi_R$ exists respectively. In particular, a two-boundary-spike pattern with $\xi_L>\xi_R$ only exists if $A<A_\text{max}^\text{top}(D)$ and $D$ satisfies \eqref{eq:A=0_asymmetric_threshold}, whereas a two-boundary-spike pattern with $\xi_L<\xi_R$ exists for all $D>0$ provided that $A<A_\text{max}^\text{bot}(D)$. Furthermore, by letting $D\rightarrow\infty$ in \eqref{eq:two_bdry_spike_example_3} we numerically calculate $l_\text{max}^\text{bot}\rightarrow 0.13772$ and $A_\text{max}^\text{bot}(D)\rightarrow 0.16148$ as $D\rightarrow\infty$ and this horizontal asymptote is indicated in Figure \ref{fig:BB_B0_thresholds}. 

Next we consider the linear stability of the two-boundary-spike patterns constructed above when $\tau=0$. In particular we restrict our attention to competition instabilities which arise through a zero eigenvalue crossing. Proceeding as in Example 2 we first deduce that both $l=l_\text{max}^\text{top}(D)$ and $l=l_\text{max}^\text{bot}(D)$ yield a competition instability threshold. Furthermore, we verify that these are the only competition instability thresholds by numerically computing the algebraic equation \eqref{eq:algebraic_eq} where $\mathscr{G}_{\omega_0}$ and $\mathscr{D}_0$ are given by \eqref{eq:symm_vector_matrix_def} and \eqref{eq:example_2_comp_det} respectively. Since all asymmetric two-boundary-spike patterns when $A=0$ are unstable with respect to competition instabilities (see Example 2 and Appendix \ref{app:A0_stability}), we immediately deduce that all asymmetric two-boundary spike patterns with $\xi_L>\xi_R$ and $\xi_L<\xi_R$ are linearly unstable when $l>l_\text{max}^\text{top}(D)$ and $l>l_\text{max}^\text{bot}(D)$ respectively, and are linearly stable otherwise. In particular, the non-zero boundary flux $A>0$ both extends the range of parameter values for which asymmetric patterns exists and are linearly stable.

\begin{figure}[t!]
	\centering
	\begin{subfigure}{0.34\textwidth}
		\centering
		\includegraphics[scale=0.65]{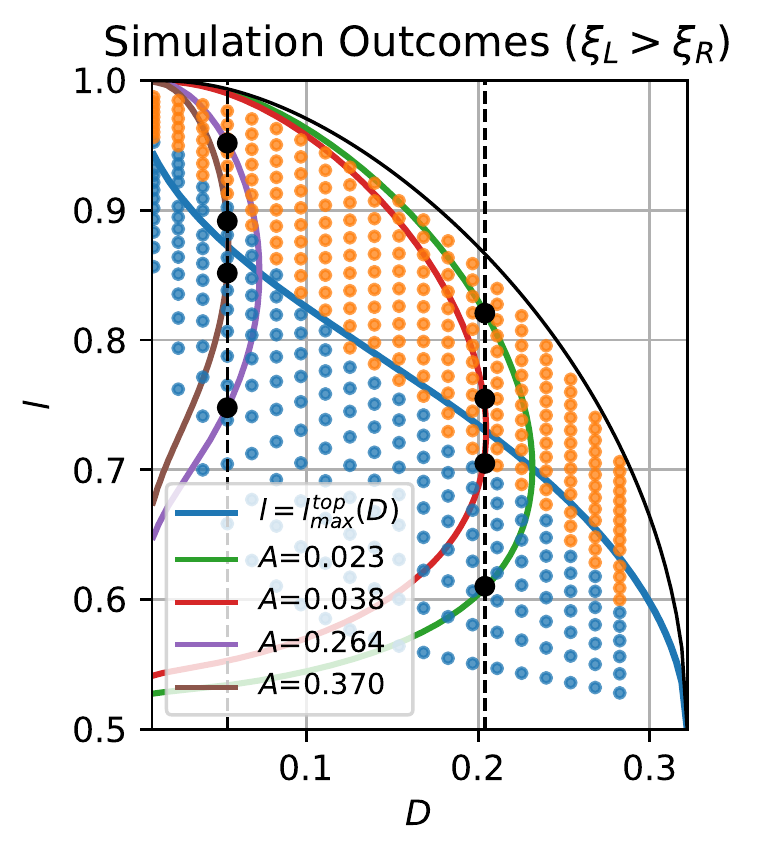}
		\caption{}\label{fig:top_bb_b0_simulations}
	\end{subfigure}%
	\begin{subfigure}{0.34\textwidth}
		\centering
		\includegraphics[scale=0.65]{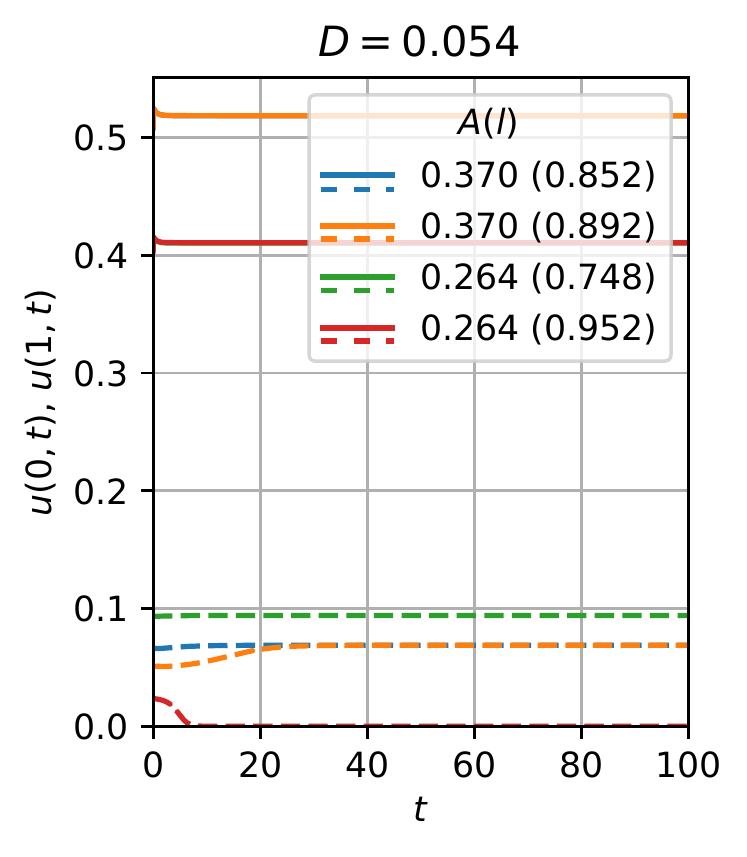}
		\caption{}\label{fig:top_bb_b0_simulations_0}
	\end{subfigure}%
	\begin{subfigure}{0.32\textwidth}
		\centering
		\includegraphics[scale=0.65]{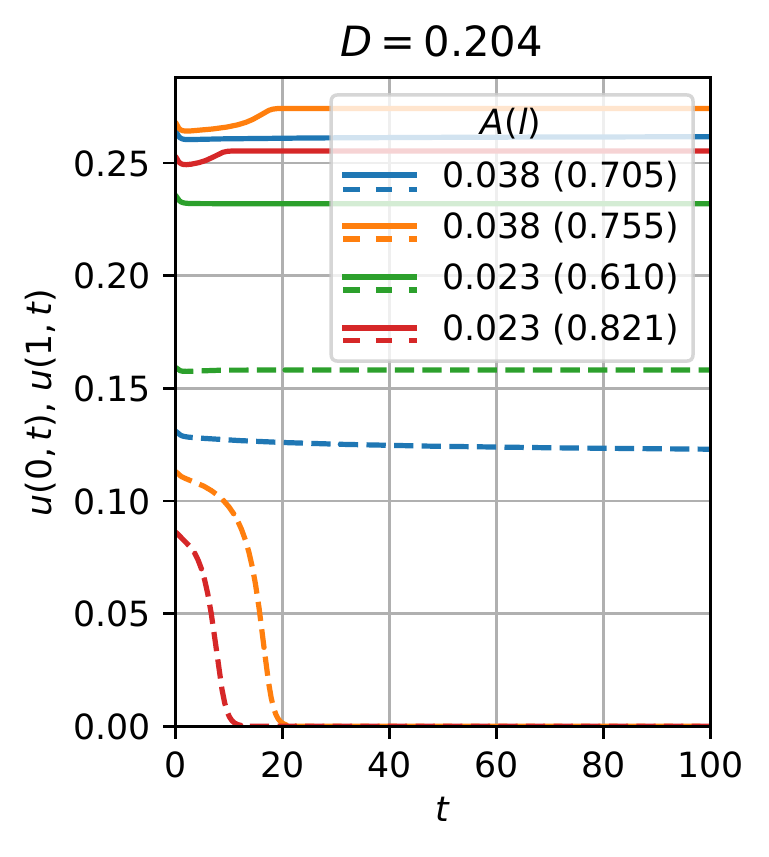}
		\caption{}\label{fig:top_bb_b0_simulations_1}
	\end{subfigure}%
	\caption{Numerical simulations for Example 2 when $\xi_L>\xi_R$. (A) Outcome of numerical simulation of \eqref{eq:pde_gm} starting from the asymmetric two-boundary-spike pattern constructed using the indicated values of $D$, $l$, and $A$. Blue and orange markers indicate the two-boundary spike pattern settled to the stable two-spike pattern (i.e. with $l<l_\text{max}^\text{top}(D)$) or collapsed to a single spike pattern respectively. Black dots indicate values of $D$, $A$, and $l$ for which the spike heights are plotted over time in Figures (B) and (C). The left and right dashed vertical lines indicate $D=0.054$ and $D=0.204$ respectively. In (B) and (C) we plot spike heights at $x=0$ (solid) and $x=1$ (dashed) at given values of $D$ and $A$ and with initial condition specified by indicated value of $l$.}\label{fig:top_bb_b0_numerical}
\end{figure}

\begin{figure}[t!]
	\centering
	\begin{subfigure}{0.34\textwidth}
		\centering
		\includegraphics[scale=0.65]{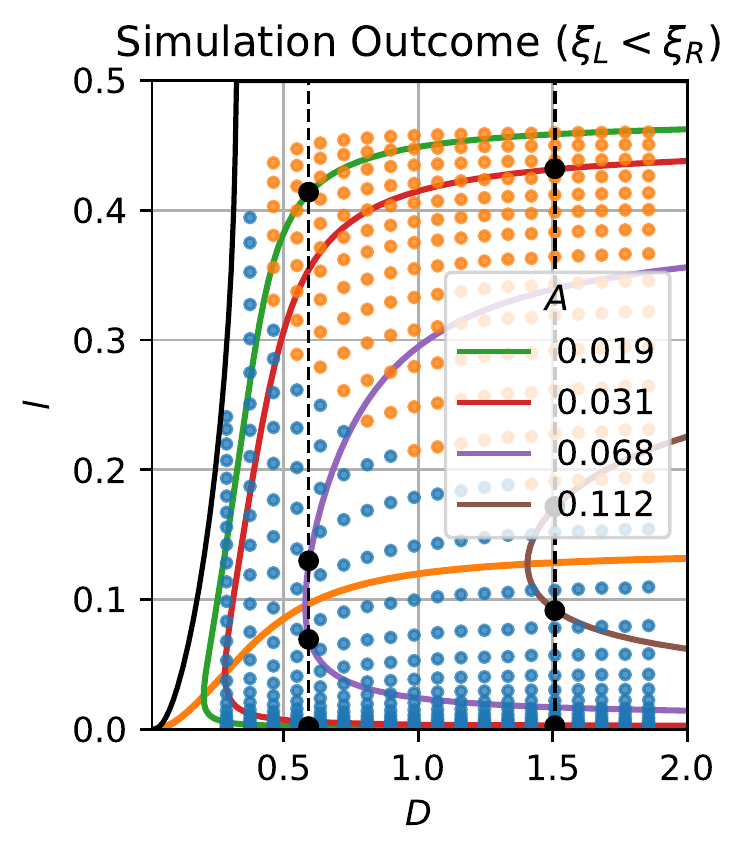}
		\caption{}\label{fig:bot_bb_b0_simulations}
	\end{subfigure}%
	\begin{subfigure}{0.34\textwidth}
		\centering
		\includegraphics[scale=0.65]{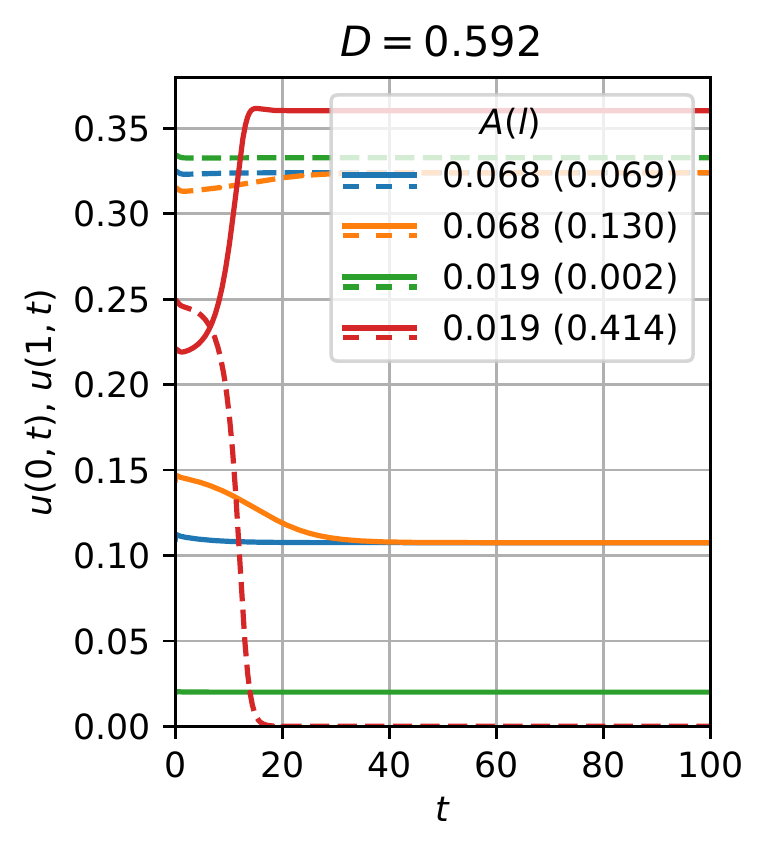}
		\caption{}\label{fig:bot_bb_b0_simulations_0}
	\end{subfigure}%
	\begin{subfigure}{0.32\textwidth}
		\centering
		\includegraphics[scale=0.65]{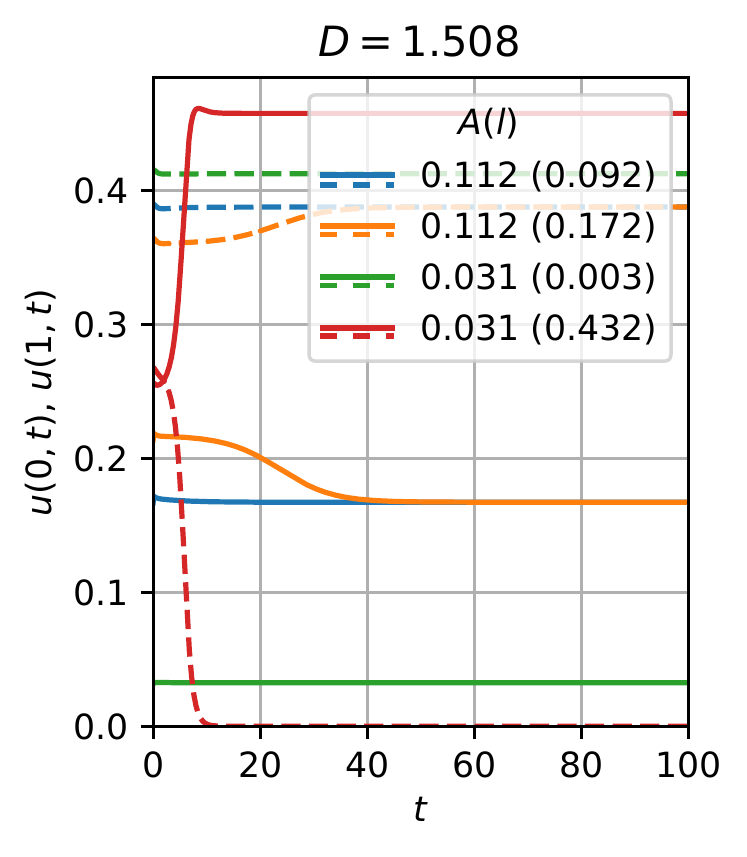}
		\caption{}\label{fig:bot_bb_b0_simulations_1}
	\end{subfigure}%
	\caption{Numerical simulations for Example 2 when $\xi_L<\xi_R$. (A) Outcome of numerical simulation of \eqref{eq:pde_gm} starting from the asymmetric two-boundary-spike pattern constructed using the indicated values of $D$, $l$, and $A$. Blue and orange markers indicate the two-boundary spike pattern settled to the stable two-spike pattern (i.e. with $l<l_\text{max}^\text{bot}(D)$) or collapsed to a single spike pattern respectively. Black dots indicate values of $D$, $A$, and $l$ for which the spike heights are plotted over time in Figures (B) and (C). The left and right dashed vertical lines indicate $D=0.592$ and $D=1.508$ respectively. In (B) and (C) we plot spike heights at $x=0$ (solid) and $x=1$ (dashed) at given values of $D$ and $A$ and with initial condition specified by indicated value of $l$.}\label{fig:bot_bb_b0_numerical}
\end{figure}

To support our asymptotic predictions we numerically calculate solutions to \eqref{eq:pde_gm} when $\varepsilon=0.005$ and $\tau=0.1$ for select values of $D$ and $A < \max\{A_\text{max}^\text{top}(D),A_\text{max}^\text{bot}(D)\}$. For each pair $(D,A)$ we let $0<l<1$ be any of the values for which $A(D,l)=A$ and then let $(u_e(x),v_e(x))$ be the corresponding equilibrium pattern constructed above. Using \eqref{eq:flexpde-initial-condition} as an initial condition we then solve \eqref{eq:pde_gm} numerically using FlexPDE 6 \cite{flexpde}. The results of our numerical simulations are illustrated in Figures \ref{fig:top_bb_b0_numerical} and \ref{fig:bot_bb_b0_numerical} when $l$ is chosen to be in $1/2<l<1$ and $0<l<1/2$ respectively. Specifically, in Figure \ref{fig:top_bb_b0_simulations} (resp. \ref{fig:bot_bb_b0_simulations}) we indicate with a blue or orange marker respectively whether the solution settles (after simulating for $0<t<200$) to the asymptotically predicted stable equilibrium with $l<l_\text{max}^\text{top}(D)$ (resp. $l<l_\text{max}^\text{bot}(D)$) or to a one-boundary spike solution in which the spike at $x=1$ collapses respectively. In Figures \ref{fig:top_bb_b0_simulations_0} and \ref{fig:top_bb_b0_simulations_1} (resp. \ref{fig:bot_bb_b0_simulations_0} and \ref{fig:bot_bb_b0_simulations_1}) we show the spike heights as functions of time at select values of $A$ and $D$ using an unstable and stable value of $l$ in $1/2<l<1$ (resp. $0<l<1/2$) to construct the initial condition (see captions for more details). The results of these numerical simulations are in good agreement with our asymptotic predictions. However, we comment that in the numerical outcomes shown in Figure \ref{fig:top_bb_b0_simulations} some of the asymmetric patterns which are predicted to be stable collapse. We expect that this error due to a combination of small errors from the asymptotic theory, numerical errors from the time integration of \eqref{eq:pde_gm}, as well as the close proximity to the fold point $A=A_\text{max}^\text{top}(D)$ for these values of $l$.

\subsection{Example 4: One Boundary and Interior Spike with One-Sided Feed ($A\geq 0$, $B=0$)}\label{sec:example_4}

In this final example we extend the results of Example 3 to the case where there is one boundary spike at $x_L=0$ and one interior spike at $0<x_1<1$. The asymptotic construction of the resulting two spike patterns, as well as the analysis of their linear stability on an $\mathcal{O}(1)$ timescale proceeds as in the previous example. However, since $x_1$ is an equilibrium of the slow dynamics equation \eqref{eq:ode_dynamics}, we must now also determine the stability of the two-spike pattern on an $\mathcal{O}(\varepsilon^{-2})$ timescale by analyzing the linearization of \eqref{eq:ode_dynamics}. We remark that an alternative approach to determine the stability of multi-spike patterns on an $\mathcal{O}(\varepsilon^{-2})$ timescale is to calculate the $\mathcal{O}(\varepsilon^2)$ eigenvalues of the linearization \eqref{eq:pde_sys_stability} though we do not pursue this approach further (see for example \cite{iron_2001} for the analysis of $\mathcal{O}(\varepsilon^2)$ eigenvalues).

Using the method of \S\ref{subsec:gluing-method}, with $l_L=l$ and $l_1=(1-l)/2$ equation \eqref{eq:gluing_equation_2} becomes
\begin{equation}\label{eq:example_4_asy_eq}
\frac{\tanh\omega_0 l}{\eta(y_L)\cosh\omega_0l} - \frac{\tanh\omega_0\tfrac{1-l}{2}}{3\cosh\omega_0\tfrac{1-l}{2}} = 0,
\end{equation}
which is to be solved for $0<l<1$ where $\eta(y_L)$ and $y_L=y_0(\omega_0 A/\tanh\omega_0l)$ are given by \eqref{eq:eta_def} and \eqref{eq:gluing_parameters} respectively. As in the previous examples we observe that $\xi_L=\xi_1$ if and only if $l=1/3$ which is a solution to \eqref{eq:example_4_asy_eq} if and only if $A=0$. In particular, together with the discussion in Example 3 we deduce that there are no symmetric two-spike patterns when there is a one-sided positive boundary flux $A>0$. As in the previous examples we also note that $\xi_L\lessgtr\xi_1$ if $l\lessgtr 1/3$. Next we note that since $l_1=(1-l)/2$ the relevant asymmetric equilibrium results for $A=0$ from \cite{ward_2002_asymmetric} summarized in Example 2 must be modified. In particular, letting $z=\omega_0 l_L$ and $\tilde{z} = \omega_0 l_1$, equations \eqref{eq:gluing_equation_1} and \eqref{eq:gluing_equation_2} when $A=0$ become
\begin{equation}
z + 2\tilde{z} = \omega_0,\qquad b(z)=b(\tilde{z}),
\end{equation}
where $b(z)$ is given in \eqref{eq:bdry_bdry_A0}. Then Result 2.3 of \cite{ward_2002_asymmetric} (with $k_1=1$ and $k_2=2$) implies that a unique asymmetric two-spike solution with $z\leq\tilde{z}$ exists if and only if
\begin{equation}
D < D_m \equiv [3 \log(1 + \sqrt{2})]^{-2}\approx 0.143,
\end{equation}
whereas Result 2.4 of \cite{ward_2002_asymmetric} (with $k_1=2$ and $k_2=1$) implies that there are either exactly one or two asymmetric two-spike solutions with $z>\tilde{z}$ if and only if
\begin{equation}
D<D_{m}\quad\text{or}\quad D_m<D<D_{m1}\equiv [2\sinh^{-1}(1/2) + \sinh^{-1}(2)]^{-2} \approx 0.17274,
\end{equation}
respectively. Proceeding as in Examples 2 and 3 we can then numerically calculate $A=A(D,l)$ from \eqref{eq:example_4_asy_eq} in the appropriate regions with $0<l<1/3$ and $1/3<l<1$. In Figures \ref{fig:TOP_BI_B0_color} and \ref{fig:BOT_BI_B0_color} we plot $A=A(D,l)$ together with the curves $l=l_\text{max}^\text{top}(D)$ and $l=l_\text{max}^\text{bot}(D)$ along which $A(D,l)$ is maximized. The resulting existence thresholds $A_\text{max}^\text{top}(D) \equiv A(D,l_\text{max}^\text{top}(D))$ and $A_\text{max}^\text{bot}(D) \equiv A(D,l_\text{max}^\text{bot}(D))$ are plotted in Figure \ref{fig:BI_B0_thresholds}. In particular a two spike pattern with a two-spike pattern consisting of one boundary and one interior spike with $\xi_L>\xi_1$ only exists for $D<D_{m1}$ when $A<A_\text{max}^\text{top}(D)$, whereas such a two-spike pattern with $\xi_L<\xi_1$ exists for all $D>0$ provided that $A<A_\text{max}^\text{bot}(D)$. Finally, by taking the limit $D\rightarrow\infty$ in \eqref{eq:example_4_asy_eq} we numerically obtain the limiting values $l_\text{max}^\text{bot}(D)\rightarrow0.0857$ and $A_\text{max}^\text{bot}(D)\rightarrow 0.087174$ as $D\rightarrow\infty$.

\begin{figure}[t!]
	\centering
	\begin{subfigure}{0.34\textwidth}
		\centering
		\includegraphics[scale=0.65]{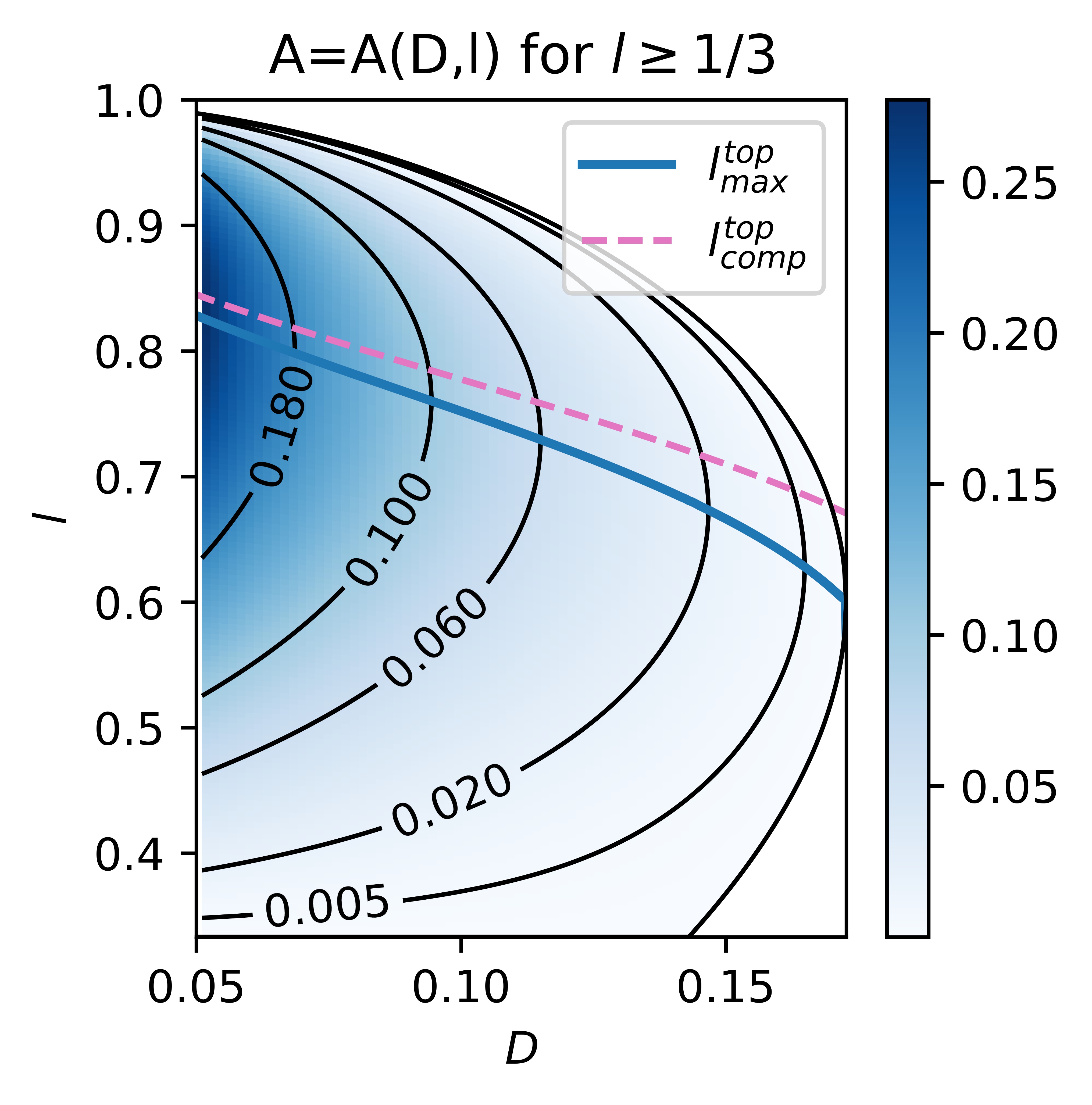}
		\caption{}\label{fig:TOP_BI_B0_color}
	\end{subfigure}%
	\begin{subfigure}{0.34\textwidth}
		\centering
		\includegraphics[scale=0.65]{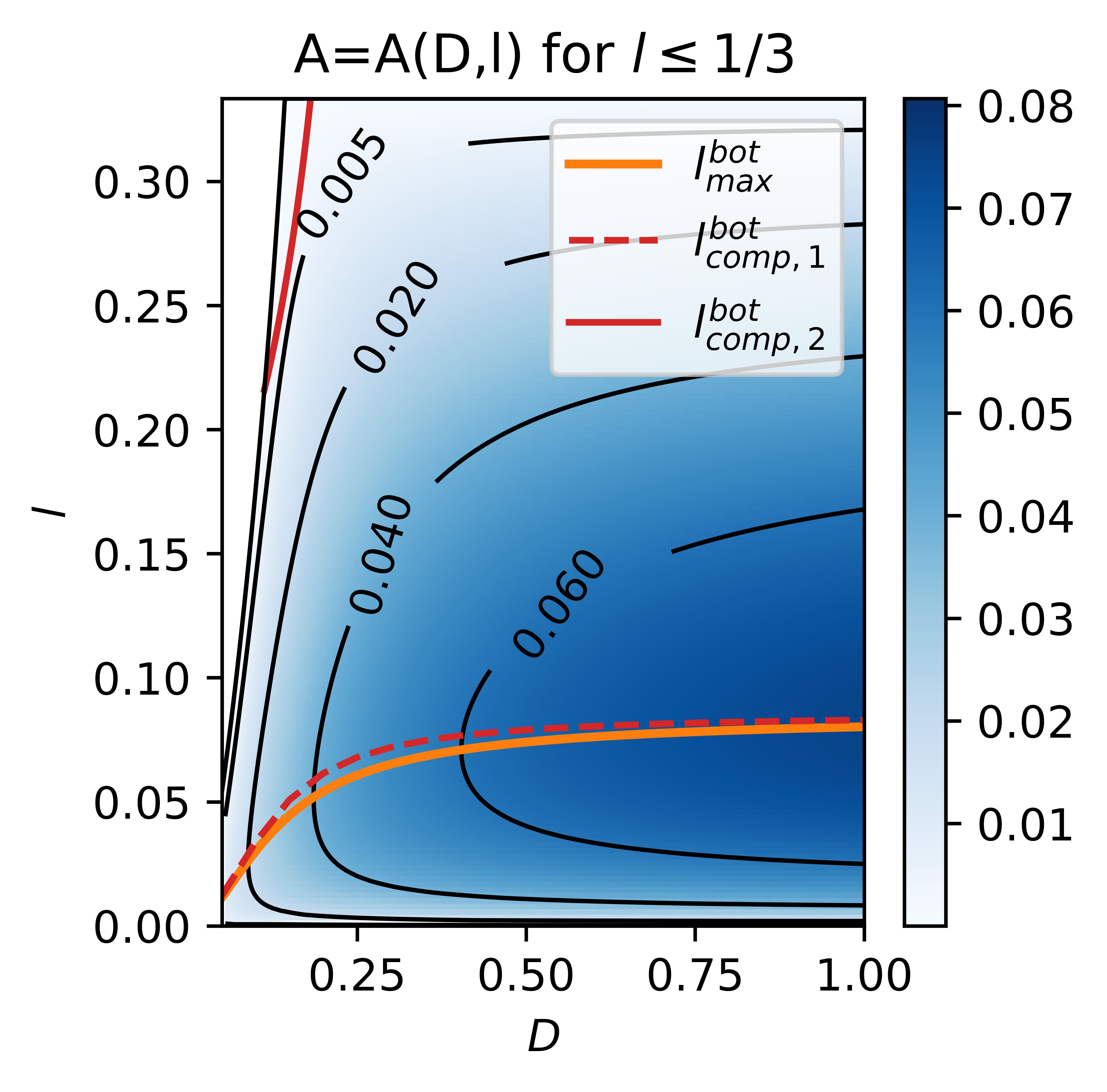}
		\caption{}\label{fig:BOT_BI_B0_color}
	\end{subfigure}%
	\begin{subfigure}{0.32\textwidth}
		\centering
		\includegraphics[scale=0.65]{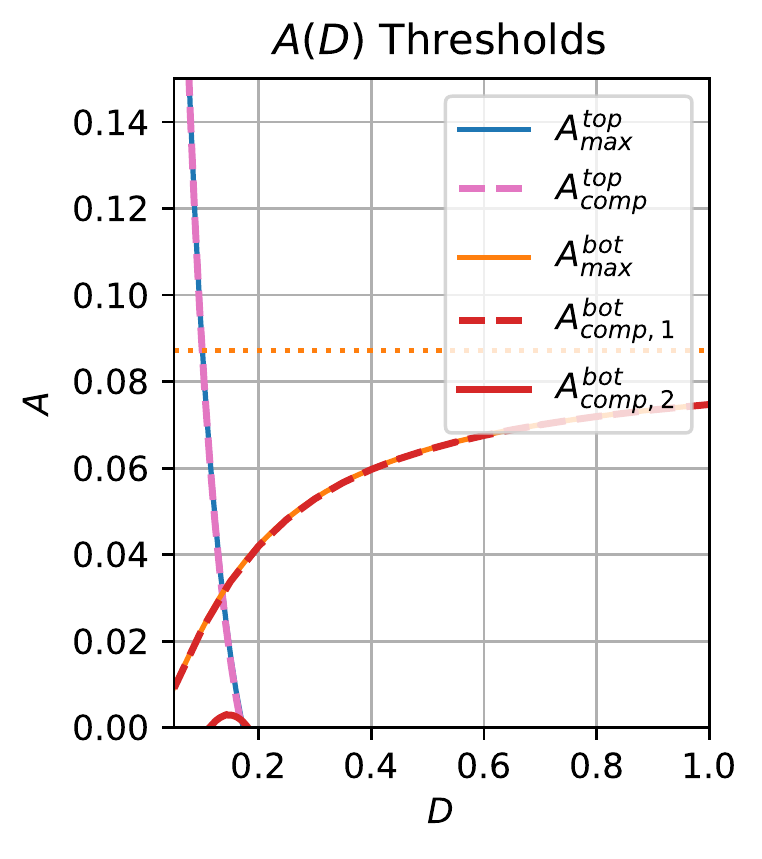}
		\caption{}\label{fig:BI_B0_thresholds}
	\end{subfigure}%
	\caption{Plot of $A=A(D,l)$ obtained by solving \eqref{eq:two_bdry_spike_example_3} when (A) $0<l<1/3$ and (B) $1/3<l<1$. The curves $l=l_\text{max}^\text{top}(D)$ and $l_\text{max}^\text{bot}(D)$ indicate the values of $l$ at which $A(D,l)$ is maximized while the curves $l_\text{comp}^\text{top}(D)$, $l_\text{comp,1}^\text{bot}(D)$, and $l_\text{comp,2}^\text{bot}(D)$ indicate the competition instability thresholds. The corresponding existence thresholds, $A_\text{max}^\text{top}(D)$ and $A_\text{max}^\text{bot}(D)$ are plotted in (C).}\label{fig:BI_B0}
\end{figure}

Next we consider the linear stability of the two-spike patterns constructed above on an $\mathcal{O}(1)$ timescale. As in Examples 2 and 3 we focus exclusively on competition instabilities by assuming that $\tau=0$ and seeking a zero eigenvalue crossing of the NLEP \eqref{eq:nlep}. In contrast to Examples 2 and 3 above the relevant competition instability threshold does not necessarily coincide with the curves $l=l_\text{max}^\text{top}(D)$ and $l=l_\text{max}^\text{bot}$. Indeed, fixing $D>0$ and differentiating the quasi-equilibrium equation $\bm{B}=0$ with respect to $l$ gives
\begin{equation}\label{eq:bdry_intr_diff_quasi}
\nabla_{\pmb{\xi}} \pmb{B} \biggl(\frac{\partial\pmb{\xi}}{\partial l} + \frac{\partial\pmb{\xi}}{\partial A}\frac{\partial A}{\partial l}\biggr) + \frac{\partial\pmb{B}}{\partial A} \frac{\partial A}{\partial l} + \frac{\partial\pmb{B}}{\partial x_1}\frac{d x_1}{d l} = 0,
\end{equation}
which along either $l=l_\text{max}^\text{top}(D)$ or $l=l_\text{max}^\text{bot}(D)$ reduces to
\begin{equation*}
\nabla_{\bm{\xi}}\bm{B}\frac{\partial\bm{\xi}}{\partial l} = -\frac{1}{2}\frac{\partial \bm{B}}{\partial x_1}.
\end{equation*}
Since $\partial\bm{B}/\partial x_1\neq \bm{0}$ it follows that $\nabla_{\bm{\xi}}\bm{B}$ is not necessarily singular along $l_\text{max}^\text{top}(D)$ or $l_\text{max}^\text{bot}(D)$ and in particular these curves need not coincide with the competition instability thresholds. Note however that $\partial\bm{B}/\partial x_1\rightarrow 0$ as $D\rightarrow\infty$ and therefore the competition instability threshold will coincide with $l_\text{max}^\text{bot}(D)$ but only in the limit $D\rightarrow\infty$. Thus, to calculate the appropriate competition instability thresholds we use the algebraic reduction of \S\ref{subsec:algebraic} and numerically solve \eqref{eq:algebraic_eq} when $\lambda=0$ and where the matrices $\mathscr{G}_{\omega_0}$ and $\mathscr{D}_{\omega_0}$ are respectively given by
\begin{subequations}
\begin{equation}
\mathscr{G}_{\omega_0} = \frac{1}{\omega_0\sinh\omega_0}\begin{pmatrix} \cosh\omega_0 & \cosh\omega_0(1-x_1) \\ \text{cosh}\omega_0(1-x_1) & \cosh\omega_0x_1\cosh\omega_0(1-x_1) \end{pmatrix},
\end{equation}
and
\begin{equation}
\mathscr{D}_0 = \frac{1}{\omega_0}\begin{pmatrix} \tanh(\omega_0l) \mathscr{F}_{y_L}(0) & 0 \\ 0 &  2\tanh\bigl(\omega_0\frac{1-l}{2}\bigr)\mathscr{F}_0(0)\end{pmatrix}.
\end{equation}
\end{subequations}
The resulting competition instability threshold $l_\text{comp}^\text{top}(D)$ when $\xi_L>\xi_1$ as well as $l_\text{comp,$i$}^\text{bot}(D)$ ($i=1,2$) when $\xi_L<\xi_1$ are indicated in Figures \ref{fig:TOP_BI_B0_color} and \ref{fig:BOT_BI_B0_color} respectively. Additionally, in Figure \ref{fig:BI_B0_thresholds} we have plotted $A_\text{comp}^\text{top}(D) \equiv A(D,l_\text{comp}^\text{top}(D))$ and $A_\text{comp,$i$}^\text{top}(D) \equiv A(D,l_\text{comp,$i$}^\text{top}(D))$ ($i=1,2$). Finally, the stability result along the $A=0$ curve calculated in Appendix \ref{app:A0_stability} implies that the two-spike pattern with $\xi_L>\xi_1$ is linearly unstable on an $\mathcal{O}(1)$ timescale when $l>l_\text{comp}^\text{top}(D)$, and similarly when $\xi_L<\xi_1$ the two-spike pattern is linearly unstable in the region bounded by the curves $l=l_\text{comp,1}^\text{bot}(D)$ and $l=l_\text{comp,2}^\text{bot}(D)$.

Next we consider the linear stability of the two-spike patterns on an $\mathcal{O}(\varepsilon^{-2})$ timescale. We explicitly calculate the right-hand-side of \eqref{eq:ode_dynamics} by first calculating
\begin{equation*}
\big\langle \partial_xG_{\omega_0}(x,x_1)\big\rangle_{x_1}= \frac{\sinh\omega_0(2x_1-1)}{\sinh\omega_0},\quad \partial_x G_{\omega_0}(x,0)\bigr|_{x_1} = -\frac{\sinh\omega_0(1-x_1)}{\sinh\omega_0},
\end{equation*}
and rearranging the quasi-equilibrium equation \eqref{eq:sys_1} as
\begin{equation*}
\frac{\omega_0^2\xi_L^2\eta(y_L)}{\xi_1} = \frac{1-6\omega_0^2\xi_1 G_{\omega_0}(x_1,x_1)}{G_{\omega_0}(x_1,0)}.
\end{equation*}
so that \eqref{eq:ode_dynamics} becomes
\begin{equation}
\frac{1}{\varepsilon^2}\frac{d x_1}{dt} = -6\omega_0^2 f(x_1),\qquad f(x_1,\xi_1) = \xi_1 - \frac{\tanh\omega_0(1-x_1)}{3\omega_0},
\end{equation}
where $\xi_1$ together with $\xi_L$ and $y_L$ are functions of $x_1$ found by solving \eqref{eq:sys_1} and \eqref{eq:yL_eq_1}. Note that the asymmetric two-spike equilibrium solutions constructed above using the method of \S\ref{subsec:gluing-method} immediately satisfy $f(x_1) = 0$. The linear stability of these asymmetric two-spike patterns on an $\mathcal{O}(\varepsilon^{-2})$ timescale is determined by the sign of $f'(x_1)$; it is stable if $f'(x_1)>0$ and unstable otherwise. We explicitly calculate
\begin{equation}
\frac{d f}{d x_1} = \frac{\partial\xi_1}{\partial x_1} + \frac{1}{3}\sech^2\omega_0(1-x_1),
\end{equation}
where $\partial\xi_1/\partial x_1$ is calculated by first differentiating the quasi-equilibrium equation $\bm{B} = 0$ with respect to $x_1$
\begin{equation}\label{eq:bdry_intr_temp_1}
\nabla_{\bm{\xi}}\bm{B}\frac{\partial\bm{\xi}}{\partial x_1} = -\frac{\partial\bm{B}}{\partial x_1},
\end{equation}
and then solving for $\partial \bm{\xi}/\partial x_1$ which we can do since we are assuming the two-spike pattern is stable on an $\mathcal{O}(1)$ timescale and the matrix $\nabla_{\bm{\xi}}\bm{B}$ is therefore invertible. Numerically evaluating $f'(x_1)$ we find that the drift instability thresholds for which $f'(x_1)=0$ coincide with the curves $l_\text{max}^\text{top}(D)$ and $l_\text{max}^\text{bot}(D)$. In fact, we can show that this is the case analytically by first evaluating \eqref{eq:bdry_intr_diff_quasi} along either $l_\text{max}^\text{top}(D)$ or $l_\text{max}^\text{bot}(D)$ to get
\begin{equation}\label{eq:bdry_intr_temp_2}
\nabla_{\bm{\xi}}\bm{B} \frac{\partial\bm{\xi}}{\partial l} = - \frac{1}{2}\frac{\partial\bm{B}}{\partial x_1}.
\end{equation}
Since the competition instability thresholds do not coincide with the curves $l_\text{max}^\text{top}(D)$ and $l_\text{max}^\text{bot}(D)$, the matrix $\nabla_{\bm{\xi}}\bm{B}$ is invertible along these curves and comparing \eqref{eq:bdry_intr_temp_1} with \eqref{eq:bdry_intr_temp_2} we obtain
\begin{equation}
\frac{\partial \xi_1}{\partial x_1} = 2\frac{\partial\xi}{\partial l} = -\frac{1}{3}\sech^2\omega_0\frac{1-l}{2} = -\frac{1}{3}\sech^2\omega_0(1-x_1).
\end{equation}
In particular $f'(x_1)=0$ along the curves $l=l_\text{max}^\text{top}(D)$ and $l_\text{max}^\text{bot}(D)$. Numerically evaluating $f'(x_1)$ at select values of $l$ above and below these thresholds we determine that the two-spike patterns constructed above with $\xi_L>\xi_1$ or $\xi_L<\xi_1$ are linearly stable on an $\mathcal{O}(\varepsilon^{-2})$ timescale if and only if $l < l_\text{max}^\text{top}(D)$ or $l<l_\text{max}^\text{bot}(D)$ respectively.

As in the previous examples we performed full numerical simulations of \eqref{eq:pde_gm} with FlexPDE 6 \cite{flexpde} to support our asymptotic predictions. Our numerical simulations were found to strongly agree with the predicted stability thresholds. In particular, we observed the following dynamics. For values of $l$ that are stable with respect to both competition and drift instabilities, that is when $l<l_\text{max}^\text{top}$ (resp. $l<l_\text{max}^\text{bot}$) for $l>1/3$ (resp. $l<1/2$) the two-spike pattern was observed to be stable. In the remaining regions (both stable and unstable with respect to competition instabilities) we observed that the interior spike either collapses and the boundary spike collapses to the one-boundary-spike solution, or else the interior spike changes height to the height of the stable pattern and then slowly drifts toward the location of the interior spike in the stable two-spike solution. As in Example 3 we noticed sensitivity to the competition instability threshold which we believe to be primarily due to the flatness of $A$ in this region

\section{Discussion}\label{sec:discussion}


We have extended the asymptotic theory developed for the singularly perturbed one-dimensional GM model to include the possibility of inhomogeneous Neumann boundary conditions for the activator. Additionally, we have rigorously established partial stability and instability results for a class of \textit{shifted} NLEPs. While the shifted NLEPs we considered are closely related to those in  \cite{maini_2007} we highlight that the difference in sign of the shift parameter leads substantial differences in the stability properties of the NLEP. Finally we considered four examples to illustrate the asymptotic and rigorous theory as well as to explore the behaviour of the GM system with non-zero flux boundary conditions. For a one-boundary spike solution we found that the non-zero Neumann boundary condition improves the stability with respect to Hopf instabilities. Moreover, by considering a two-boundary-spike pattern with equal inhomogeneous boundary fluxes we illustrated that the non-zero boundary flux improves the stability of symmetric two-spike patterns and also extends the region of $D>0$ values for which asymmetric patterns exist provided that $A=B>0$ is not larger than a computed threshold. Similar results were obtained when considering a one-sided boundary flux for which we considered patterns where both spikes concentrate on the boundary and where one concentrates on the boundary and the other in the interior. In each of our two-spike pattern examples we observed that there are two asymmetric patterns, where one is always linearly unstable and the other is always linearly stable. In a sense, the stable asymmetric pattern can be considered a \textit{boundary layer} solution that is a direct consequence of the inhomogeneous Neumann boundary condition. In particular its existence is mandated by the inhomogeneous boundary condition which makes a direct comparison with asymmetric spike patterns in the absence of boundary flux conditions difficult. However, our results illustrate that inhomogeneities at the boundaries predispose the GM to forming patterns concentrating at the boundaries in both symmetric and asymmetric configurations. We believe the distinction between interior and boundary-layer like localized pattern will play a key role in understanding more complicated mathematical models such as those incorporating bulk-surface coupling (see Figure 3 in \cite{madzvamuse_2015} for an example of a boundary-layer type pattern in a bulk-surface model).

There are several key open problems and directions for future research. First, our rigorous results for the shifted NLEP do not provide tight bounds for regions of stability and instability. Specifically, in the small shift-case we have determined that the NLEP is unstable if $\mu<\mu_c(y_0)$, and stable if $\mu_1(y_0)<\mu<\mu_2(y_0)$ where we have highlighted that $\mu_c(y_0) < \mu_1(y_0)$. As indicated in \S\ref{sec:rigorous} we conjecture that in fact the shifted NLEP is stable for all $\mu>\mu_c$ and in Appendix \ref{app:conjecture} we provide numerical support for this conjecture. Proving this conjecture is our first open problem. In addition, to calculate the stability of asymmetric patterns for which the shift parameters are different we could not directly use the rigorous results established in \S\ref{sec:rigorous} since the NLEP \eqref{eq:nlep} could not be diagonalized. The development of a rigorous stability theory for NLEP systems of this form is an additional direction for future research.

One of the key insights from our investigation of a two-boundary spike configuration is that the presence of equal or one sided boundary fluxes for the activator greatly extends the range of diffusivity values for which asymmetric patterns exist and are linearly stable. This expanded region of existence and stability parallels that found when spatially inhomogeneous precursors are included in the GM model. However, it can be argued that introducing inhomogeneous flux conditions provides a simpler alternative for generating asymmetric patterns. This warrants further research into the role of inhomogeneous boundary conditions for the activator in both activator-inhibitor and activator-substrate models in one-, two-, and three-dimensional domains.

\appendix
\section{Large $\lambda_I$ Asymptotics of $\mathscr{F}_{y_0}(i\lambda_I)$}\label{app:F_properties}

\begin{figure}[t!]
	\centering
	\begin{subfigure}{0.33\textwidth}
		\centering
		\includegraphics[scale=0.675]{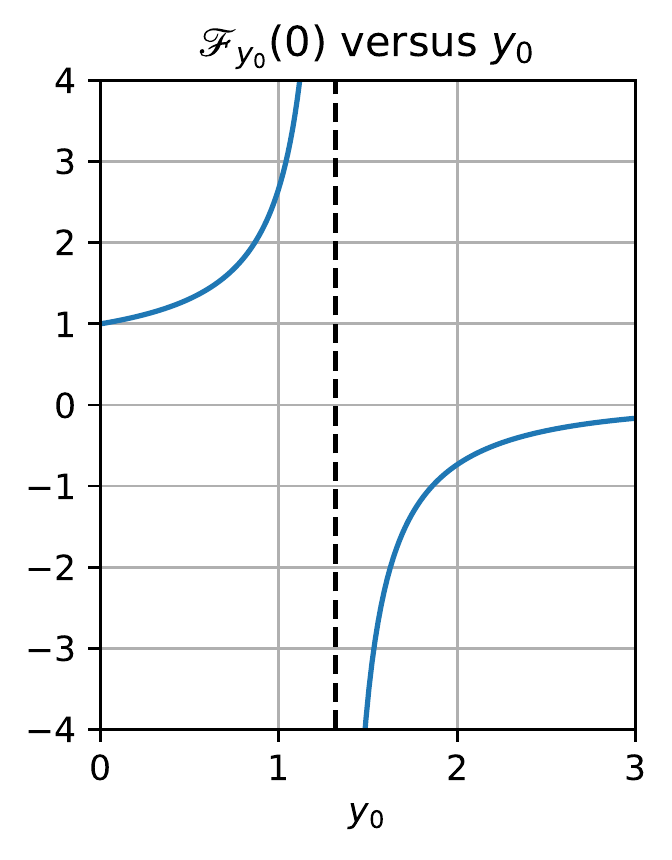}
		\caption{}\label{fig:F0}
	\end{subfigure}%
	\begin{subfigure}{0.33\textwidth}
		\centering
		\includegraphics[scale=0.675]{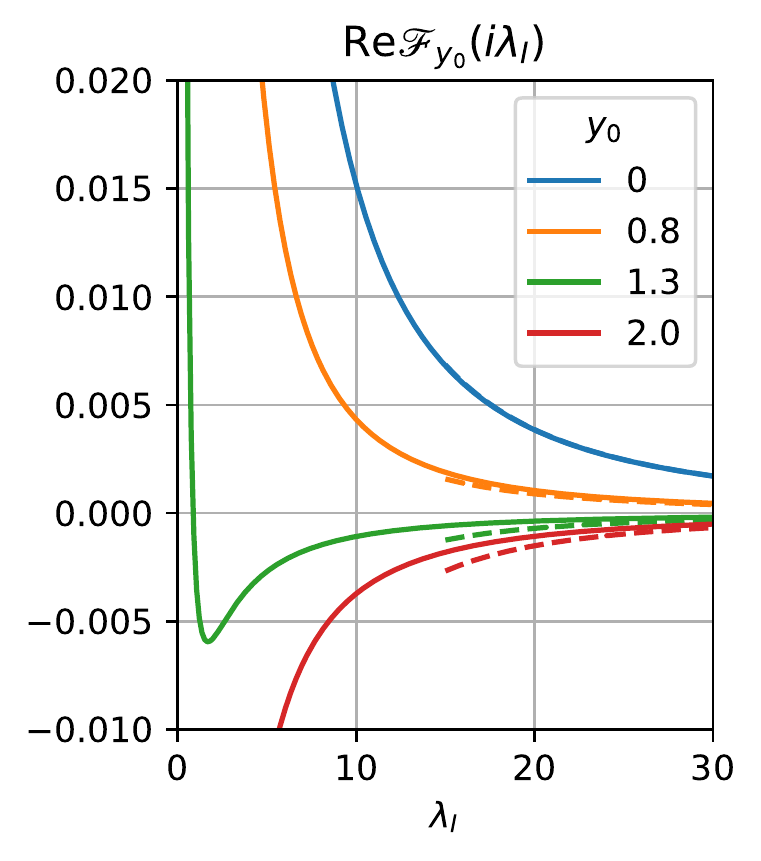}
		\caption{}\label{fig:real_F_imag}
	\end{subfigure}%
	\begin{subfigure}{0.33\textwidth}
		\centering
		\includegraphics[scale=0.675]{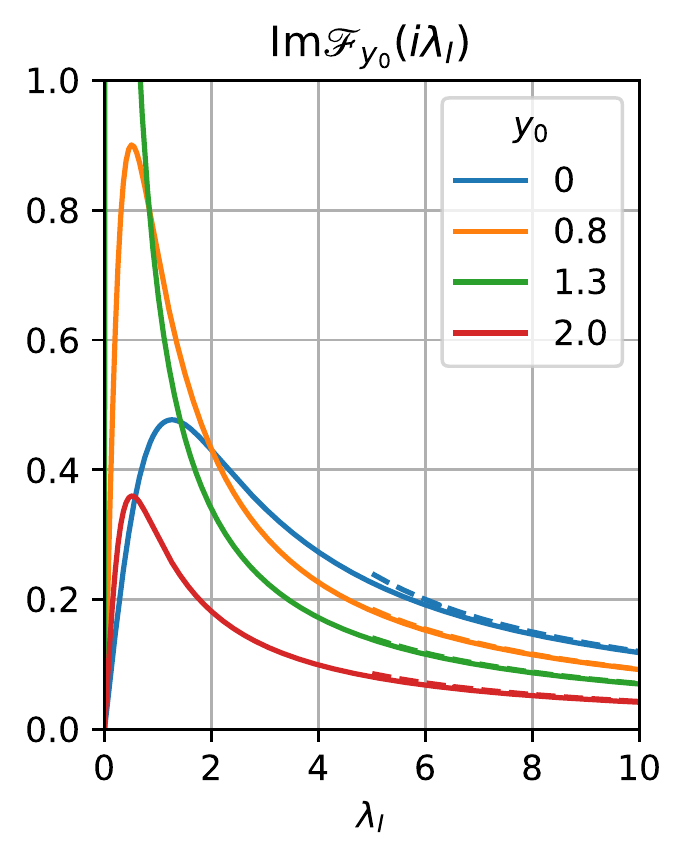}
		\caption{}\label{fig:imag_F_imag}
	\end{subfigure}%
	\caption{(A) Plot of $\mathscr{F}_{y_0}(0)$ versus the shift parameter $y_0$. (B) and (C) Real and imaginary parts of $\mathscr{F}_{y_0}(i\lambda_I)$ for select values of $y_0\geq 0$. The dashed lines indicate the $\lambda_I\gg1$ asymptotics. }\label{fig:F_properties}
\end{figure}

In this appendix we determine some key properties of $\mathscr{F}_{y_0}(\lambda)$ defined in \eqref{eq:F_y0_def}. Recalling \eqref{eq:Fy00}, in Figure \ref{fig:F0} we plot $\mathscr{F}_{y_0}(0)$ versus $y_0\geq 0$. Next we calculate the limiting behaviour of $\mathscr{F}_{y_0}(i\lambda_I)$ as $\lambda_I\rightarrow\infty$. First, we let $(\mathscr{L}_{y_0} - i\lambda_I)^{-1} w_c(y+y_0)^2 = \Phi_R + i\Phi_I$ where $\Phi_R$ and $\Phi_I$ solve
\begin{equation}\label{eq:Phi_real_imaginary_eq}
\mathscr{L}_{y_0}\Phi_R + \lambda_I\Phi_I = w_c(y+y_0)^2,\qquad \mathscr{L}_{y_0}\Phi_I - \lambda_I\Phi_R = 0,
\end{equation}
with the boundary conditions $\Phi_R'(0) = \Phi_I'(0)$ and $\Phi_R,\Phi_L\rightarrow 0$ as $y\rightarrow\infty$. Taking $\lambda_I\gg 1$ and assuming that $y=\mathcal{O}(1)$ we obtain
$$
\Phi_I(y) \sim \frac{1}{\lambda_I}w_c(y+y_0)^2,\qquad \Phi_R(y) \sim \frac{1}{\lambda_I}\mathscr{L}_{y_0}\Phi_I = \frac{1}{\lambda_I^2}\bigl( 2w_c'(y+y_0)^2 + w_c(y+y_0)^2\bigr).
$$
If $y_0>0$ then $\Phi_R'(0)=0$ and $\Phi_I'(0)=0$ are not satisfied and we must therefore consider the boundary layer at $y=0$. Setting $z = \lambda_I^{1/2} y$ we consider the inner expansion $\Phi_R \sim \tilde{\Phi}_R(z)$ and $\Phi_I \sim \tilde{\Phi}_I(z)$ where $\tilde{\Phi}_I$ satisfies
\begin{equation*}
\frac{d^4\tilde{\Phi}_I}{dz^4} + \tilde{\Phi}_I = \frac{1}{\lambda_I} w_c(y_0)^2\quad z > 0;\qquad \frac{d\tilde{\Phi}_I}{dz} = \frac{d^3 \tilde{\Phi}_I}{dz^3} = 0,\quad z=0,
\end{equation*}
and must be matched to the outer, $y=\mathcal{O}(1)$, solution through the far-field behaviour
$$
\tilde{\Phi}_I \sim \frac{1}{\lambda_I}w_c(y_0)^2,\quad \frac{d^2\tilde{\Phi}_I}{dz^2} \sim \frac{1}{\lambda_I^2}\bigl(2w_c'(y+y_0)^2 + w_c(y+y_0)^2\bigr),\qquad z\rightarrow\infty.
$$
It is clear that the leading order solution is $\tilde{\Phi}_I(z) \sim \lambda_I^{-1} w_c(y_0)^2$. The constant behaviour of $\Phi_I$ at the boundary layer therefore does not contribute to the leading order behaviour of the integral
\begin{equation*}
\int_0^\infty w_c(y+y_0)\Phi_I(y) dy \sim \lambda_I^{-1} \int_0^\infty w_c(y+y_0)^3 dy,\qquad \lambda_I\gg 1.
\end{equation*}
Moreover, multiplying the right equation in \eqref{eq:Phi_real_imaginary_eq} by $w_c(y+y_0)$ and integrating we calculate
\begin{align*}
\int_0^\infty w_c(y+y_0)\Phi_R(y) dy & = \frac{1}{\lambda_I}\biggl( w_c'(y_0)\Phi_I(0) + \int_0^\infty \Phi_I\mathscr{L}_{y_0} w_c(y+y_0) dy \biggr)\\
& \sim \frac{1}{\lambda_I^2}\biggl(w_c'(y_0)w_c(y_0)^2 + \int_0^\infty w_c(y+y_0)^4 dy\biggr),
\end{align*}
for $\lambda_I\gg 1$ where we have used $\mathscr{L}_{y_0} w_c(y+y_0) = w_c(y+y_0)^2$. In summary, we have the large $\lambda_I$ asymptotics
\begin{equation}\label{eq:F_imag_limit}
\mathscr{F}_{y_0}(i\lambda_I) \sim \frac{1}{\lambda_I^2}\frac{w_c'(y_0)w_c(y_0)^2 + \int_0^\infty w_c(y+y_0)^4dy}{\int_0^\infty w_c(y+y_0)^2dy} + \frac{i}{\lambda_I}\frac{\int_0^\infty w_c(y+y_0)^3 dy}{\int_0^\infty w_c(y+y_0)^2dy},\qquad \lambda_I\gg 1.
\end{equation}
Note that the real part changes from positive to negative as $y_0$ exceeds $y_0\approx 1.0487$. In Figures \ref{fig:real_F_imag} and \ref{fig:imag_F_imag} we plot the real and imaginary parts of $\mathscr{F}_{y_0}(i\lambda_I)$ respectively for select values of $y_0$. In addition, we have included the large $\lambda_I$ asymptotics which indicate close agreement for moderately large values of $\lambda_I$.

\section{Numerical Support for Stability Conjecture}\label{app:conjecture}

In this appendix we provide numerical support for the conjecture that the shifted NLEP \eqref{eq:rigorous_nlep} has a stable spectrum when $\mu>\mu_c(y_0)$ by numerically calculating the dominant eigenvalue of the NLEP for $0\leq y_{0} \leq 1.5$ and $0\leq \mu \leq 10$. The numerical calculation of the spectrum was performed by truncating the domain $0<y<\infty$ to $0<y<20$ and discretizing it with $600$ uniformly distributed points. Then, we used a finite difference approximation for the second derivatives and a trapezoidal rule discretization for the integral term to approximate the NLEP \eqref{eq:rigorous_nlep} with a discrete matrix eigenvalue problem. We then numerically calculated the dominant eigenvalue of matrix by using the eig function in the Python scipy.linalg library for our numerical computation of the dominant eigenvalue. In Figure \ref{fig:conjecture_1} we plot $\text{Re}\lambda_0$ versus $y_0$ and $\mu$. We observe the real part of the dominant eigenvalue is negative when $\mu$ exceeds the threshold $\mu_c(y_0)$. Additionally, in Figure \ref{fig:conjecture_2} we plot $\Lambda_0-\text{Re}(\lambda_0)$ for the same range of $y_0$ and $\mu$ values. We observe that this difference is non-negative which suggest that $\text{Re}\lambda_0\leq \Lambda_0$.

\begin{figure}[t!]
	\centering
	\begin{subfigure}{0.5\textwidth}
		\centering
		\includegraphics[scale=0.675]{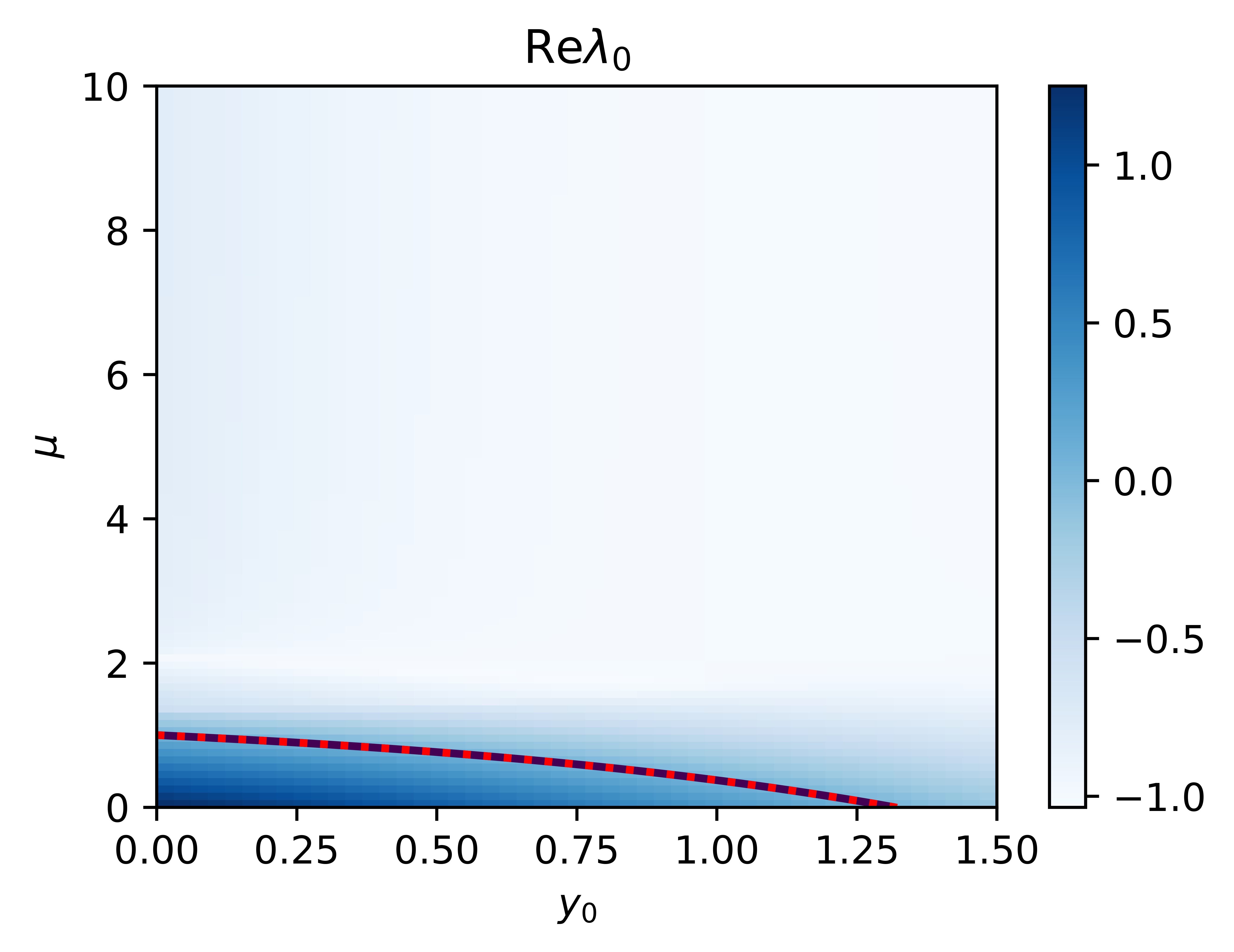}
		\caption{}\label{fig:conjecture_1}
	\end{subfigure}%
	\begin{subfigure}{0.5\textwidth}
		\centering
		\includegraphics[scale=0.675]{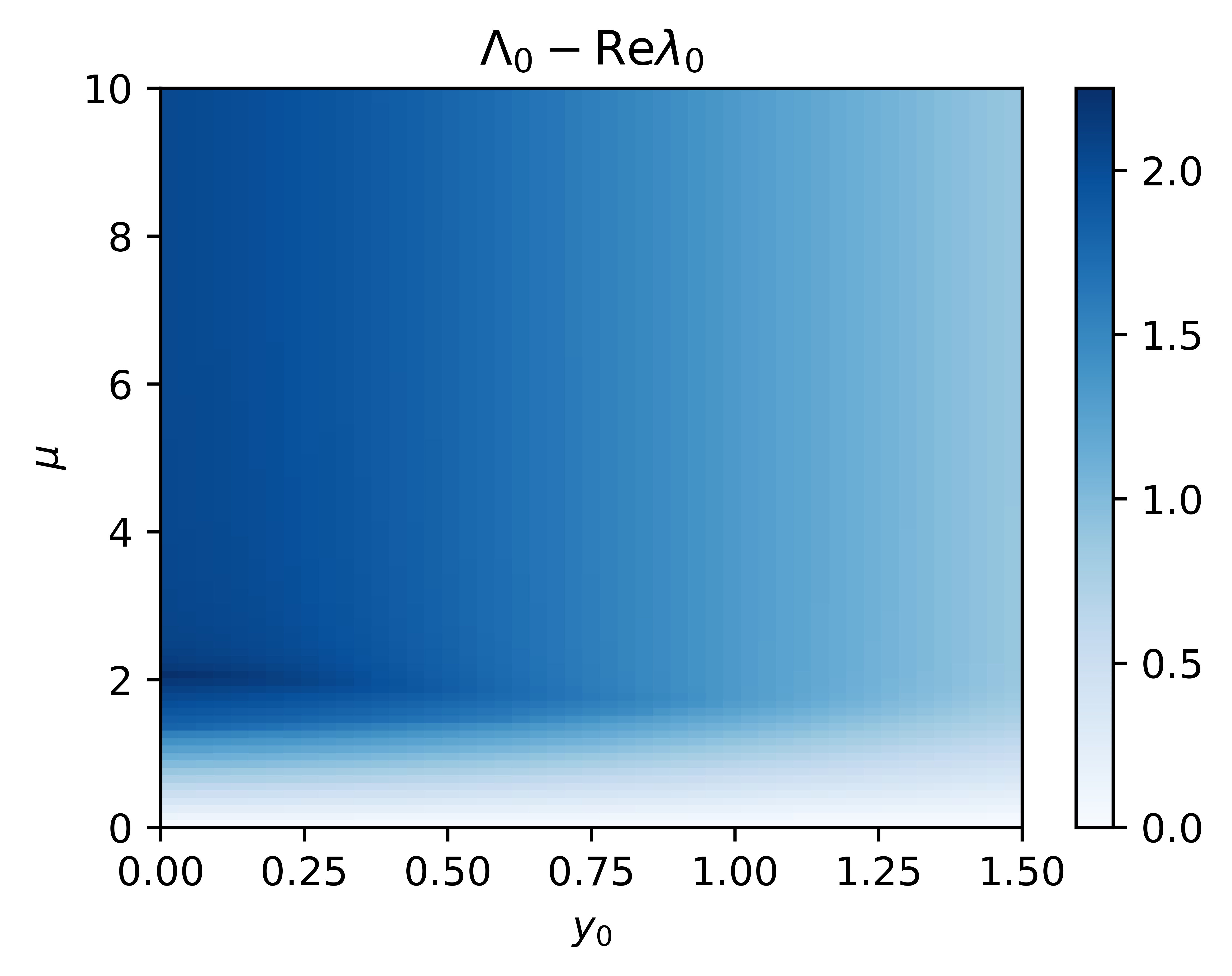}
		\caption{}\label{fig:conjecture_2}
	\end{subfigure}%
	\caption{(A) Plot of the real part of the dominant eigenvalue of the shifted NLEP \eqref{eq:rigorous_nlep} versus shift parameter $y_0$ and multiplier $\mu$. The dotted red line corresponds to the critical threshold $\mu_c$ defined in \eqref{eq:muc_def} and the solid dark line is the zero-contour of $\text{Re}\lambda_0$. (B) Plot of the difference between dominant eigenvalues of $\mathscr{L}_{y_0}$ and the NLEP \eqref{eq:rigorous_nlep}.}\label{fig:conjectures}
\end{figure}

\section{Stability of Asymmetric Two-Boundary Spike Pattern when $A=0$}\label{app:A0_stability}

Previous results on the stability of asymmetric two spike equilibria of \eqref{eq:pde_gm} when $A=0$ have focused exclusively on interior multi-spike solutions \cite{ward_2002_asymmetric}. To compare the $A=B=0$ theory with our results obtained in Examples 2-4 we include here a summary of the stability of an asymmetric two-spike solution where one spike concentrates at $x=0$ and the other concentrates at either $x=1$ or in the interior $0<x<1$. In both cases the NLEP \eqref{eq:nlep} with $A=0$ can be written as
\begin{equation}\label{eq:nlep_appendix}
\mathscr{L}_0\pmb{\phi} - 2 w_c(y)^2 \frac{\int_0^\infty w_c(y) \mathscr{E}\pmb{\phi}dy}{\int_0^\infty w_c(y)^2 dy} = \lambda\pmb{\phi},
\end{equation}
where for two boundary spikes we let
\begin{equation}
\mathscr{E} = \mathscr{E}_{bb} \equiv \begin{pmatrix}  \coth\omega_0\tanh\omega_0l & \text{csch}\omega_0 \tanh\omega_0(1-l) \\ \text{csch}\omega_0\tanh\omega_0l & \coth\omega_0\tanh\omega_0(1-l) \end{pmatrix},
\end{equation}
and in the case of one boundary and one interior spike we let
\begin{equation}
\mathscr{E} = \mathscr{E}_{bi} \equiv \begin{pmatrix}  \coth\omega_0\tanh\omega_0l & 2\text{csch}\omega_0\sinh\omega_0\tfrac{1-l}{2} \\ \text{csch}\omega_0\tanh\omega_0l\cosh\omega_0\tfrac{1-l}{2} & 2\text{csch}\omega_0\cosh\omega_0\tfrac{1+l}{2}\sinh\omega_0\tfrac{1-l}{2} \end{pmatrix}.
\end{equation}
It is then straightforward to verify that $\sigma=1$ is an eigenvalue of both matrices $\mathscr{E}_{bb}$ and $\mathscr{E}_{bi}$. The remaining eigenvalue in each case is then given by the determinant. By diagonalizing $\mathscr{E}$ the NLEP \eqref{eq:nlep_appendix} can therefore be written as \eqref{eq:rigorous_nlep} with $\mu=2$ as well as with $\mu=2\det\mathscr{E}_{bb}$ and $\mu = 2\det\mathscr{E}_{bi}$ for the boundary-boundary and boundary-interior cases respectively. Since the $A=B=0$ stability theory implies that the NLEP \eqref{eq:rigorous_nlep} is stable if and only if $\mu>1$ \cite{wei_1999} we immediately deduce that the $\mu=2$ modes are stable in both the boundary-boundary and boundary-interior cases. To determine the stability of the remaining modes we explicitly calculate
\begin{equation}
\det\mathscr{E}_{bb} = \tanh\omega_0 l\tanh\omega_0(1-l),\qquad \det\mathscr{E}_{bi} = 2\frac{\tanh\omega_0l\sinh\omega_0\tfrac{1-l}{2}\sinh\omega_0\tfrac{1+l}{2}}{\sinh\omega_0}.
\end{equation}
Finally, for the boundary-boundary and boundary-interior cases we solve \eqref{eq:two_bdry_spike_eq} and \eqref{eq:example_4_asy_eq} for $D=D(l)$ respectively and the resulting values of $2\det\mathscr{E}_{bb}$ and $2\det\mathscr{E}_{bi}$ for $0<l<1$ are shown in Figure \ref{fig:app-nlep-mult}. In particular the asymmetric pattern with two boundary spikes is always linearly unstable, while the pattern with one boundary and one interior spike has a region of stability (with respect to the $\mathcal{O}(1)$ eigenvalues). In Figure \ref{fig:app-bifurc} we plot $l=l(D)$ (cf. Figures \ref{fig:BB_B0} and \ref{fig:BI_B0}) for both two-spike configurations, indicating where the pattern is stable (solid line) and unstable (dashed line). 

\begin{figure}[t!]
	\centering
	\begin{subfigure}{0.5\textwidth}
		\centering
		\includegraphics[scale=0.65]{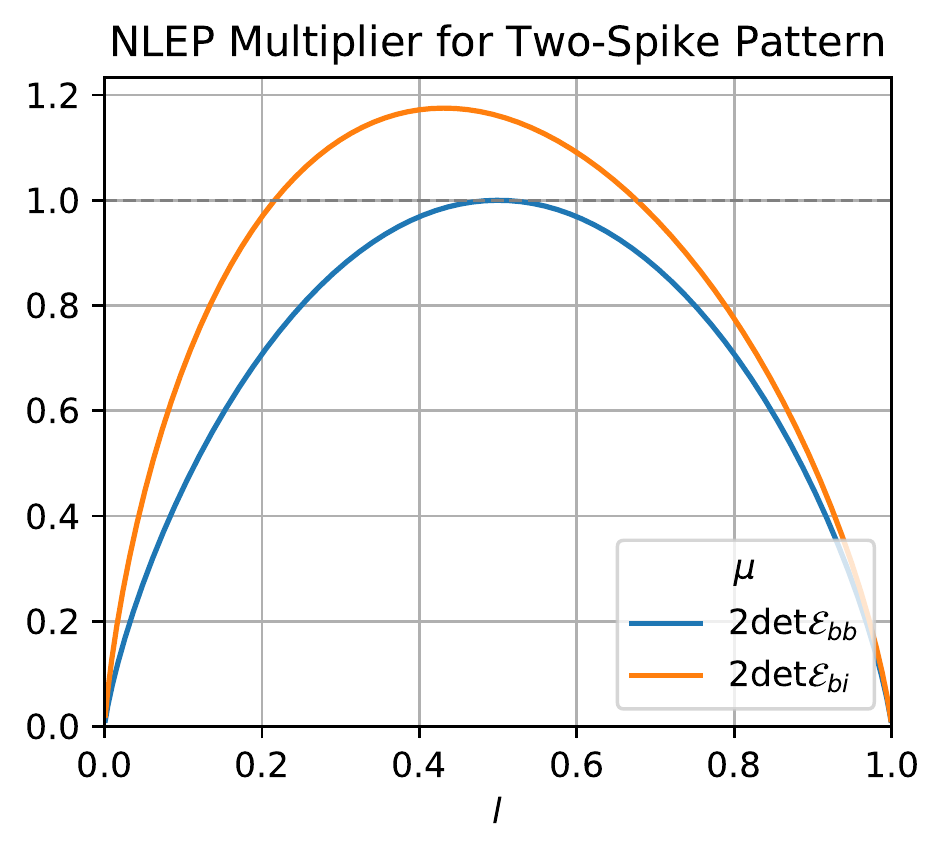}
		\caption{}\label{fig:app-nlep-mult}
	\end{subfigure}%
	\begin{subfigure}{0.5\textwidth}
		\centering
		\includegraphics[scale=0.65]{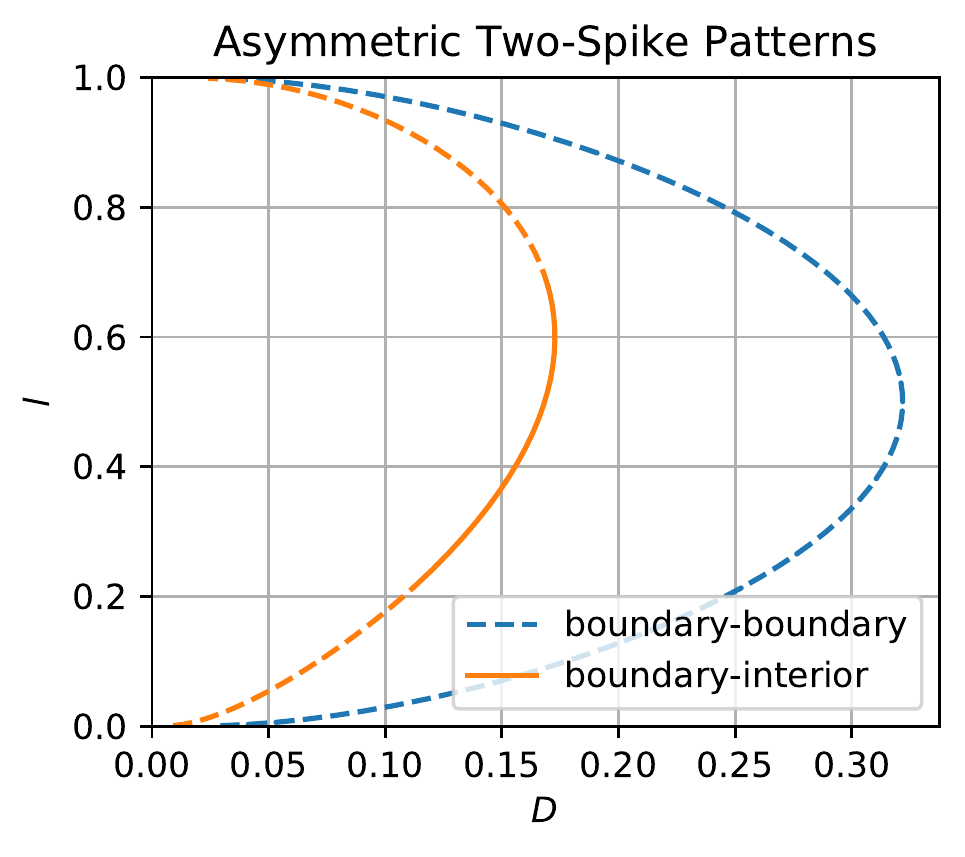}
		\caption{}\label{fig:app-bifurc}
	\end{subfigure}%
	\caption{}\label{fig:app-a0b0}
\end{figure}

\addcontentsline{toc}{section}{References}
\bibliographystyle{abbrv}
\bibliography{bibliography}

\end{document}